\documentclass[11pt]{article}
\usepackage[utf8]{inputenc}

\usepackage{style}

\newcommand{\des}[1]{\mathrm{DES}[#1]}
\DeclareMathOperator{\polylog}{polylog}

\title{Pseudorandomness Properties of Random Reversible Circuits}
 \author{
William Gay${}^*$
\and
 William He${}^*$
 \and
 Nicholas Kocurek${}^*$
 \and 
Ryan O'Donnell\thanks{Computer Science Department, Carnegie Mellon University. \texttt{\{wrhe,odonnell\}@cs.cmu.edu}. Supported in part by ARO grant W911NF2110001. \texttt{\{wgay,nkocurek\}@andrew.cmu.edu}}.}
\date{\small\today}

\begin{document}
\maketitle
\allowdisplaybreaks
\begin{abstract}
    Motivated by practical concerns in cryptography, we study pseudorandomness properties of permutations on $\{0,1\}^n$ computed by random circuits made from reversible $3$-bit gates (permutations on $\{0,1\}^3$). Our main result is that a random circuit of depth $\sqrt{n} \cdot \wt{O}(k^3)$, with each layer consisting of $\Theta(n)$ random gates in a fixed two-dimensional nearest-neighbor architecture, yields approximate $k$-wise independent permutations. 

    Our result can be seen as a particularly simple/practical block cipher construction that gives provable statistical security against attackers with access to $k$~input-output pairs within few rounds. 
    
    The main technical component of our proof consists of two parts:
    \begin{enumerate}
        \item We show that the Markov chain on $k$-tuples of $n$-bit strings induced by a single random $3$-bit one-dimensional nearest-neighbor gate has spectral gap at least $1/n \cdot \wt{O}(k)$. Then we infer that a random circuit with layers of random gates in a fixed \textit{one-dimensional} gate architecture yields approximate $k$-wise independent permutations of $\{0,1\}^n$ in depth $n\cdot \wt{O}(k^2)$
        \item We show that if the $n$ wires are layed out on a \textit{two-dimensional} lattice of bits, then repeatedly alternating applications of approximate $k$-wise independent permutations of $\{0,1\}^{\sqrt n}$ to the rows and columns of the lattice yields an approximate $k$-wise independent permutation of $\{0,1\}^n$ in small depth.
    \end{enumerate}
    Our work improves on the original work of Gowers~\cite{gowers1996almost}, who showed a gap of $1/\poly(n,k)$ for one random gate (with non-neighboring inputs); and, on subsequent work~\cite{hoory2005simple,brodsky2008simple} improving the gap to $\Omega(1/n^2k)$ in the same setting.
\end{abstract}

\newpage

\tableofcontents
\section{Introduction}
Motivated by questions in the analysis of practical cryptosystems (block ciphers), we study pseudorandomness properties of random reversible circuits. That is, we study the indistinguishability of truly random permutations on $\{0,1\}^n$ versus permutations computed by small, randomly chosen reversible circuits. 

Our main results concern the extent to which small random reversible circuits compute \emph{approximate $k$-wise independent} permutations. This corresponds to \emph{statistical security} against adversaries that get $k$ input-output pairs from the permutation. 

\subsection{Cryptographic considerations}\label{sec:crypto}
Block ciphers lie at the core of many practical implementations of cryptography. Constructions such as the Advanced Encryption Standard (AES) and Data Encryption Standard (DES) are commonly implemented in hardware on computers as instantiations of pseudorandom permutations, and much is demanded of both their security and efficiency. The need for simple and efficient hardware implementations is often in tension with the hope that these block ciphers are secure against cryptanalysis.

Short of showing unconditional security of a particular block cipher against general polynomial-time adversaries, one might hope to show that random reversible circuits are secure against certain classes of cryptographic attacks. Examples of known attacks against block ciphers include linear cryptanalysis, differential cryptanalysis, higher-order attacks, algebraic attacks, integral cryptanalysis, and more. See, for example, \cite{liu2021t,liu2023layout} and the references therein.

We study a pseudorandomness property that implies security against many of these known classes of attacks:
\begin{definition}
    A distribution $\mathcal{P}$ on permutations of $\{0,1\}^n$ is said to be \emph{$\epsilon$-approximate $k$-wise independent} if for all distinct $x^{(1)}, \dots, x^{(k)} \in \{0,1\}^n$, the distribution of $(\boldsymbol{g}(x^{(1)}), \dots, \boldsymbol{g}(x^{(k)}))$ for $\boldsymbol{g} \sim \mathcal{P}$ has total variation distance at most $\epsilon$ from the uniform distribution on $k$-tuples of distinct strings from~$\{0,1\}^n$.
\end{definition}
By definition, such distributions are statistically secure (up to advantage~$\epsilon$) against an adversary that gets access to any $k$ input-output pairs, chosen non-adaptively. For example, approximate 2-wise independence implies security against linear and differential attacks, and approximate $k$-wise independence for larger $k$ implies security against higher-order differential attacks. Moreover, Maurer and Pietrzak~\cite{maurer2004composition} have shown that this can easily be upgraded to security against \emph{adaptive} queries by composing two draws from the pseudorandom permutation (the second inverted). 

One might be more interested in the security of $\calP$ against any polynomial-time adversary though. Such adversaries can make any polynomial number of queries to the permutation, so to use the security guarantee of approximate $k$-wise independence to infer computational security, one would need to set $k$ to be superpolynomial. A simple argument shows that such permutations require superpolynomial circuit complexity.

However, all known computationally efficient attacks against block ciphers fail as soon as approximate 4-wise independence holds. That is, we know of no approximate 4-wise independent permutation generated as the composition of many local permutations (gates) that is efficiently distinguishable from a completely random permutation of the set of $n$-bit strings. To highlight this fact, it was conjectured in \cite{hoory2005simple} (Section 6) that any $\exp(-n)$-approximate 4-wise independent permutation which is obtained by composing some number of reversible gates forms a pseudorandom permutation. This provides further motivation for the study of approximate $k$-wise independence. 

In addition to providing security guarantees, much focus is put on designing simple and efficient block ciphers. Recall that a random reversible circuit takes as input some $n$-bit string and computes a permutation by letting a sequence of random gates on 3 bits\footnotemark act on the $n$-bit strings.  In this paper we study a natural class of very efficiently implementable circuits, known as brickwork circuits. A brickwork circuit is formed by organizing the bits (also called wires of the circuit) into a low-dimensional lattice and applying layers of gates acting on nearest-neighbors in this lattice at a time. See \Cref{fig:2D brickwork} for an illustration of this architecture. 

\footnotetext{Note that gates that act on three wires are necessary, since the set of gates acting on two wires at a time can generate only $\F_2$-affine permutations of the hypercube. Such permutations do not exhibit nontrivial types of pseudorandomness.}

The natural efficiency-related question to then ask is how much depth is required for these brickwork circuits to become pseudorandom. In other words, we hope for a class of random reversible circuits that has the following desirable properties:
\begin{itemize}
    \item The gate architecture is fixed and all gates in the circuit act on nearest-neighbor wires in a low-dimensional Euclidean geometry.
    \item The circuits become approximate $k$-wise independent in very low depth.
    \item The circuit yields a straightforward implementation of its inverse.
\end{itemize}
The first item in this list are especially important for implementations in hardware. The second item shows that they can be run for a short amount of time to achieve pseudorandomness, since circuit depth in this setting corresponds to ``wall clock time." Finally, the third item is relevant to efficient decryption. 

In this paper we introduce a construction of a candidate block cipher/pseudorandom permutation from random local circuits that could yield straightforward and efficient implementations in hardware, and our main result shows that this class of random reversible circuits satisfies these desirable properties:
\begin{theorem}\label{thm:2D main}
    For $k \leq 2^{O(\sqrt{n})}$, there is a class of random reversible circuits with a fixed gate architecture on a two-dimensional lattice (\Cref{fig:2D brickwork}) computing permutations that are $2^{-\sqrt{n}k}$-approximate $k$-wise independent at depth $\sqrt{n} \cdot \widetilde{O}(k^3)$.\footnote{Our result holds even when every gate is assumed to be a $\des{2}$ gate. A gate of type $\des{2}$ if (up to ordering the $c$~bits) it is of the form $(x,b) \mapsto (x,b \oplus f(x))$ for some $f : \{0,1\}^{2} \to \{0,1\}$.}
\end{theorem}
\begin{figure}[t]
    \centering
    \includegraphics[width=0.7\linewidth]{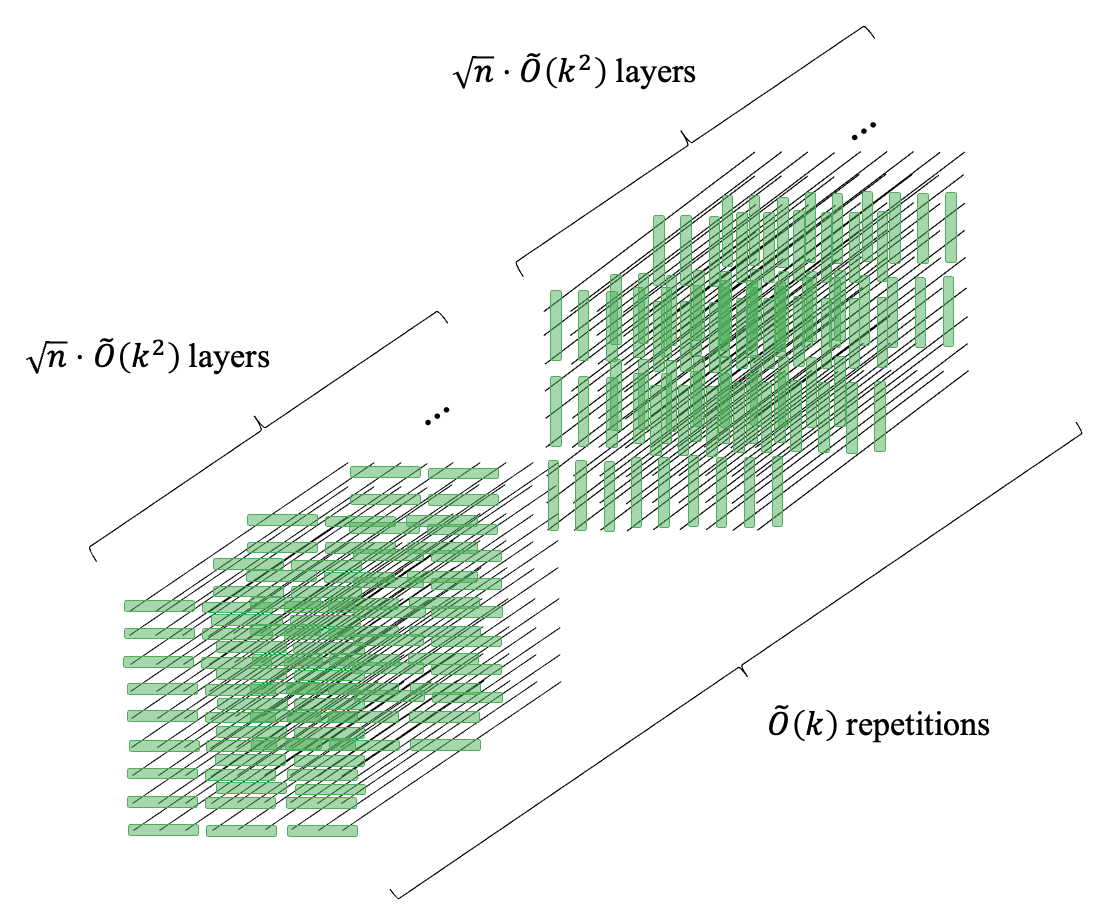}
    \caption{Our two-dimensional brickwork circuit architecture. Each green bar represents a single random 3-bit gate. \Cref{thm:2D main} states that with the stated number of layers of random gates, the permutation becomes $2^{-\sqrt{n}k}$-approximate $k$-wise independent.}
    \label{fig:2D brickwork}
\end{figure}

This result shows that our candidate block cipher becomes secure against a wide class of attacks in very low depth, by setting $k$ to be various small values. If the conjecture from \cite{hoory2005simple} were to hold, then by setting $k=4$, \Cref{thm:2D main} would imply that random brickwork circuits of \emph{sublinear} depth form computationally secure pseudorandom permutations. 

Our construction is a variant of the random reversible circuits introduced by Gowers~\cite{gowers1996almost}. In this initial work, Gowers also analyzed the approximate $k$-wise independence of random reversible circuits. However, we emphasize that this analysis is a departure from the previous literature on approximate $k$-wise independent permutations from random reversible circuits, which has focused on circuits with non-nearest-neighbor gates without fixed architectures. Our main contribution is showing that even with such desirable properties (fixed gate architecture, geometric locality of gates), circuits still exhibit pseudorandomness properties with small size/depth. Note that the AES block cipher also exhibits a similar 2D-lattice-based structure. 

The construction of these structured circuits are partially inspired by the study of \textit{unitary designs} in quantum physics, in particular \cite{brandao2016local,harrow2023approximate}, but the resulting analysis for our case of permutations differs significantly. 

\Cref{thm:2D main} follows from the following two results. First we show that random reversible circuits with fixed \textit{one-dimensional} gate architectures become approximate $k$-wise independent in small depth.

\begin{theorem}\label{thm:1D main}
    A random one-dimensional brickwork circuit (\Cref{fig:nonlocal local brickwork}) of depth $n\cdot\widetilde{O}(k^2)$ computes a permutation that is $2^{-nk}$-approximate $k$-wise independent.
\end{theorem}

Then, we show that instantiating a certain construction of two-dimensional random reversible circuits with the one-dimensional reversible circuits from \Cref{thm:1D main} maintains approximate $k$-wise independence:

\begin{theorem}\label{thm:2D to 1D reduction}
    Suppose there exists a random reversible circuit with a fixed one-dimensional gate architecture that is $2^{-nk}$-approximate $k$-wise independent. Then there exists a class of random reversible circuits with a fixed gate architecture on a two-dimensional lattice computing permutations that are $2^{-\sqrt{n}k}$-approximate $k$-wise independent at depth $\sqrt{n} \cdot \widetilde{O}(k^3)$.
\end{theorem}

\begin{figure}[t]
    \centering 
    \includegraphics[width=\linewidth]{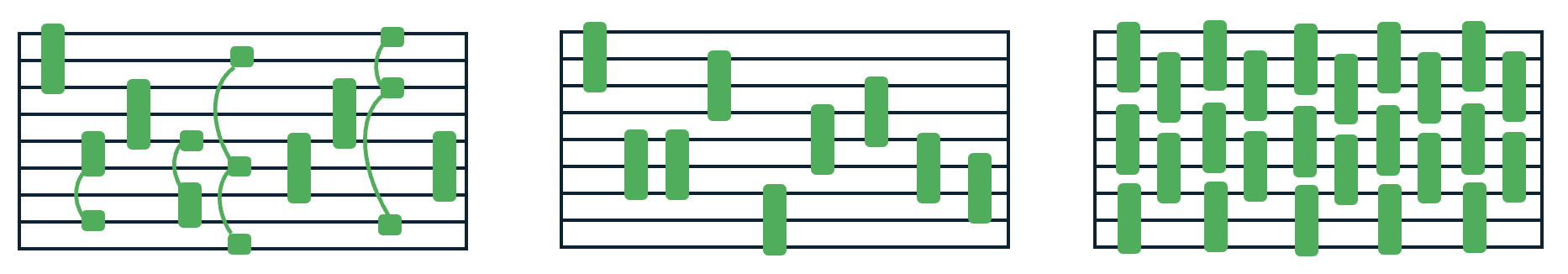}
    \caption{The first circuit is an example of a circuit with generic 3-bit gates. The second is an example of a circuit with 1D nearest-neighbor 3-bit gates. The third is an example of a 1D brickwork circuit with five layers.}
    \label{fig:nonlocal local brickwork}
\end{figure}

We prove \Cref{thm:1D main} with a more general quantitative guarantee, which follows immediately from \Cref{thm:one-layer-brickwork}: that a random brickwork circuit of depth $(nk + \log(1/\eps)) \cdot \wt{O}(k)$ computes an $\epsilon$-approximate $k$-wise independent permutation of $\{0,1\}^n$. Note that \Cref{thm:1D main} follows by setting $\epsilon=2^{-nk}$.

We also obtain a generalization of \Cref{thm:2D main} by generalizing \Cref{thm:2D to 1D reduction} to higher-dimensional lattices:
\begin{theorem}\label{thm:genresult}
    For all $3 \leq D \leq O\pbra{\frac{\log n}{\log \log n}}$ a random reversible circuit with a certain $D$-dimensional gate architecture computes permutations that are $2^{-n^{1/D}}$-approximate $k$-wise independent permutations of $\{0,1\}^n$ with depth $\exp(D) \cdot n^{1/D} \cdot \widetilde{O}(k^3)$, given that $k \log k \leq O(n^{1/3})$ and $n$ is large enough.
\end{theorem}

\paragraph{Optimal dependence on $n$.} Note that the dependencies on $n$ in \Cref{thm:2D main}, \Cref{thm:1D main}, and \Cref{thm:genresult} are all tight for the specific gate architectures mentioned in each respective result. This can be seen by a light-cone argument: even approximate 2-wise independence is not possible if there are two wires that cannot influence each other.

\paragraph{On circuit size.} While the parameter we emphasize in our results is depth, it is interesting to note that our high-dimensional circuit achieves a size vs. $\varepsilon$ tradeoff comparable to that of \cite{gretta2024more}. To illustrate this, we set $k=O(1)$. However, we note that our tradeoffs hold for growing $k$ as well; we make this simplification for ease of exposition. 

In our general \Cref{thm:genresult} we show that (given $D = O(\log n/\log\log n)$) a class of random reversible circuits of size $n^{1+1/D} \cdot \exp(D)$ compute $2^{-\Theta(n^{1/D})}$-approximate $k$-wise independent permutations. The tradeoff in \cite{gretta2024more} is that random reversible circuits of size $\widetilde{O}(n^{1+1/D})$ also compute $2^{-\Theta(n^{1/D})}$-approximate $k$-wise independent permutations. 

While we use purely spectral techniques to obtain our results, \cite{gretta2024more} proceeded by proving log-Sobolev inequalities for random walks associated with random reversible circuits. We believe it is interesting that using spectral techniques, we recover similar mixing time results as those obtained from log-Sobolev inequalities.

\subsection{Related Topics in Pseudorandomness}

\paragraph{Unitary designs.} The definition of approximate $k$-wise independence can be framed in terms of approximating a truly random $2^n$-by-$2^n$ permutation matrix by a pseudorandom one up to $k$th moments. If one replaces permutation matrices with general unitary matrices, then one arrives at the definition of an (approximate) unitary $k$-design. An (approximate) unitary $k$-design is a distribution $\mathcal{P}$ on the unitary group that resembles the Haar random measure on the unitary group up to $k$th moments. 

Among various motivations for constructing (approximate) unitary $k$-designs are derandomization of quantum algorithms, modeling of black holes, and the study of topological order. See, for example, \cite{hayden2007black,scott2008optimizing,brandao2016local,dankert2009exact,huang2020predicting}. Often, the goal is to obtain more efficient implementations of unitary designs using small quantum circuits; this goal is similar to the goal of this paper \cite{brandao2016local,hunter2019unitary,haferkamp2021improved,harrow2023approximate,haferkamp2023efficient,chen2024efficient,metger2024simple,chen2024incompressibility,ma2024construct}.

Recently there has been work on reducing the design of approximate unitary $k$-designs to the design of approximate $k$-wise independent permutations. For example, \cite{metger2024simple,chen2024efficient} provides constructions of approximate unitary $k$-designs from approximate $k$-wise independent permutations in a black-box way. These constructions provide further motivation for the study of efficiently generated $k$-wise independent permutations.

\paragraph{Derandomization.} Approximate $k$-wise independent permutation distributions~$\calP$ have many applications outside of cryptography; derandomization, for example~\cite{mohanty2020explicit}. In such applications, another important parameter is the number of truly random ``seed'' bits needed to generate a draw from~$\calP$. By using techniques such as derandomized squaring, one can generally reduce the seed length to $O(nk)$ for any construction; see~\cite{kaplan2009derandomized}. This is true for the results in our paper, and we don't discuss this angle further. We are generally focused on the circuit complexity of our permutations.

\subsection{Techniques and Comparison with Previous Work}
Recall that our main result \Cref{thm:2D main} is about the approximate $k$-wise independence of brickwork circuits of small depth. We arrive at our result via two steps: proving \Cref{thm:1D main} and proving \Cref{thm:2D to 1D reduction}. We first discuss the proof of \Cref{thm:1D main}.

\subsubsection{Spectral Gaps of Structured Random Reversible Circuits}
We prove \Cref{thm:1D main} by lower bounding the spectral gap of the random walk on the set of $k$-tuples of $n$-bit strings associated with a random one-dimensional brickwork circuits. The only prior work on brickwork or nearest-neighbor gates for implementing permutations is work of \cite{feng2024dynamics}, but their results are quantitatively weaker. 

Prior work in this area, which deals with non-nearest-neighbor random gates, proceeded by bounding the spectral gap of the random walk on $k$-tuples induced by a single random $\des{2}$ gate. 

These earlier works bounded the spectral gap using the canonical paths method (\cite{gowers1996almost,hoory2005simple}) or multicommodity flows and the comparison method (\cite{brodsky2008simple}). We depart from these methods and use techniques from the physics literature concerned with the extent to which random \emph{quantum} circuits are \emph{unitary} $k$-designs. Specifically, for \Cref{thm:one-random-nonlocal} we use the  induction-on-$n$ technique developed in \cite{haferkamp2021improved}, and for our main \Cref{thm:one-random-local} we employ the more sophisticated Nachtergaele method~\cite{nachtergaele1996spectral} as in the work of Brand\~{a}o, Harrow, and Horodecki~\cite{brandao2016local}. In fact, it was posed as a question in \cite{brandao2016local} whether their techniques could be extended to construct approximate $k$-wise independent permutations. We answer this question in the affirmative. Our analyses use Fourier and spectral graph theory methods. In both cases, the inductive technique only begins to work for $n \geq \Theta(\log k)$, and for smaller~$n$ we need to base our argument on~\cite{brodsky2008simple}; in the case of nearest-neighbor gates, this requires a further comparison-method based argument.

One can also interpret our result as the construction of a Cayley graph on $\mathfrak{A}_{2^n}$ with spectral gap $\Omega\pbra{1/n2^n}$ and degree $O(n)$. Previous work on this by Kassabov~\cite{kassabov2007symmetric} constructed constant-degree expanders out of Cayley graphs on $\mathfrak{A}_{2^n}$, but it is not clear if this random walk can be implemented by short circuits. Especially for the cryptographic applications of this work, it is important that our random walks have low circuit complexity. 

In the context of non-nearest-neighbor random gates, the best prior result was the following theorem of Brodsky and Hoory for general (non-nearest-neighbor) reversible architecture:
\begin{theorem} \label{thm:BH}
    (\cite{brodsky2008simple}.)
    Consider the permutation on $\{0,1\}^n$ computed by a reversible circuit of $O(n^3 k^2)$ randomly chosen $\des{2}$-gates (meaning in particular that each gate's $3$ fan-in wires are randomly chosen).
    This is $2^{-O(nk)}$-approximate $k$-wise independent.  More generally, such circuits of size $O(n^2 k) \cdot (nk + \log(1/\eps))$ suffices for $\eps$-approximate $k$-wise independence.
\end{theorem}
There is a simple reduction from general gates to nearest-neighbor gates that incurs factor-$\Omega(n)$ size blowup.  Plugging this into \Cref{thm:BH} would yield $2^{-O(nk)}$-approximate $k$-wise independent permutations formed from nearest-neighbor reversible circuits of size $O(n^4 k^2)$. Besides being non-brickwork, this is worse than our \Cref{thm:1D main} by a factor of about~$n^2$.
We should note that Brodsky and Hoory also prove the following:
\begin{theorem} \label{thm:BH2}
    (\cite{brodsky2008simple}.)
    If $k \leq 2^{n/50}$, then \Cref{thm:BH} also holds for random reversible circuits of size $\wt{O}(n^2 k^2 \log(1/\eps))$.
\end{theorem}
\noindent The improved dependence on $n$ in this theorem, namely $\wt{O}(n^2)$, is good, but one should caution that the dependence on $\log(1/\eps)$ is \emph{multiplicative}.  Thus except in the rather weak case when $\eps \gg 2^{-nk}$, this term introduces a factor of at least $nk$ back into the bound, making it worse than \Cref{thm:BH}.

All of the difficulty in our main result \Cref{thm:1D main} comes from analyzing the \emph{spectral gap} of the natural random walk on $k$-tuples of strings arising from picking \emph{one} random nearest-neighbor gate. Note that by spectral gap we actually mean the gap between the top eigenvalue and the next largest eigenvalue, since the chains we work with often are disconnected. We prove:
\begin{theorem} \label{thm:one-random-local}
    Let $P$ be the transition matrix for the random walk on $\{0,1\}^{nk}$ corresponding to one random nearest-neighbor $\des{2}$ gate. Then $P$ has a spectral gap of $1/(n \cdot \wt{O}(k))$.
\end{theorem}
Given this result, \Cref{thm:1D main} follows almost directly by using the \textit{detectability lemma} from Hamiltonian complexity theory~\cite{aharonov2009detectability}, as in analogous results for unitary designs due to~\cite{brandao2016local}. The intermediate step is again proving the one-step spectral gap lower bound.

\begin{theorem} \label{thm:one-layer-brickwork}
    Let $P$ be the transition matrix for the random walk on $\{0,1\}^{nk}$ corresponding to one layer of brickwork gates. Then $P$ has a spectral gap of $\eta \geq 1/\wt{O}(k)$.
\end{theorem}
Note that \Cref{thm:one-layer-brickwork} directly implies \Cref{thm:1D main} by the standard Markov chain mixing time bound using spectral gaps.

We also prove an analogous result to \Cref{thm:one-random-local} in the case of non-nearest-neighbor gates:
\begin{theorem} \label{thm:one-random-nonlocal}
    Let $P$ be the transition matrix for the random walk on $\{0,1\}^{nk}$ corresponding to one random $\des{2}$ gate. Then $P$ has a spectral gap of $\Omega(1/(n k\cdot \log k))$.
\end{theorem}
Although it may look like this result is conceptually dominated by \Cref{thm:one-random-local}, we include it as it has improved $\log k$ factors, and its proof reveals some of the ideas we use in our proof of \Cref{thm:one-random-local}.

\paragraph{Proof ideas.} The idea that underlies the proofs of all of these theorems is to compare the random walks induced by applying random local permutations to $k$-tuples of strings to the random walks induced by completely resampling subsets of bits in every element of a $k$-tuple of strings. The difference between these two operations is essentially the difference between sampling with replacement and sampling without replacement. To illustrate, given a tuple $(x^1,\dots,x^k)$ of strings, let $\bm{F}_S(x^1,\dots,x^k)=(\by^1,\dots,\by^k)$ be a random tuple of strings where for all $a\not\in S$ we have $\by^i_a = x^i_a$, but where each $\by^i_a$ for $a\in S$ is an independent completely random element of $\{\pm1\}$. From this, it is clear that if $S\cup T=[n]$ then $\bm{F}_T\circ\bm{F}_S(x^1,\dots,x^k)$ is a completely random tuple of strings. 

Much of the work in our proofs is showing that in fact, this intuition still approximately holds when $\bm{F}_S$ is replaced with a random local permutation $\bm{\pi}_S$ that only acts on bits in $S$, applied to each element of the tuple, when $S$ is a large enough set. Note that if $(x^1,\dots,x^k)$ is a tuple of strings, then to sample $(\by^1,\dots,\by^k)=(\bm{\pi}_S(x^1),\dots,\bm{\pi}_S(x^k))$ one can equivalently carry out the following process. Consider the substrings $(x^1|_S,\dots,x^k|_S)$. Then, for every $i\in[k]$, resample $\by^i|_S$ from $\{\pm1\}^S$ at random from the remaining available strings (noting that the $\by^i|_S$ must be distinct for distinct $x^i|_S$. If $S$ is a large enough set, and the strings $x^1|_S,\dots,x^k|_S$ were all distinct, then this process strongly resembles the process of drawing each $\by^i|_S$ independently. In this sense, for such tuples $(x^1,\dots,x^k)$, we have the following resemblance in distributions:
\begin{align}\label{eq:approximation of permutations}
    (\bm{\pi}_S(x^1),\dots,\bm{\pi}_S(x^k)) \approx \bm{F}_S(x^1,\dots,x^k).
\end{align}
On the other hand, if the $x^1|_S,\dots,x^k|_S$ were not distinct, then we use a simple argument based on escape probabilities to show that such tuples also exhibit good expansion.

Equipped with this comparison, we are able to prove statements of the form 
\begin{align*}
    (\bm{\pi}_T\bm{\pi}_S(x^1),\dots,\bm{\pi}_T\bm{\pi}_S(x^k)) \approx \bm{F}_T\bm{F}_S(x^1,\dots,x^k)
\end{align*}
when $S$ and $T$ satisfy certain properties, such as having sufficiently large overlap. When $S\cup T=[n]$ it turns out that the distribution on the right-hand side is equal to a completely random $k$-tuple of $n$-bit strings. Via a lemma of Nachtergaele~\cite{nachtergaele1996spectral} from quantum physics, a linear-algebraic formalization of this approximation results in a large spectral gap for the spectral gaps associated with geometrically-local random reversible circuits.

\subsubsection{Two-Dimensional Construction}
To obtain the reduction \Cref{thm:2D to 1D reduction} we use a technique inspired by \cite{harrow2023approximate}, which obtained a similar reduction in setting of constructing unitary $k$-designs using random quantum circuits. However, there are some complications that arise when working with permutations rather than general unitaries, so our approach departs significantly from that of \cite{harrow2023approximate}.

Our construction, however, is essentially identical to that of \cite{harrow2023approximate}. Let $\mcB$ be a distribution on circuits computing an approximately $k$-wise independent permutation of $\{0,1\}^{\sqrt{n}}$. Let $X \in \{\pm1\}^{\sqrt{n}\times \sqrt{n}}$ be input to our circuit in the form of a two-dimensional lattice, a grid. 
\begin{enumerate}
    \item For each row in $X$, we sample independent circuits from $\mcB$ and apply in parallel.
    \item For each column in $X$, we sample independent circuits from $\mcB$ and apply in parallel.
    \item We repeat steps and 1 and 2 a total of $O(k\log k)$ times, and finally step 1 exactly once more.
\end{enumerate}

What we get as a result is a circuit of $2t+1$ ``layers'', with each layer consisting of $\sqrt{n}$ parallel circuits from the family $\mcB$ all in one of two directions in our lattice. If the circuits in $\mcB$ have depth $d$ then the depth of our circuit is $O(k\log k) \cdot d$.

\begin{figure}[h]
    \centering
    \begin{tikzpicture}
        \def \width {0.33};
        \foreach \y in {0, 1, 2, 3} {
            \filldraw[fill=white, draw=green!50!black] (0,\y-\width) -- (3, \y-\width) arc[start angle=-90, end angle=90, radius=\width] -- (0, \y+\width) arc[start angle=90, end angle=270, radius=\width] -- cycle;
        }
    
        \foreach \x in {0, 1, 2, 3} {
            \foreach \y in {0, 1, 2, 3} {
                \fill (\x,\y) circle (2pt);
            }
        }

        \foreach \x in {5, 6, 7, 8} {
            \filldraw[fill=white, draw=green!50!black] (\x-\width, 0) -- (\x-\width, 3) arc[start angle=180, end angle=0, radius=\width] -- (\x+\width, 0) arc[start angle=0, end angle=-180, radius=\width] -- cycle;
        }
    
        \foreach \x in {5, 6, 7, 8} {
            \foreach \y in {0, 1, 2, 3} {
                \fill (\x,\y) circle (2pt);
            }
        }
    \end{tikzpicture}
    \vspace{10px}
    \caption{Step 1 applies parallel circuits from $\mcB$ to the rows, while Step 2 applies to the columns. Our circuit alternates between layers of the two.}
    \label{fig:gridcircuit}
\end{figure}
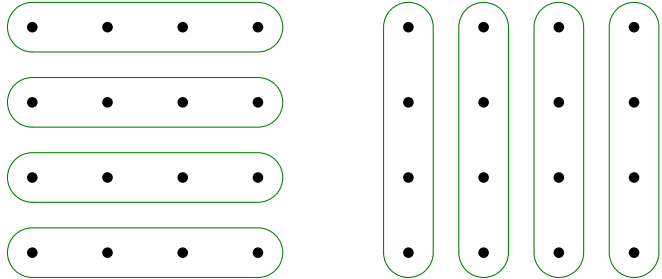
We show that the circuit construction above computes permutations that are $\varepsilon$-approximate $k$-wise independent.

\paragraph{Proof ideas.} As in the 1D case, we reduce analyzing $k$-wise independence of our random reversible circuits to the analyzing the natural induced Markov chain on $\{0,1\}^{nk}$. The induced circuit distribution being approximately uniform then corresponds to mixing in the Markov chain which can be accomplished via establishing a spectral gap or a log-Sobolev inequality.

Our proof does not rely on log-Sobolev inequalities, but rather proceeds entirely through spectral arguments about the random walk operators. However, for sublinear-in-$n$ mixing time (corresponding to circuit depth), it does not suffice to prove a simple spectral gap for our random walk. The reason for this is that naively bounding the mixing time using the spectral gap immediately results in a mixing time at least $\Theta(\log(2^{nk}))=\Theta(nk)$. Thus, we need to proceed more carefully by expressing the total variation distance (which is an $\ell_1$ distance) of the distributions induced by running the Markov chain for some number of steps more directly in terms of the spectral properties of the transition matrices. 

In particular, we derive approximations that resemble those of the type in \Cref{eq:approximation of permutations}, but now the quality of the approximation depends on how many substring collisions the tuple $(x^1,\dots,x^k)$ has. To formalize this quantitatively, we assign each tuple $(x^1,\dots,x^k)$ to a certain group of tuples based on the collision pattern it experiences. By bounding the contribution to the total variation distance of each group of tuples and analyzing the transition probabilities between these groups, we can directly analyze the mixing time without directly exhibiting a functional inequality for the corresponding Markov chain.

\subsection{Comparison with Subsequent Work}\label{sec:subsequent work}
Since the release of the first version of this work, there have been a number of subsequent works on approximate $k$-wise independent permutations and unitary $k$-designs. Of these, perhaps the most relevant to this work are the recent results in \cite{gretta2024more} and \cite{chen2024incompressibility}:
\begin{theorem}[\cite{gretta2024more}, Theorem 2]\label{thm:gretta}
    For any $n$ and $k\leq 2^{n/50}$, a random reversible circuit with $\wt{O}(nk\cdot \log(1/\varepsilon))$ $\des{2}$ gates computes an $\epsilon$-approximate $k$-wise independent permutation.
\end{theorem}

\begin{theorem}[\cite{chen2024incompressibility}, Corollary 1.3]\label{thm:chen}
    For any $n\geq4$ and $k\leq \Theta(2^{n/6.1})$, a random reversible circuit with $\wt{O}(n(nk+\log(1/\varepsilon)))$ $\des{2}$ gates computes an $\epsilon$-approximate $k$-wise independent permutation.
\end{theorem}

It is worth noting the ways in which our results, \Cref{thm:gretta}, and \Cref{thm:chen} supersede each other and are incomparable. We first note that the size bound of $\wt{O}(nk(nk+\log(1/\epsilon))$ implied by \Cref{thm:one-random-nonlocal} is immediately superseded by both \Cref{thm:gretta} and \Cref{thm:chen}.\footnote{In this discussion we ignore factors logarithmic in $n$ and $k$.}

However, the main results of this paper on random reversible circuits with nearest-neighbor gates (\Cref{thm:one-random-local}) and random brickwork circuits (\Cref{thm:1D main}) are still the state-of-the-art. Subsequent works \cite{gretta2024more} and \cite{chen2024incompressibility} say nothing more about reversible circuits with such structured gates than what can be deduced by an application of the comparison method as in \Cref{lem:initial spectral gap local random gates large k}. Such an application immediately incurs a $\mathrm{poly}(n)$ blow-up in the size of the circuits. 

\cite{chen2024incompressibility} claims that it is likely that their techniques can be adapted to show a $\wt{O}(1/n)$ bound on the spectral gap of the random walk induced by one layer of a random brickwork circuit. On the other hand, our spectral gap is $\wt{O}(1/k)$. Thus, if such a result were to be proved, it would be an improvement only if $k$ exceeds $n$. We emphasize that for many purposes one might be interested in the case where $k$ is a constant and $n$ grows. See, for example, the discussion in \Cref{sec:crypto} on a conjecture of \cite{hoory2005simple}.

\subsection{Organization}
The proof of \Cref{thm:1D main} spans \Cref{sec:spectral gaps} and \Cref{sec:nachtergaele hypothesis}. In between, we give a brief overview of our techniques in \Cref{sec:overview} and then prove the separate but related \Cref{thm:one-random-nonlocal} in \Cref{sec:small k}.

The proof of \Cref{thm:2D to 1D reduction} is contained in \Cref{sec:proof of 2D to 1D reduction} and \Cref{sec:spectralproof}, and we prove the more general result about higher-dimensional lattices in \Cref{sec:generallattices}.

\section{Proof Overview}\label{sec:overview}

\subsection{One-Dimensional Circuits}
Recall that the main technical content behind the proof of \Cref{thm:1D main} is to exhibit a spectral gap for the random walk matrix induced by applying a single brickwork layer of fixed random gates. The first step to exhibit this spectral gap is to show that a \textit{single} random geometrically local gate also exhibits a large spectral gap, and the result about a single brickwork layer follows by the detectability lemma of Aharonov et al.~\cite{aharonov2009detectability}. 

To bound the spectral gap induced by a single random geometrically local gate, we note that the corresponding Laplacian is of the form 
\begin{align*}
    L &= \Ex_{\ba\in\{1,\dots,n-2\}}\sbra{L_{\ba}},
\end{align*}
where each $L_{\ba}$ corresponds to the application of a single random gate to the bits $\ba,\ba+1,\ba+2$. This type of operator can also be viewed as the Hamiltonian of a physical system consisting of $n$ particles on a one-dimensional lattice with nearest-neighbor interactions. To analyze the spectral gap of this operator, we use tools from quantum many-body physics.

In particular, a result of Nachtergaele~\cite{nachtergaele1996spectral} shows how to reduce exhibiting a spectral gap for this $n$-particle system to exhibiting a spectral gap for a smaller system. Previous results allow us to bound the spectral gap for this smaller system, but the main technical content in the proof of \Cref{thm:1D main} is to show that the hypotheses of this result of \cite{nachtergaele1996spectral} are satisfied by the operators we are interested in.

As alluded to before, showing that these hypotheses are satisfied involves proving statements such as 
\begin{align*}
    (\bm{\pi}_T\bm{\pi}_S(x^1),\dots,\bm{\pi}_T\bm{\pi}_S(x^k)) \approx \bm{F}_T\bm{F}_S(x^1,\dots,x^k).
\end{align*}
Recall that to show comparisons of the above form, the plan is to compare the processes of sampling with replacement and sampling without replacement in the setting where not too many elements are being resampled, relative to the total number of elements. To see why this comparison is relevant, note that applying a random permutation to each element in a $k$-tuple of strings is almost equivalent to sampling $k$ elements from a set of strings with replacement. 

Combinatorially, directly analyzing probabilities allows us to bound the distance between resampling with replacement and resampling without replacement. To relate this combinatorial bound to a linear-algebraic quantity, we use \Cref{lem:TV distance bound}.

An issue arises when we try to apply this comparison: the comparison does not hold on a small subset of $k$-tuples of $n$-bit strings. However, we show that this set exhibits good expansion already, or that the probability of a random walk escaping from this region is already large. Formally, \Cref{lem:escape probs} relates the escape probabilities from a region of the Markov chain to the quadratic form induced by its transition operator.

\subsection{Two-Dimensional Circuits}

The proof of \Cref{thm:2D main} proceeds by instantiating the lattice construction in \Cref{thm:2D to 1D reduction} with the one-dimensional construction in \Cref{thm:1D main}. We now overview the proof of \Cref{thm:2D to 1D reduction}.

Recall that in the 1D case we related our circuit construction to an induced Markov chain on the set of $k$-tuples of $n$-bit strings and exhibited a spectral gap for this Markov chain to deduce fast mixing. We will start by showing why this argument does not suffice for a sublinear-in-$n$ dependence on $n$ in the mixing time.

Recall that the standard spectral argument states that given spectral gap $\lambda$ for a reversible Markov chain on state space $S$ the time to mix is bounded above by
\begin{align*}
    \frac{1}{\lambda}\cdot(\log{\abs{S}} + \log(1/\varepsilon)).
\end{align*}
However, note that the state space $S$ in our context consists of the set of $k$-tuples of distinct $n$-bit strings, and therefore $\log(\abs{|S|}) \approx nk$. Therefore, since $\lambda$ is always bounded above by 1, the best mixing time one could hope for via this naive analysis would already have a linear dependence in $n$.

Therefore, we need to develop a more refined argument, which deals more directly with total variation distance. We now describe this approach. Recall that our 2D circuit construction involves alternating layers of row and column-oriented 1D circuits. As such, the Markov operator $T_P$ for the circuit can be expressed as $T_P = T_{P_R}(T_{P_C}T_{P_R})^t$ where $T_{P_R}$ and $T_{P_C}$ are the induced operators corresponding to individual row and column-oriented layers. 

We can now express the TV distance between relevant distributions as follows. Let $X \in \{\pm 1\}^{nk}$ be a $k$-tuple of unique $n$-bit strings (we will call the set of such $k$-tuples $\sD$), which corresponds to a single state in our Markov chain. Then the TV distance between the uniform distribution on $\sD$ and the output of a random circuit from our construction on $X$ can be expressed as
\begin{align*}
    \sum_{Y \in \sD} \abs{\Pr[X \to_{T_P} Y] - \frac{1}{\abs{\sD}}} = \sum_{Y \in \sD} \abs{\langle e_X, (T_P - T_G) e_Y \rangle}.
\end{align*}
where $\Pr[X \to_{T_P} Y]$ specifies the probability that our sampled circuit outputs $Y$ on the inputs in $X$ and $e_X$ represents the point mass on $X$ vector. This linear algebraic view turns out to be useful in providing bounds, since it allows us to use spectral techniques.

First, we introduce as an intermediate for analysis a \textit{fully idealized} row-column operator $T_{G_R}(T_{G_C}T_{G_R})^t$ which is analogous to our circuit operator $T_{P_R}(T_{P_C}T_{P_R})$ but with each 1D circuit replaced by a fully random permutation. By the approximate $k$-wise independence of the row and column permutations, replacing $T_P$ by this operator does not change the TV distance much.

We then look at the quantity
\begin{align*}
    \sum_{Y\in \sD}\abs{\langle e_X,(T_{G_R}\pbra{T_{G_C}T_{G_R}}^t-T_G)e_Y\rangle }.
\end{align*}
Each term in this sum can be bound using spectral techniques. Here, we reuse the idea of providing an orthogonal decomposition of the space $\R^{\{\pm1\}^{nk}}$ by partitioning $\{\pm1\}^{nk}$ and doing casework on how these random permutations act on particular tuples in $\{\pm1\}^{nk}$. For example, if a $k$-tuple of grids $X$ has that all rows across the tuple are distinct, then the action of $G_R$ on $X$ is actually very close to the action of $G$ already. To see this, observe that $G_R$ applies a uniform permutation on each row, which can be seen as sampling distinct rows across the tuple. Conversely, $G$ samples distinct grids. Compare this to the operators $H_R$ and $H$, which do the same sampling completely uniformly. While $G_R$ and $G$ are different, $H_R$ and $H$ are exactly the same. In the regime $k \ll \sqrt{n}$, the classic birthday bound tells us that these processes then look very similar.

These regions where $G_R$ looks like $G$ end up being quite large. However, there exist small regions of the graph on which $G_R$ and $G$ act very differently, and indeed it is this fact that causes the operator norm of $T_{G_R} - T_G$ to be large. For example, it can be the case that $X$ is such that all but one row is completely uniform on all elements in a $k$-tuple of $\sqrt{n}$-by-$\sqrt{n}$ grids. Then $G_R X$ must have the same property.

However, in such regions, $G_C$ must then act somewhat similarly to a completely random permutation $G$. More specifically, we will be able to show that by applying $G_C$ in between applications of $G_R$, we are able to ``escape'' the bad regions where $G_R$ does not look like $G$ and show that the end result operator is comparable. We end up with a spectral bound along the lines of:
\begin{equation*}
    \norm{T_{G_C}T_{G_R} - T_G}_{2} \ll \frac{1}{2^{\sqrt{n}}}.
\end{equation*}
Powering (which corresponds to repeating the construction sequentially) allows us to improve this bound exponentially. There is one slight problem that arises here: black box converting this to the statement on the TV distance as above suffers from a blow up on the order $2^{nk}$, similar to the reason we cannot apply a naive spectral norm bound to mixing time argument in the first place.

We get around this by observing that in our setting the Markov chain is actually ``warm-started'' by the first application of $G_R$, in that, $\norm{T_{G_R}e_X}_{2}$ is small already. This does not quite work as is: the bad regions of $\{\pm1\}^{nk}$ still have that $\norm{T_{G_R}e_X}_{2}$ is too large. We supplement by showing that in this case $G_C$ helps us escape the bad region with good probability. This argument breaks in our favor: viewed this way the warm-start brings the blow up from $2^{nk}$ to $2^{\sqrt{n} \cdot \widetilde{O}(k)}$.

See \Cref{sec:proof of 2D to 1D reduction} for the full proof of \Cref{thm:2D to 1D reduction}.

\section{Preliminaries}

\subsection{Some Linear Operators}
\paragraph{Notation.} We switch from $\{0,1\}$ notation to $\{\pm 1\}$ notation to facilitate later Fourier analysis.
We will regard elements of $\{\pm1\}^{nk}$ as tuples $(X^1,\dots,X^k)$, where each $X^i\in \{\pm1\}^n$. If $X^i \in \{\pm1\}^n$, let $X^i_S$ denote the vector $X^i$ restricted to indices in $S$, so that if $S=\{a_1,\dots,a_{|S|}\}$ with $a_1<\dots<a_{|S|}$ then $X^i_S=(X^i_{a_1},\dots,X^i_{a_{|S|}})$. For two vectors $x,y\in \{\pm1\}^n$, define $\Delta(x,y)=\{a\in[n]: x_a\neq y_a\}$. Unless otherwise specified, $\log$ is $\log_2$. For positive integers $a\leq b$ denote $[a,b]=\{a,a+1,\cdots, b-1,b\}$ and let $[a]=[1,a]$. Let $\mathbf{1}$ be the all 1s vector, with length determined by context. When $L$ is a self-adjoint operator let $\lambda_2(L)$ be its second-smallest distinct eigenvalue. If $g\in\mathfrak{S}_{\{0,1\}^{\gamma}}$ is a permutation and $S\subseteq [n]$ has $|S|=\gamma$ then define $g^S\in\mathfrak{S}_{\{0,1\}^n}$ by setting $(g^Sx)_S=g(x_S)$, where coordinates are interpreted to be in lexicographic order, and $(g^Sx)_{[n]\setminus S}=x_{[n]\setminus S}$.

When $\mathcal U$ is a set, we let $\R^{\mathcal U}$ denote the vector space of all functions $\mathcal U\to \R$. We equip this space with the inner product $\left\langle \cdot,\cdot\right\rangle$ given by
\begin{align*}
    \left\langle f,g\right\rangle =& \sum_{x\in \mathcal U}f(x)g(x)
\end{align*}
for $f,g:\mathcal U\to \R$. Also define the norm by $\norm{f}_2=\sqrt{\left\langle f,f\right\rangle}$.

\paragraph{Random walk operators.} Suppose $R$ is a random walk operator corresponding to drawing the next state from the distribution $\mcD_X$, where $X$ is the current state. For $S\subseteq \mathcal{U}$ we will often write $\Pr_{\bm{Y}\sim \mcD_x}[\bm{Y}\in S]$ as $\Pr[X \to_{T_\mcD} S]$. When $S=\{Y\}$ is a singleton we use write $\Pr[X \to_{T_\mcD} S]=\Pr[X \to_{T_\mcD} Y]$.
\begin{fact}\label[fact]{fact:onesteptransition}
    For $X, Y \in \{\pm1\}^{nk}$, $\ip{e_X}{Re_Y} = \Pr[X \to_{R} Y]$.
\end{fact}

This above fact is especially useful as it can be extended across sequential distributions:
\begin{fact}\label[fact]{fact:twosteptransition}
    For $X, Y \in \{\pm1\}^{nk}$, $\ip{e_X}{R_1R_2e_Y} = \Pr[X \to_{R_2R_1} Y]$.
\end{fact}
It is worth pointing out the order of the operators on the bottom is written interpreting $R_1$ and $R_2$ as matrices. The order of application effectively flips under the composition here from how it is written above, so it is important to keep track of orientation.

Throughout the proof of \Cref{thm:1D main}, we use the following notation:
\begin{definition}\label{def:the random walk operator}
    Let $n$ be fixed. Define the distribution $\mathcal{D}^{m,S,k}_X$ to be the law of $\bm{Y}$, where to draw $\bm{Y}$ first one draws $\bm{\sigma'}\in \mathfrak{S}_{\{\pm1\}^{|S|}}$ and then defines $\bm{Y} = \pbra{(\bm{\sigma'}x^1_S,x^1_{\overline S}),\dots,(\bm{\sigma'}x^k_S,x^k_{\overline S})}$. The following is the corresponding random walk operator.
    \begin{align*}
        (R_{n,S,k}f)(X)=& \Ex_{\mathbf{Y} \sim\mathcal{D}^{m,S,k}_{X}}\sbra{f(\mathbf{Y})}.
    \end{align*}
\end{definition}

\begin{fact}\label{fact:self-adjoint}
    For any $n,S,k$ we have that $R_{n,S,k}$ is self-adjoint.
\end{fact}

\begin{fact}\label{fact:uniform is stationary}
    For any $S\subseteq[n]$ we have $R_{n,S,k}\mathbf{1}=\mathbf{1}$.
\end{fact}

In some contexts it is easier to work with the Laplacians of these operators:
\begin{definition}
    For any operator $R$ on a vector space $V$ let $L(R)$ denote its Laplacian $\Id_V-R$. 
\end{definition}

\subsection{Fourier Analysis of Boolean Functions}
\begin{definition}
    Let $S_1,\dots,S_k\subseteq[n]$. Then define the function $\chi_{S_1,\dots,S_k}$ by defining for $X\in\{\pm1\}^{nk}$,
    \begin{align*}
        \chi_{S_1,\dots,S_k}(X)=&\prod_{i\in [k]}\prod_{a\in S_i}X^i_a.
    \end{align*}
\end{definition}

\begin{fact}[\cite{ODonnell2014}]\label{fact:fourier characters}
    The functions $\chi_{S_1,\dots,S_k}$ for $S_1,\dots,S_k\subseteq[n]$ form an orthogonal basis of $\R^{\{\pm1\}^{nk}}$.
\end{fact}

\subsection{Comparison Bounds}
Recall that our main idea is to approximate one random walk by a simpler random walk. The following \Cref{lem:TV distance bound} formalizes this linear algebraically.
\begin{lemma}\label{lem:TV distance bound}
    Suppose $A$ and $B$ are self-adjoint random walk matrices on a domain $\mathcal{U}$. Let $\mathrm{Supp}\subseteq\mathcal{U}$. Then for any $f,g:\mathcal{U}\to\R$ with $\mathrm{Supp}(f)=\mathrm{Supp}(g)=\mathrm{Supp}$ we have
    \begin{align*}
        \abs{\left\langle f, (A-B)g\right\rangle}
        &\leq \sqrt{\sum_{X\in \mathrm{Supp}}f(X)^2\sum_{Y\in \mathrm{Supp}}\abs{\Pr\sbra{X\to_A Y}-\Pr\sbra{X\to_B Y}}}
         \\
         &\qquad\cdot\sqrt{\sum_{X\in \mathrm{Supp}}g(X)^2\sum_{Y\in \mathrm{Supp}}\abs{\Pr\sbra{X\to_A Y}-\Pr\sbra{X\to_B Y}}}.
    \end{align*}    
    In particular, when $f=g$ we get
    \begin{align*}
        \abs{\left\langle f, (A-B)f\right\rangle}\leq \sum_{X\in \mathrm{Supp}(f)}f(X)^2\sum_{Y\in \mathrm{Supp}(f)}\abs{\Pr\sbra{X\to_A Y}-\Pr\sbra{X\to_B Y}},
    \end{align*}
\end{lemma}
\begin{proof}
    We directly compute
\begin{align*}
        &\abs{\left\langle f,(A-B)g\right\rangle}\\
        =&\abs{\sum_{X\in \mcU}f(X)\pbra{(Ag)(X)-(Bg)(X)}}\\
        =&\abs{\sum_{X\in \mcU}{f(X)\sum_{Y\in \calU}g(Y)\pbra{\Pr\sbra{X\to_A Y}-\Pr\sbra{X\to_B Y}}}}\\
        \leq & \sum_{X,Y\in \mathrm{Supp}}\abs{f({X})}\abs{g( Y)}\abs{\Pr\sbra{X\to_A Y}-\Pr\sbra{X\to_B Y}}\tag{Support of $f$, $g$}\\
        \leq & \sqrt{\sum_{X,Y\in \mathrm{Supp}}f(X)^2\abs{\Pr\sbra{X\to_A Y}-\Pr\sbra{X\to_B Y}}}\sqrt{\sum_{X,Y\in \mathrm{Supp}}g(Y)^2\abs{\Pr\sbra{X\to_A Y}-\Pr\sbra{X\to_B Y}}}\\
        = & \sqrt{\sum_{X,Y\in \mathrm{Supp}}f(X)^2\abs{\Pr\sbra{X\to_A Y}-\Pr\sbra{X\to_B Y}}}\sqrt{\sum_{X,Y\in \mathrm{Supp}}g(Y)^2\abs{\Pr[Y \to_A X]-\Pr[Y \to_B X]}}\tag{Self-adjointness of $A$ and $B$}\\
        = & \sqrt{\sum_{X\in \mathrm{Supp}}f(X)^2\sum_{Y\in \mathrm{Supp}}\abs{\Pr\sbra{X\to_A Y}-\Pr\sbra{X\to_B Y}}}\\
        &\qquad\cdot\sqrt{\sum_{X\in \mathrm{Supp}}g(X)^2\sum_{Y\in \mathrm{Supp}}\abs{\Pr\sbra{X\to_A Y}-\Pr\sbra{X\to_B Y}}}.
    \end{align*}
\end{proof}

This \Cref{lem:TV distance bound} will be frequently used in conjunction with the following bound relating sampling with replacement and sampling without replacement:
\begin{fact}\label{fact:k^2}
    If $k\leq \sqrt{N}$ then 
    \begin{align*}
        \prod_{i=0}^{k-1}\frac{N}{N-i}\leq 1+\frac{k^2}{N}.
    \end{align*}
\end{fact}
\begin{proof}
    We prove by induction on $i$. If $i=0$ then the result is trivially true. Now assume that $\prod_{i=0}^{k-2}\frac{N}{N-i}\leq 1+\frac{(k-1)^2}{N}$. Then
    \begin{align*}
        &\prod_{i=0}^{k-1}\frac{N}{N-i}
        \leq \pbra{1+\frac{(k-1)^2}{N}}\pbra{1+\frac{k-1}{N-k+1}}\\
        &\leq \pbra{1+\frac{(k-1)^2}{N}}\pbra{1+\frac{k}N}
        \leq \pbra{1+\frac{k^2}{N}}.
    \end{align*}
    The inequalities follow because $k\leq \sqrt{N}$.
\end{proof}

\subsection{Escape Probabilities}
Another idea that is key to our spectral gap proofs is to bound the contribution to spectra of ``badly-behaved regions" of the state spaces of random walks by showing that these sets exhibit probable escape. This \Cref{lem:escape probs} formalizes this linear algebraically.
\begin{lemma}\label{lem:escape probs}
    Suppose $A$ is a random walk matrix on a domain $\mathcal{U}$ such that the uniform distribution on $\mathcal{U}$ is a stationary distribution for $A$. Let $f,g:\mathcal{U}\to\R$ be such that for any $X\in \mathrm{Supp}(f)$,    
    \begin{align*}
        \Pr\sbra{X\to_{A}\mathrm{Supp}(g)}&\leq \epsilon.
    \end{align*}
    Then
    \begin{align*}
         \abs{\left\langle f,Ag\right\rangle} \leq \sqrt\epsilon\norm{f}_2\norm{g}_2.
    \end{align*}
\end{lemma}
\begin{proof}
    We directly compute:
    \begin{align*}
        &\abs{\left\langle f,Ag\right\rangle}\\
        =&\;\abs{\sum_{X\in\mathcal{U}}{f(X)\mathbf{1}\sbra{X\in \mathrm{Supp}(f)}\Ex_{\mathbf{Y}\sim\mathcal{D}_{X}}\sbra{g(\mathbf{Y})\mathbf{1}\sbra{\mathbf{Y}\in \mathrm{Supp}(g)}}}}\\
        \leq&\;{\norm{f}_2\sqrt{\sum_{X\in\mathcal{U}}{\mathbf{1}\sbra{X\in \mathrm{Supp}(f)}^2\Ex_{\mathbf{Y}\sim\mathcal{D}_{X}}\sbra{g(\mathbf{Y})\mathbf{1}\sbra{\mathbf{Y}\in \mathrm{Supp}(g)}}^2}}}\tag{Cauchy-Schwarz}\\
        \leq&\;{\norm{f}_2\sqrt{\sum_{X \in\mathcal{U}}{\mathbf{1}\sbra{X\in \mathrm{Supp}(f)}\pbra{\sqrt{\Ex_{\mathbf{Y}\sim\mathcal{D}_{X}}\sbra{g(\mathbf{Y})^2}\Ex_{\mathbf{Y}\sim\mathcal{D}_{X}}\sbra{\mathbf{1}\sbra{\mathbf{Y}\in \mathrm{Supp}(g)}^2}}}^2}}}\tag{Cauchy-Schwarz}\\
        =&\;{\norm{f}_2\sqrt{\sum_{X \in\mathcal{U}}{\mathbf{1}\sbra{X\in \mathrm{Supp}(f)}\Ex_{\mathbf{Y}\sim\mathcal{D}_{X}}\sbra{g(\mathbf{Y})^2}\Ex_{\mathbf{Y}\sim\mathcal{D}_{X}}\sbra{\mathbf{1}\sbra{\mathbf{Y}\in \mathrm{Supp}(g)}}}}}\\
        =&\; \norm{f}_2\sqrt{\sum_{X\in\mathcal{U}}{\mathbf{1}\sbra{X\in \mathrm{Supp}(f)}\Ex_{\mathbf{Y}\sim\mathcal{D}_{X}}\sbra{g(\mathbf{Y})^2}}\Pr\sbra{X\to_{A}\mathrm{Supp}(g)}}.
    \end{align*}
    Now, consider each possible value that $X$ may take. If $X\in\mathrm{Supp}(f)$ then by the assumption, we have that $\Pr\sbra{X\to_{A}\mathrm{Supp}(g)}\leq \epsilon$. Otherwise if $X\not\in\mathrm{Supp}(f)$ we have that $\mathbf{1}\sbra{X\in\mathrm{Supp}(f)}=0$. In either case, we have $\mathbf{1}\sbra{X\in \mathrm{Supp}(f)}\Pr\sbra{X\to_{A}\mathrm{Supp}(g)}\leq \epsilon$. Continuing with our calculation, we find that:
    \begin{align*}
        &\abs{\left\langle f,Ag\right\rangle }
        \leq\norm{f}_2\abs{ \sqrt{\abs{\mcU}\Ex_{\mathbf{X} \sim\mathcal{U}}\sbra{\epsilon\Ex_{\mathbf{Y}\sim\mathcal{D}_{\mathbf{X}}}\sbra{g(\mathbf{Y})^2}}} } =\sqrt\epsilon\norm{f}_2\abs{\sqrt{\abs{\mcU}\Ex_{\mathbf{X} \sim \mathcal{U}}\sbra{\Ex_{\mathbf{Y} \sim \mathcal{D}_{\mathbf{X}}}\sbra{g(\mathbf{Y})^2}}}}.
    \end{align*}
    At this point, we note that sampling $\mathbf{X}$ uniformly at random and then sampling $\mathbf{Y}\sim\mathcal{D}_{\mathbf{X}}$ gives the same distribution as sampling $\mathbf{Y}$ uniformly by stationarity of the uniform distribution. This completes the proof by the following:
    \begin{align*}
        \abs{\left\langle f,Ag\right\rangle}\leq& \sqrt\epsilon\norm{f}_2\abs{\sqrt{\abs{\mcU}\Ex_{\mathbf{X} \in\mathcal{U}}\sbra{g(\mathbf{X})^2}}}=\sqrt\epsilon\norm{f}_2\norm{g}_2.
    \end{align*}
\end{proof}

\section{Spectral Gaps}\label{sec:spectral gaps}

\subsection{Fully Random Gates}
Throughout this section let $n$ be a fixed positive integer. Recalling \Cref{def:the random walk operator}, the operators in focus in this section are operators of the following form for $m\leq n$:
\begin{align*}
    R_{n,m,k}=\frac1{\binom{n}{m}}\sum_{S\in \binom{[n]}{m}}R_{n,S,k}.
\end{align*}
Also define the Laplacian $L_{n,m,k}=L(R_{n,m,k})$.

In this section we prove \Cref{thm:one-random-nonlocal} by building on the previously-mentioned following result of Brodsky and Hoory:

\begin{theorem}[\cite{brodsky2008simple}, Theorem 2]\label{Brodsky-Hoory}
    For any $m$ and $f:\{\pm1\}^{mk}\to\R$ and $k\leq 2^m-2$, we have that
    \begin{align*}
       \abs{ \left\langle f,(R_{m,3,k}-R_{m,m,k})f\right\rangle }\leq 1-\Omega\pbra{\frac1{m^2k}}\langle f,f\rangle.\footnotemark
    \end{align*}
\end{theorem}
\footnotetext{\cite{brodsky2008simple} actually proves this inequality for the operator $R_{m,3,k}^{\des{2}}$ in place of $R_{m,3,k}$, which is the random walk operator induced by placing a random width-2 permutation (which acts on 3 bits). However, a standard comparison of Markov chains shows that our statement of the result easily follows. See \Cref{sec:DES gate set}.}
Our main contribution is a finer analysis in the case when $k$ is small relative to $m$, which results in the following theorem.

\begin{theorem}\label{thm:small k}
    Assume that $m\geq100$ and $k\leq 2^{m/10}$. Given $f:\{\pm1\}^{mk}\to\R$, we have 
    \begin{align*}
       \abs{ \left\langle f,(R_{m,m-1,k}-R_{m,m,k})f\right\rangle }\leq  \pbra{\frac1m+\frac{k^2}{2^{m/4}}}\left\langle f,f\right\rangle.
    \end{align*}
\end{theorem}

\Cref{thm:small k} is proven as \Cref{thm:small k internal} in \Cref{sec:small k}.

The following lemma allows us to combine \Cref{thm:small k} with the previously-known \Cref{Brodsky-Hoory} to bring the quadratic dependence on the number of wires ($n$) to linear.

\begin{lemma}[\cite{o2023explicit}, Lemma 3.2]\label{lem:induction}
    Fix a positive integer $n_0\geq 4$. For each $k$ and $m_1\geq m_2$ let $\tau_{m_1,m_2,k}$ be some real number such that
    \begin{align*}
        L_{m_1,m_2,k}\geq \tau_{m_1,m_2,k}L_{m_1,m_1,k}.\footnotemark
    \end{align*}
    \footnotetext{All inequalities between operators in this paper are in the PSD order.}
    Then for any sequence $n_0=m_0\leq m_1\leq \dots\leq m_{t-1}\leq m_t = n$ we have
    \begin{align*}
        L_{n,n_0,k}\geq \pbra{\prod_{i\in[t-1]}\tau_{m_i,m_{i-1},k}}L_{n,n,k}.
    \end{align*}
\end{lemma}

Together, \Cref{Brodsky-Hoory}, \Cref{thm:small k}, and \Cref{lem:induction} yield the following initial spectral gap. 

\begin{corollary}\label{cor:initial spectral gap}
    For any $n$ and $k\leq 2^{n/3}$, we have
    \begin{align*}
        L_{n,3,k}\geq \Omega\pbra{\frac1{nk \cdot\log k}}L_{n,n,k}.
    \end{align*}
\end{corollary}
\begin{proof}
    \Cref{thm:small k} shows that for all $m\geq 20\log k$ we have $\tau_{m,m-1,k}\geq 1-\frac1m-\frac{k^2}{2^{m/3}}$. \Cref{Brodsky-Hoory} shows that 
    \begin{align*}
        L_{20\log k,3,k}\geq  \Omega\pbra{\frac1{k\cdot\log^2k }}L_{20\log k,20\log k,k},
    \end{align*}
    or equivalently $\tau_{20\log k,3,k}\geq \Omega\pbra{\frac1{k\log^2 k}}$.

    Combining these using \Cref{lem:induction} we have for large enough $k$ that
    \begin{align*}
        L_{n,3,k}\geq &\Omega\pbra{\frac1{k\log^2 k}}\prod_{m=20\log k}^{n}\pbra{1-\frac1m-\frac{k^2}{2^{m/4}}}L_{n,n,k}\\
        \geq &\Omega\pbra{\frac1{k\log^2 k}}\frac{20\log k}{n} \prod_{m=20\log k}^{n}\pbra{1-\frac{k^2}{2^{m/4-1}}}\cdot L_{n,n,k}\\
        \geq&\Omega\pbra{\frac1{nk\cdot\log k}}\cdot\pbra{1-\sum_{m=20\log k}^{n}\frac{k^2}{2^{m/4-1}}}\cdot L_{n,n,k} \geq \Omega\pbra{\frac1{nk\cdot\log k}}\cdot L_{n,n,k}.\qedhere
    \end{align*}
\end{proof}

We leverage the initial spectral gap from \Cref{cor:initial spectral gap} to produce our designs by sequentially composing many copies of this pseudorandom permutation. This is akin to showing that the second largest eigenvalue of the square of a graph is quadratically smaller than that of the original graph.

\Cref{thm:one-random-nonlocal} follows directly from \Cref{sec:DES gate set} with standard spectral-gap to mixing time bounds.

\subsection{Nearest-Neighbor Random Gates}
\subsubsection{Reduction to the Large $k$ Case}

One step in a random reversible circuit with 1D-nearest-neighbor gates is described by the operator $R_{n,[n-2],k}^{3\text{-NN}}$. Note that we can write 
\begin{align*}
    R_{n,[n-2],k}^{3\text{-NN}}=& \Ex_{\mathbf{a}\in[n-2]}\sbra{R_{n,\{\mathbf{a},\mathbf{a}+1,\mathbf{a}+2\},k}}.\footnotemark
\end{align*}
Define the corresponding Laplacian $L_{n,I,k}^{\gamma\text{-NN}}=L(R_{n,I,k}^{\gamma\text{-NN}})$. 

Because the local terms $R_{n,\{\mathbf{a},\mathbf{a}+1,\mathbf{a}+2\},k}$ are projectors, to analyze such an operator we can use the following theorem of Nachtergaele.
\footnotetext{More generally, we write $R_{n,S,k}^{\gamma\text{-NN}}$ to denote $\Ex_{a\in S}\sbra{R_{n,\{\mathbf{a},\dots,\mathbf{a}+\gamma-1\},k}}$.}
\begin{theorem}[\cite{nachtergaele1996spectral}, Theorem 3]\label{thm:nachtergaele}
    Let $\{h_{a,a+1,a+2}\}_{a\in[n-2]}$ be projectors acting on $(\R^2)^{\otimes d}$ such that each $h_{\{a,a+1,a+2\}}$ only acts on the $a,a+1,a+2$th tensor factor. For $I=[a,b]\subseteq[n]$ define the subspace
    \begin{align*}
        \mathcal{G}_I=& \cbra{f\in (\R^2)^{\otimes d}:\sum_{a'\in [a,b-2]}h_{a',a'+1,a'+2}f = 0}.
    \end{align*}
    Let $G_I$ be the projector to $\mathcal{G}_I$.

    Now suppose there exists $\ell$ and $n_\ell$ and $\epsilon_\ell\leq \frac1{\sqrt\ell}$ such that for all $n_\ell\leq m\leq n$,
    \begin{align*}
        \norm{G_{[m-\ell-1,m]}\pbra{G_{[m-1]}-G_{[m]}}}_{\mathrm{op}}\leq \epsilon_\ell.
    \end{align*}
    Then
    \begin{align*}
        \lambda_{2}\pbra{\sum_{a\in[n-2]}h_{\{{a},{a}+1,{a}+2\}}}&\geq \frac{\pbra{1-\epsilon_\ell\sqrt\ell}^2}{\ell-3}\lambda_{2}\pbra{
        \sum_{a\in[\ell]}h_{\{{a},{a}+1,{a}+2\}}}.
    \end{align*}
    Recall $\lambda_2(h)$ denotes the second-smallest distinct eigenvalue of the operator $h$.
\end{theorem}

\begin{theorem}\label{thm:nachtergaele hypothesis}
    Fix any $m\geq 100$ and $k\leq 2^m-2$ and set $\ell=10\log k$. Then we have
    \begin{align*}
        \norm{R_{m,[m-\ell-1,m],k}\pbra{R_{m,[m-1],k}-R_{m,[m],k}}}_{\mathrm{op}}\leq \frac1\ell.
    \end{align*}
\end{theorem}
\Cref{thm:nachtergaele hypothesis} is established in \Cref{sec:nachtergaele hypothesis} as \Cref{thm:nachtergaele hypothesis internal}.

\begin{corollary}\label{cor:reduce to logk}
    We have for $k\geq 3$ that 
    \begin{align*}
        \lambda_2\pbra{L_{n,[n-2],k}^{3\text{-NN}}}\geq \frac{1}{2n}\lambda_2\pbra{L_{10\log k+2,[10\log k-2],k}^{3\text{-NN}}}
    \end{align*}
\end{corollary}
\begin{proof}
    Setting $h_{a,a+1,a+2}=\Id-R_{n,\{a,a+1,a+2\},k}$ for each $a\in [n]$, we see that the projections $G$ (as in the statement of \Cref{thm:nachtergaele}) are given by $G_{[a]}=R_{n,[\min\{n,a+2\}],k}$ for any $a\in[n]$. To see this, note first that $R_{n,[\min\{n,a+2\}],k}$ is indeed a projection. Now let $f$ be such that $R_{n,[\min\{n,a+2\}],k}f=f$. For every $\sigma\in \mathfrak{S}_{\{\pm1\}^{n}}$ define $f^\sigma$ by $f^\sigma(X)=f(\sigma X)$ for all $X\in\{\pm1\}^{nk}$. Then by invariance of $R_{n,[\min\{n,a+2\}],k}$ under a permutation applied to bits in $[\min\{n,a+2\}]$ we have
    \begin{align*}
        f^\sigma = R_{n,[\min\{n,a+2\}],k}f^\sigma =R_{n,[\min\{n,a+2\}],k}f =f.
    \end{align*}
    for any $\sigma \in \mathfrak{S}_{\{\pm1\}^{[\min\{n,a+2\}]}}$. The converse of this holds as well by a similar argument. Therefore, $f^{\sigma^{\{a',a'+1,a'+2\}}} = f$ for any $\sigma\in \mathfrak{S}_{\{0,1\}^8}$ for $a'\leq a$, proving that such an $f$ is truly in the ground space of $R_{n,\{a',a'+1,a'+2\},k}$ for any $a'\leq a$. The converse of this holds as well, because the permutations of the form $\sigma^{\{a',a'+1,a'+2\}}$ with $a'\leq a$ generate the group of permutations of the form $\rho^{\{a,b,c\}}$ for $\{a,b,c\}\in\binom{[\min\{n,a+2\}]}{3}$. Such an argument also proves the converse, so the $R$ operators are serve as the projections from \Cref{thm:nachtergaele}.
    
    Therefore, by \Cref{thm:nachtergaele hypothesis} we have that the hypotheses of \Cref{thm:nachtergaele} are satisfied with $\ell=10\log k$ and $\epsilon_\ell = \frac1\ell$. That is, for any $m\leq n$ we have
    \begin{align*}
    &\norm{G_{[m-\ell-1,m]}\pbra{G_{[m-1]}-G_{[m]}}}_{\mathrm{op}}
        =\norm{R_{n,[m-\ell-1,m],k}\pbra{R_{n,[m-1],k}-R_{n,[m],k}}}_{\mathrm{op}}\\
        =&\norm{\pbra{R_{m,[m-\ell-1,m],k}\pbra{R_{m,[m-1],k}-R_{m,[m],k}}}\otimes \Id_{[m+1,n]}}_{\mathrm{op}}
        \leq \frac1\ell.
    \end{align*}
    Therefore the conclusion of \Cref{thm:nachtergaele} is that 
    \begin{align*}
        &\lambda_2\pbra{(n-2)L_{n,[n-2],k}^{3\text{-NN}}}
        \geq \frac{\pbra{1-\frac1{\sqrt\ell}}^2}{\ell-2}\lambda_2\pbra{\ell L_{\ell+2,[\ell-2],k}^{3\text{-NN}}}
        \geq \frac{1}{2}\lambda_2\pbra{L_{\ell+2,[\ell-2],k}^{3\text{-NN}}}.
    \end{align*}
    Recalling our setting of $\ell$ completes the proof. 
\end{proof}

\subsubsection{Comparison Method for the Large $k$ Case}
We can use the spectral gap proved in \Cref{Brodsky-Hoory} for the random walk induced by completely random 3-bit gates to show a spectral gap for the random walk induced by random 3-bit gates, where the three bits on which the gate acts on are $a,a+1,a+2$ for some $a\in[n]$. Our proof is a simple application of the comparison method applied to random walks on multigraphs.

We take the following definition of (multi)graphs. A graph is a pair of sets $(V,E)$ such that there is a partition $E=\bigcup_{(x,y)\in V^2}E_{x,y}$. If $e\in E_{x,y}$ then we say that $e$ \textit{connects} the vertex $x$ to the vertex $y$, and we define $u(e)=x$ and $v(e)=y$. The degree $\deg(x)$ of a vertex $x\in V$ is the number of edges originating at $e$, or $\sum_{y\in V}|E_{x,y}|$. We say that a graph is regular if $\deg(x)=\deg(y)$ for all $x,y\in V$.

The \textit{random walk} on a graph $(V,E)$ beginning at a vertex $x\in V$ consists of the Markov chain $\{\bm{x}_i\}_{i\geq0}$ on state space $V$ such that $\bm{x}_0=x$ with probability 1, and to draw $\bm{x}_{i+1}$ given $\bm{x}_i$ we sample a uniform random edge $\bm{e}$ from $\bigcup_{y\in V}E_{x,y}$. We set $\bm{x}_{i+1}$ equal to the unique $y\in V$ such that $\bm{e}\in E_{x,y}$.

A \textit{Schreier graph} is a graph with vertices $V$ such that some group $\mathfrak{G}$ acts on $V$. Let $S\subseteq \mathfrak{G}$ be some subset of group elements. The edge set consists of elements of the form $(X,\sigma)$ ($\sigma\in S$), so that $(X,\sigma)\in E_{X,\sigma X}$. We call the resulting graph $\mathrm{Sch}(V,S)$.

\begin{definition}\label{def:congestion}
    Let $S$ and $\widetilde{S}$ be subsets of a group acting on a set $V$. For each ${\sigma}\in {S}$ let $\Gamma({\sigma})$ be a sequence $(\wt{\sigma}_1,\dots,{\wt\sigma}_t)$ of elements of $\wt S$ such that $\sigma v=\wt{\sigma}_t\dots \wt{\sigma}_1 v$ for all $v\in V$, so we regard $\Gamma$ as a map from $S$ to sets of paths using edges in $\wt{S}$. Define the \textit{congestion ratio} of $\Gamma$ to be
    \begin{align*}
        B({\Gamma})=&\max_{\wt\sigma\in \wt{S}}\cbra{|\wt{S}|\Ex_{\sigma\in S}\sbra{{N(\wt{\sigma},\Gamma(\sigma))|\Gamma({\sigma})|}}},
    \end{align*}
    where $N(\wt{\sigma},\Gamma(\sigma))$ is the number of times $\wt \sigma$ appears in the sequence $\Gamma(\sigma)$.
\end{definition}

\begin{lemma}\label{lem:comparison schreier}
    Let $G=\mathrm{Sch}(V,S)$ and $\wt{G}=\mathrm{Sch}(V,\wt{S})$ be the connected Schreier graphs of the action of a group. Let $L$ and $\widetilde{L}$ be the Laplacian operators for the non-lazy random walks on $G$ and $\wt{G}$, respectively. Suppose there exists a $\Gamma$ as in \Cref{def:congestion}. Then 
    \begin{align*}
        \lambda_2(L)\geq \pbra{\max_{v\in V}\frac{\pi(v)}{\widetilde{\pi}(v)}}B({\Gamma})\lambda_2(\wt{L}).
    \end{align*}
    Here $\pi$ and $\widetilde{\pi}$ are the stationary distributions for $G$ and $\widetilde{G}$, respectively.
\end{lemma}

We prove this in \Cref{appendix:comparison}. It is essentially a reformulation of a standard result about comparisons on general Markov chains in \cite{wilmer2009markov}. 

\begin{lemma}\label{lem:initial spectral gap local random gates large k}
    For any $n,k$ we have
    \begin{align*}
        \lambda_2\pbra{L_{n,[n-2],k}^{3\text{-NN}}}\geq \frac1{100000n^3}\lambda_2\pbra{L_{n,3,k}}.
    \end{align*}
\end{lemma}
\begin{proof}
    Note that $L_{n,3,k}$ and $L_{n,[n-2],k}^{3\text{-NN}}$ are simply the Laplacians of random walks on Schreier graphs with $\mathfrak{S}_{\{\pm1\}^{n}}$ acting on $\{\pm1\}^{nk}$ by $e(X^1,\dots,X^k)=(eX^1,\dots,eX^k)$. In the case of $L_{n,3,k}$ the edges are given by elements of the form $h^{\{a,b,c\}}$ for $h\in\mathfrak{S}_{\{\pm1\}^3}$ and $\{a,b,c\}\in\binom{[n]}{3}$. In the case of $L_{n,[n-3],k}^{3\text{-NN}}$ the edges are given by elements of the form $g^{\{a,a+1,a+2\}}$ for $g\in\mathfrak{S}_{\{\pm1\}^3}$ and $a\in[n-2]$. We deal with each connected component separately. Note that every connected component is isomorphic to $\{(X^1,\dots,X^{k'}):X^i\neq X^j\iff i\neq j\}$ for $k'\leq k$, so we bound the spectral gap for the walk on $\{(X^1,\dots,X^{k}):X^i\neq X^j\iff i\neq j\}$.

    We provide a map $\Gamma$ from $\{h^{\{a,b,c\}}:h\in\mathfrak{S}_{\{0,1\}^3},a,b,c\in[n]\}$ to sequences of elements of the form $g^{\{a,a+1,a+2\}}$ for $a\in[n-2]$ such that for any $h^{\{a,b,c\}}$ the sequence $\Gamma(h^{\{a,b,c\}})=(g_1^{\{a_1,a_1+1,a_1+2\}},\dots ,g_t^{\{a_t,a_t+1,a_t+2\}})$ satisfies $h^{\{a,b,c\}}=g_t^{\{a_t,a_t+1,a_t+2\}}\dots g_1^{\{a_1,a_1+1,a_1+2\}}$.

    The hope is to construct $\Gamma$ such that $B(\Gamma)$ is small and then to apply \Cref{lem:comparison schreier}. To this end we define $\Gamma$ as follows. Assume $a<b<c$. Fix $h\in\mathfrak{S}_8,a\in[n-2]$. Let $d\in[n-2]$ be arbitrary. Then write
    \begin{align*}
        h^{\{a,b,c\}}=\mathsf{Sort}^{-1} \cdot g^{\{d,d+1,d+2\}}\cdot\mathsf{Sort},
    \end{align*}
    Here $\mathsf{Sort}$ sends the $a$th coordinate to the $d$th coordinate, the $b$th to the $d+1$th, and the $c$th to the $d+2$th. The permutations $\mathsf{Sort}$ and $\mathsf{Sort}^{-1}$ can each be implemented using at most $3n$ gates of the form $g^{\{a'-1,a',a'+1\}}$ and $g^{\{a'+1,a'+2,a'+3\}}$ where each $g$ swaps coordinates; this is just by a standard partial sorting algorithm. Write $\mathsf{Sort}=g_{3n}^{\{a_{3n},a_{3n}+1,a_{3n}+2\}}\dots g_1^{\{a_1,a_1+1,a_1+2\}}$. Then set 
    \begin{align*}
        \Gamma(h^{\{a,b,c\}})=\pbra{g_1^{\{a_1,a_1+1,a_1+2\}},\dots,g_{3n}^{\{a_{3n},a_{3n}+1,a_{3n}+2\}},h^{(d,d+1,d+2)},\dots,\pbra{g_1^{\{a_1,a_1+1,a_1+2\}}}^{-1}}.
    \end{align*}
    We have $B(\Gamma)\leq 100000n^3$ trivially, so we have proved the result by applying \Cref{lem:comparison schreier} and the fact that the stationary distributions for both chains are uniform.
\end{proof}

\begin{corollary}\label{cor:initial spectral gap local}
For any $n,k$ we have
    \begin{align*}
        \lambda_2\pbra{L_{n,[n-2],k}^{3\text{-NN}}}\geq \Omega\pbra{\frac1{nk\cdot\log^5(k)}}.
    \end{align*}
\end{corollary}
\begin{proof}
    \Cref{lem:initial spectral gap local random gates large k} shows that $\lambda_2\pbra{L_{\ell+2,[\ell],k'}^{3\text{-NN}}}\geq \frac{1}{\ell^3}\lambda_2\pbra{L_{\ell+2,3,k'}}$ for all $\ell$ and $k'$. \Cref{Brodsky-Hoory} states that $\lambda_2\pbra{L_{\ell+2,3,k'}}\geq \frac{1}{(\ell+2)^2k'}$ for all $\ell$ and $k'$. If we set $\ell=10\log k$ and use \Cref{cor:reduce to logk} then we get
    \begin{align*}
         &\lambda_2\pbra{L_{n,[n-2],k}^{3\text{-NN}}} 
         \geq \frac{1}{2n}\cdot\lambda_2\pbra{L_{10\log k+2,[10\log k],k}^{3\text{-NN}}}
         \geq  \frac{1}{2n}\cdot\frac1{100000k\log^5(k)}.
    \end{align*}
    This implies the result.
\end{proof}

As in the case for fully random gates, \Cref{thm:one-random-local} follows from \Cref{cor:initial spectral gap local} and \Cref{sec:DES gate set}.

\subsection{Restricting the Gate Set}\label{sec:DES gate set}

So far all of our results have dealt with random circuits with arbitrary gates acting on 3 bits. However, for practical applications we are often in further restricting the type of 3-bit gates. However, as long as the arbitrary gate set is universal on 3 bits, we lose just a constant factor in the mixing time when we restrict our random circuits to use that gate set, by a standard application of the comparison method. We prove that we can perform this conversion before proving our results about brickwork circuits in \Cref{sec:brickwork} because it is somewhat easier to prove for the case of single gates acting at a time.

\begin{lemma}\label{lem:any gate set comparison}
    We have
    \begin{align*}
        \lambda_2\pbra{L_{n,[n-2],k}^{3\text{-NN},\des{2}}}\geq \Omega\pbra{\lambda_2\pbra{L_{n,[n-2],k}^{3\text{-NN}}}}
    \end{align*}
\end{lemma}
\begin{proof}
    We compare the Markov chains given by these two Laplacians by providing a way to write edges in the one induced by arbitrary 3-bit nearest-neighbor gates as paths in the one induced by 3-bit nearest-neighbor gates with generators $\mathcal{G}$. Again we focus on the connected component $\{(X^1,\dots,X^{k}):X^i\neq X^j\iff i\neq j\}$.

    For each $g^S$ for $g\in\mathfrak{S}_8$ and $S\in \binom{[n]}{3}$ let $\Gamma(g^S)=\pbra{g_{i_1}^S,\dots,g_{i_{8!}}^S}$ where we have fixed an arbitrary expansion of $g=g_{i_1}\dots g_{i_{8!}}$, and each $g_{i_j}$ is of type $\des{2}$. Then in the notation of \Cref{lem:comparison schreier} we have that
    \begin{align*}
        B(\Gamma)=&\max_{g\in \mathcal{G},a\in[n-2]}\cbra{\abs{\mathcal{G}}(n-2)\Ex_{\mathbf{g}\in \mathfrak{S}_{8},\mathbf{a}\in[n-2]}\sbra{N\pbra{g^{\{a,a+1,a+2\}},\Gamma(\mathbf{g}^{\{\mathbf a,\mathbf a+1,\mathbf a+2\}})}\abs{\Gamma(\mathbf{g}^{\{\mathbf a,\mathbf a+1,\mathbf a+2\}})}}}\\
        \leq &\max_{g\in \mathcal{G},a\in [n-2]}\cbra{8!n\Ex_{\mathbf{g}\in \mathfrak{S}_{8},\mathbf{a}\in[n-2]}\sbra{8!N\pbra{g^{\{a,a+1,a+2\}},\Gamma(\mathbf{g}^{\{\mathbf a,\mathbf a+1,\mathbf a+2\}})}}}\\
        \leq &\max_{g\in \mathcal{G},a\in [n-2]}\cbra{(8!)^3n\Pr_{\mathbf{g}\in \mathfrak{S}_{8},\mathbf{a}\in[n-2]}\sbra{g^{\{a,a+1,a+2\}}\in \Gamma(\mathbf{g}^{\{\mathbf a,\mathbf a+1,\mathbf a+2\}})}}\\
        \leq &\max_{g\in \mathcal{G},a\in [n-2]}\cbra{(8!)^3n\Pr_{\mathbf{g}\in \mathfrak{S}_{8},\mathbf{a}\in[n-2]}\sbra{a=\mathbf{a}}}\\
        \leq &(8!)^3.
    \end{align*}
    Applying \Cref{lem:comparison schreier} completes the proof.
\end{proof}

This shows that the random walk given by applying random gates on 3 bits of type $\des{2}$ has spectral gap $\wt{\Omega}\pbra{1/nk}$, by combining with \Cref{cor:initial spectral gap local}. A similar proof essentially shows the same result for the random circuit models where gates on arbitrary sets of 3 bits.

\subsection{Brickwork Circuits}\label{sec:brickwork}
The spectral gap for brickwork circuits follows almost directly from the spectral gap for circuits with nearest-neighbor gates (\Cref{cor:initial spectral gap local}), as in~\cite{brandao2016local}. First we show that the random walk induced by 3-bit nearest neighbor $\des{2}$ gates, where the 3 bits on which gates act on are of the form $\{a,a+1,a+2\}$ for any $a\in[n-2]$, has approximately the same spectral gap as that in which the random gates are of the form $\{a,a+1,a+2\}$ for $a\in[n-2]$ but with the restriction that $a\neq 0$ mod 3. Use the notation $L_{n,[n-2],k}^{3\text{-NN},\des{2}}$ and $L_{n,\{a\in[n-2],a=1,2\text{ mod } 3\},k}^{3\text{-NN},\des{2}}$ for the Laplacians of these random walks. Assume that $n=0$ mod 3; the other cases follow similarly. 

\begin{lemma}\label{lem:012 to 01}
    For any $n,k$ we have
    \begin{align*}
       \lambda_2\pbra{L_{n,\{a\in[n-2],a=1,2\text{ mod } 3\},k}^{3\text{-NN},\des{2}}}\geq  \Omega\pbra{\lambda_2\pbra{L_{n,[n-2],k}^{3\text{-NN}}}}.
    \end{align*}
\end{lemma}
\begin{proof}
    By \Cref{lem:any gate set comparison}, it suffices to show 
    \begin{align*}
       \lambda_2\pbra{L_{n,\{a\in[n-2],a=1,2\text{ mod } 3\},k}^{3\text{-NN},\des{2}}}\geq  \Omega\pbra{\lambda_2\pbra{L_{n,[n-2],k}^{3\text{-NN},\des{2}}}}.
    \end{align*}
    We use the comparison method. Again we focus on the connected component $\{(X^1,\dots,X^{k}):X^i\neq X^j\iff i\neq j\}$. For each $g\in\mathfrak{S}_{\{0,1\}^3}\cong \mathfrak{S}_8$ of type $\des{2}$ and $a\in[n-2]$ we provide a sequence $\Gamma(g^{\{a,a+1,a+2\}})$ of permutations multiplying to $g^{\{a,a+1,a+2\}}$ using only permutations of the form $h^{\{b,b+1,b+2\}}$ with $b\neq 0$ mod 3 such that the resulting congestion $B(\Gamma)$ is small. 

    We define $\Gamma$ as follows. Fix $g\in\mathfrak{S}_8,a\in[n-2]$ where $g$ is of type $\des{2}$. If $a\neq 0$ mod 3 then simply set $\Gamma(g^{\{a,a+1,a+2\}})=(g^{\{a,a+1,a+2\}})$. Otherwise $a=0$ mod 3. Then there exists a sequence of 64! permutations of the form $g_i^{\{b_i,b_i+1,b_i+2\}}$ with each $b_i\in \{a-1,a+1\}$ such that $g=g_1^{\{b_1,b_1+1,b_{1}+2\}}\cdots g_1^{\{b_{64!},\dots,b_{64!}+1,\dots,b_{64!}+2\}}$. This is because we can implement the gate $g^{\{a,a+1,a+2\}}$ as 
    \begin{align*}
        g^{\{a,a+1,a+2\}}=\mathsf{Sort}^{-1} \cdot g^{\{a-1,a,a+1\}}\cdot\mathsf{Sort},
    \end{align*}
    where $\mathsf{Sort}$ sends $(x_1,\dots,x_{a-1},x_a,x_{a+1},x_{a+2},\dots,x_n)\to (x_1,\dots,x_a,x_{a+1},x_{a+2},x_{a-1},\dots,x_n)$. The permutations $\mathsf{Sort}$ and $\mathsf{Sort}^{-1}$ can each be implemented as the product of at most $32!$ permutations of the form $h^{\{a-1,a,a+1\}}$ and $h^{\{a+1,a+2,a+3\}}$ where each $h$ is of type $\des{2}$. This gives the implementation of $g^{\{a,a+1,a+2\}}$ as the product of at most $64!$ elements of the form $g^{\{a-1,a,a+1\}}$ or $g^{\{a,a+1,a+2\}}$, and this defines $\Gamma(g^{\{a,a+1,a+2\}})$ for such $g$ and $a$. 
    
    In the notation of \Cref{lem:comparison schreier}, the congestion of $\Gamma$ is bounded:
    \begin{align*}
        B(\Gamma)\leq &\max_{g\in \mathcal{G},a=1,2\text{ mod 3}}\cbra{\frac{8!\cdot2n}{3}\Ex_{\mathbf{g}\in \mathfrak{S}_{8},\mathbf{a}\in[n-2]}\sbra{N\pbra{g^{\{a,a+1,a+2\}},\Gamma(\mathbf{g}^{\{\mathbf a,\mathbf a+1,\mathbf a+2\}})}\abs{\Gamma(\mathbf{g}^{\{\mathbf a,\mathbf a+1,\mathbf a+2\}})}}}\\
        \leq & 70!\max_{g\in \mathcal{G},a=1,2\text{ mod 3}}\cbra{n\Ex_{\mathbf{g}\in \mathfrak{S}_{8},\mathbf{a}\in[n-2]}\sbra{N\pbra{g^{\{a,a+1,a+2\}},\Gamma(\mathbf{g}^{\{\mathbf a,\mathbf a+1,\mathbf a+2\}})}}}\\
        \leq &70!\max_{g\in \mathcal{G},a=1,2\text{ mod 3}}\cbra{n\Pr_{\mathbf{g}\in \mathfrak{S}_{8},\mathbf{a}\in[n-2]}\sbra{g^{\{a,a+1,a+2\}}\in \Gamma(\mathbf{g}^{\{\mathbf a,\mathbf a+1,\mathbf a+2\}})}}\\
        \leq &70!\max_{g\in \mathcal{G},a=1,2\text{ mod 3}}\cbra{n\Pr_{\mathbf{g}\in \mathfrak{S}_{8},\mathbf{a}\in[n-2]}\sbra{a\in\{\mathbf{a}-1,\mathbf{a},\mathbf{a}+1\}}}\\
        \leq & 71!.
    \end{align*}
    Here we used that $\abs{\Gamma(g^{\{a,a+1,a+2\}})}\leq 64!$ always. Applying \Cref{lem:comparison schreier} completes the proof.
\end{proof}

Given this restriction to $\des{2}$ gates, the idea to prove the spectral gap for the random walk corresponding to a random brickwork circuit is to write the transition operator corresponding to one layer of brickwork gates as
\begin{align*}
    &R_{n,k}^{\text{brickwork},\des{2}}\\
    =&R_{n,\{1,2,3\},k}^{\des{2}}R_{n,\{4,5,6\},k}^{\des{2}}\cdots R_{n,\{n-2,n-1,n\},k}^{\des{2}}R_{n,\{2,3,4\},k}^{\des{2}}R_{n,\{5,6,7\},k}^{\des{2}}\cdots R_{n,\{n-4,n-3,n-2\},k}^{\des{2}}.
\end{align*}
Here the operator $R_{n,S,k}^{\des{2}}$ is the transition operator for the Markov chain that is similar to $R_{n,S,k}$ but with the restriction that each step is induced by a gate of type $\des{2}$. Then we note that each individual factor in each of the two products all commute, and every factor commutes with all but at most 7 other factors. The detectability lemma~\cite{aharonov2009detectability} bounds the spectral gap of this operator.

\begin{lemma}[\cite{brandao2016local}, Section 4.A]\label{lem:local to brickwork}
    For any $n,k$ we have
    \begin{align*}
        \lambda_2\pbra{L_{n,k}^{\mathrm{brickwork},\des{2}}}\geq  n\Omega\pbra{\lambda_2\pbra{L_{n,\{a\in[n-2],a=1,2\text{ mod } 3\},k}^{3\text{-NN},\des{2}}}}.
    \end{align*}
\end{lemma}

\begin{corollary}\label{cor:initial spectral gap brickwork}
    For any $n,k$ we have
    \begin{align*}
        \lambda_2\pbra{L_{n,k}^{\mathrm{brickwork},\des{2}}}\geq  \Omega\pbra{\frac1{k\cdot\polylog(k)}}.
    \end{align*}
\end{corollary}

As in the case for fully random gates and nearest-neighbor random gates, \Cref{thm:one-layer-brickwork} follows from \Cref{cor:initial spectral gap brickwork}.

\section{Proof of \Cref{thm:small k}}\label{sec:small k}
Throughout this section fix $m\geq 3$. Our goal in this section is to establish \Cref{thm:small k}, which states that $R_{m,m-1,k}-R_{m,m,k}$ has small spectral norm. Informally, we show that completely randomizing $m-1$ out of $m$ wires in a reversible circuit is very similar to randomizing all $m$ wires. Recall that $R_{m,m-1,k}$ is defined via the distributions $\mathcal{D}_X^{m,m-1,k}$ (the equal mixture of $\mcD_X^{m,S,k}$ for all $S$ with $|S|=m-1$) for $X\in\{\pm1\}^{mk}$, from which one samples by sampling a random set $\mathbf{S}\subseteq [m]$ with $|\mathbf{S}|=1$ (so really $\mathbf{S}=\{\mathbf{a}\}$, where $\mathbf{a}$ is a random element of $[m]$), ``fixes" the entries $X^i_{\mathbf{a}}$, and applies a random permutation to the coordinates not equal to $\mathbf{a}$. We use this notation and terminology of a ``fixed" coordinate $\mathbf{a}$ throughout this section.

As alluded to in \Cref{sec:overview}, our proof will decompose the space $\R^{\{\pm1\}^{mk}}$ on which these operators act into three orthogonal components, to be defined in \Cref{sec:decomposition}. Then \Cref{sec:decomposition}, \Cref{sec:B=1}, and \Cref{sec:B>=2} will bound the contributions from vectors lying in these orthogonal components and their cross terms.

\subsection{An Orthogonal Decomposition}\label{sec:decomposition}
\begin{definition}
    Regard elements of $\{\pm1\}^{mk}$ as $k$-by-$m$ matrices, so that the $i$th row of $X$ is $X^i$, and the $a$th column of $X$ is $X_a$. Define 
    \begin{align*}
        B_{\geq 2}&=\cbra{X\in\{\pm1\}^{mk}:\forall i\neq j\in[k], d\pbra{X^i, X^j}\geq 2},\\
        B_{=1}&=\cbra{X\in\{\pm1\}^{mk}:\forall i\neq j\in[k], d\pbra{X^i, X^j}\geq 1}\setminus B_{\geq 2},\\
        B_{=0}&=\cbra{X\in\{\pm1\}^{mk}:\exists i\neq j\in[k], d\pbra{X^i, X^j}=0}.
    \end{align*}
\end{definition}

Our proof that $R_{m,m-1,k}-R_{m,m,k}$ has small spectral norm will go by induction on $k$. \Cref{lemma:f supported on B0} helps to connect the cases of $k-1$ and $k$ in the proof. In particular, it shows that we can pass those functions supported on $B_{=0}$ into the induction.

\begin{lemma}\label{lemma:f supported on B0}
    Let $f:\{\pm1\}^{mk}\to \R$ be supported on $B_{=0}$. Then for any $S_1,\dots,S_t\subseteq[m]$ and $c_1,\dots,c_t\in\R$ we have
    \begin{align*}
        \abs{\left\langle f,\sum_{s=1}^t c_s \pbra{R_{m,S_s,k}-R_{m,m,k}}f\right\rangle }\leq \norm{\sum_{s=1}^t c_s \pbra{R_{m,S_s,k-1}-R_{m,m,k-1}}}_{\mathrm{op}}\cdot\norm{f}_2^2.
    \end{align*}
\end{lemma}
\begin{proof}
    For any map $\phi:[k]\to[k-1]$ (viewed as a coloring of $[k]$ with $k-1$ colors) define the set 
    \begin{align*}
        \mathcal{J}_\phi = \cbra{\pbra{X^1,\dots,X^k}:X^i=X^j\iff \phi(i)=\phi(j)}.
    \end{align*}
    These sets $\mathcal{J}_\phi$ partition $B_{=0}$. Thus, for every $f$ supported on $B_{=0}$ we have a decomposition $f=\sum_{\phi}f_\phi$, where each $f_\phi$ is supported on $\mathcal{J}_\phi$. 

    Now, for each $\phi$ define a map $\mathrm{Res}_\phi:\cbra{f:{\{\pm1\}^{mk}}\to \R:f\text{ supported on }\mathcal{J}_\phi}\to \R^{\{\pm1\}^{m(k-1)}}$ by arbitrarily choosing $i,j\in[k]$ such that $\phi(i)=\phi(j)$ and defining for $f_\phi:\{\pm1\}^{mk}\to \R$ supported on $\mathcal{J}_{\phi}$ the new function $\mathrm{Res}_\phi f_\phi:\{\pm1\}^{m(k-1)}\to \R$ by defining for $X'\in \{\pm1\}^{m(k-1)}$
    \begin{align*}
        \mathrm{Res}_\phi f_\phi(X')=& f(\mathrm{Res}_\phi^*(X')).
    \end{align*}
    where $\mathrm{Res}_\phi^*(X')$ is the unique element of $\mathcal{J}_\phi$ such that $(\mathrm{Res}_\phi^*(X'))^{[k]\setminus\{j\}}=X'$. Note this is well-defined because $f_\phi$ is supported on $\mathcal{J}_\phi$.

    \begin{claim}\label{claim:Res norm}
        For any $f_\phi:{\{\pm1\}^{mk}}\to \R$ supported on $\mathcal{J}_\phi$ we have $\norm{f_\phi}_2=\norm{\mathrm{Res}_\phi f_\phi}_2$ and for any $S\subseteq[m]$,
        \begin{align*}
        \left\langle f_\phi,R_{m,S,k}f_\phi\right\rangle = \left\langle \mathrm{Res}_\phi f_\phi,R_{m,S,k}\mathrm{Res}_\phi f_\phi\right\rangle.
    \end{align*}
    \end{claim}
    We prove the claim later, and for now use it to compute
    \begin{align*}
        &\abs{\left\langle f,\sum_{s=1}^t c_s \pbra{R_{m,S_s,k}-R_{m,m,k}}f\right\rangle} \\
        =&\abs{\sum_{\phi,\phi'}\left\langle f_\phi , \sum_{s=1}^t c_s \pbra{R_{m,S_s,k}-R_{m,m,k}}f_{\phi'}\right\rangle} \\
        =&\abs{\sum_{\phi}\left\langle f_\phi , \sum_{s=1}^t c_s \pbra{R_{m,S_s,k}-R_{m,m,k}}f_{\phi}\right\rangle + \sum_{\phi\neq \phi'}\left\langle f_\phi , \sum_{s=1}^t c_s \pbra{R_{m,S_s,k}-R_{m,m,k}}f_{\phi'}\right\rangle} \\
        =&\abs{\sum_{\phi}\left\langle \mathrm{Res}_\phi f_\phi,\sum_{s=1}^t c_s \pbra{R_{m,S_s,k}-R_{m,m,k}}\mathrm{Res}_\phi f_\phi \right\rangle} \tag{\Cref{claim:Res norm}, \Cref{eq:different colorings} below}\\
        \leq & \sum_{\phi}\norm{\sum_{s=1}^t c_s \pbra{R_{m,S_s,k-1}-R_{m,m,k-1}}}_{\mathrm{op}}\norm{\mathrm{Res}_\phi f_\phi}_2^2 \\
        =&\norm{\sum_{s=1}^t c_s \pbra{R_{m,S_s,k-1}-R_{m,m,k-1}}}_{\mathrm{op}}\sum_{\phi}\norm{f_\phi}_2^2\\
        =& \norm{\sum_{s=1}^t c_s \pbra{R_{m,S_s,k-1}-R_{m,m,k-1}}}_{\mathrm{op}}\norm{f}_2^2.
    \end{align*}
    The last equality follows from orthogonality of the $f_\phi$.
    
    To take care of the cross terms, we observe that for any $X\in\mathcal{J}_\phi$ and $S\subseteq[n]$, we have $\Pr\sbra{X\to_{R_{m,S,k}} \mathcal{J}_{\phi'}}=0$ for any $\phi'\neq \phi$. Then by \Cref{lem:escape probs} we have
    \begin{align}\label{eq:different colorings}
        &\sum_{\phi\neq \phi'}\left\langle f_\phi , \sum_{s=1}^t c_s R_{m,S_s,k}f_{\phi'}\right\rangle =0. 
    \end{align}
    This completes the proof.
\end{proof}

\begin{proof}[Proof of \Cref{claim:Res norm}]
    Without loss of generality assume that $\phi(k-1)=\phi(k)$ so we can regard $\mathrm{Res}_\phi f_\phi$ as a real function on $\{\pm1\}^{m(k-1)}$. Then
    \begin{align*}
        &\left\langle f_\phi,f_\phi\right\rangle 
        =\sum_{X\in\{\pm1\}^{mk}}f_\phi(X)^2
        = \sum_{X\in \mathcal{J}_\phi}\pbra{{f_\phi(X)}}^2\\
        =& \sum_{X'\in \{\pm1\}^{m(k-1)}}\left(\mathrm{Res}_\phi f(X')\right)^2
        = \left\langle \mathrm{Res}_\phi f_\phi , \mathrm{Res}_\phi f_\phi \right\rangle.
    \end{align*}
    To prove the second statement, we compute
    \begin{align*}
        &\left\langle f_\phi,R_{m,S,k}f_\phi\right\rangle\\
        =&\sum_{X\in\{\pm1\}^{mk}}f_\phi(X)\sum_{Y\in\{\pm1\}^{mk}}f_\phi(Y)\Pr\sbra{X\to_{R_{m,S,k}} Y}\\
        =&\sum_{X\in\mathcal{J}_\phi}f_\phi(X)\sum_{Y\in\mathcal{J}_\phi}f_\phi(Y)\Pr\sbra{X\to_{R_{m,S,k}} Y}\\
        =&\sum_{X'\in\{\pm1\}^{m(k-1)}}\mathrm{Res}_\phi f_\phi(X')\sum_{Y'\in\{\pm1\}^{m(k-1)}}\mathrm{Res}_\phi f_\phi(Y')\Pr\sbra{X'\to_{R_{m,S,k-1}} Y'}\\
        =&\left\langle \mathrm{Res}_\phi f_\phi ,R_{m,S,k-1}\mathrm{Res}_\phi f_\phi  \right\rangle.
    \end{align*}
    The second-to-last equality follows because the corresponding $X,Y$ are in $\mathcal{J}_\phi$.
\end{proof}

We now prove \Cref{thm:small k}, deferring proofs of the remaining needed auxiliary results to \Cref{sec:B=1}, \Cref{sec:cross terms}, and \Cref{sec:B>=2}.
\begin{theorem}[\Cref{thm:small k} restated]\label{thm:small k internal}
    Let $m\geq 100$ and let $2\leq k\leq 2^{m/10}$. Given any $f:\{\pm1\}^{mk}\to\R$, we have
    \begin{align*}
       \abs{ \left\langle f,(R_{m,m-1,k}-R_{m,m,k})f\right\rangle }\leq  \pbra{\frac1m+\frac{k^2}{2^{m/4}}}\left\langle f,f\right\rangle.
    \end{align*}
\end{theorem}
\begin{proof}
    We prove by induction on $k$. In the base case $k=1$ and the result holds by the following argument when we write $f=f_2$. Now assume that the result holds for real functions on $\{\pm1\}^{m(k-1)}$. 
    
    Let $f:\{\pm1\}^{mk}\to\R$. Write $f=f_0+f_1+f_2$ where $f_0$ is supported on $B_{=0}$, $f_1$ is supported on $B_{=1}$, and $f_2$ is supported on $B_{\geq 2}$. By \Cref{lem:cross terms B0} applied with both $R_{m,m-1,k}$ and $R_{m,m,k}$ and $B_{=0}$ and $\{\pm1\}^{mk}\setminus B_{=0}=B_{=1}\cup B_{\geq 2}$, the other cross terms vanish, and we have
    \begin{align*}
        &\abs{ \left\langle f,(R_{m,m-1,k}-R_{m,m,k})f\right\rangle }\\
        \leq & \abs{\left\langle f_0,(R_{m,m-1,k}-R_{m,m,k})f_0\right\rangle}+\abs{\left\langle f_1,(R_{m,m-1,k}-R_{m,m,k})f_1\right\rangle}\\
        &\;\;\;\;\;\;+\abs{\left\langle f_2,(R_{m,m-1,k}-R_{m,m,k})f_2\right\rangle}+\abs{2\left\langle f_1,(R_{m,m-1,k}-R_{m,m,k})f_2\right\rangle }\tag{Self-adjointness (\Cref{fact:self-adjoint})}\\
        \leq & \pbra{\frac1m+\frac{(k-1)^2}{2^{m/4}}}\left\langle f_0,f_0\right\rangle +\pbra{\frac1m+\frac{m^5k}{2^{m/2-2}}}\left\langle f_1,f_1\right\rangle + \pbra{\frac1m+\frac{k^2}{2^{m/2}}}\left\langle f_2,f_2\right\rangle + \frac{\sqrt{m}k}{2^{m/2-2}}\norm{f_1}_2\norm{f_2}_2\tag{\Cref{lemma:f supported on B0} + induction, \Cref{cor:f1-f1}, \Cref{prop:B>=2 square term}, \Cref{lem:cross terms B=1 to B=2}, in that order}\\
        \leq & \pbra{\frac1m+\frac{k^2-k}{2^{m/4}}}\left\langle f,f\right\rangle+\frac{m^5k}{2^{m/2-3}}\norm{f_1}_2\norm{f_2}_2\\
        \leq & \pbra{\frac1m+\frac{k^2-k}{2^{m/4}}}\left\langle f,f\right\rangle+\frac{k}{2^{m/3}}\left\langle f,f\right\rangle\\
        \leq&\pbra{\frac1m+\frac{k^2}{2^{m/4}}}\left\langle f,f\right\rangle.\qedhere
    \end{align*}
\end{proof}

\subsection{$f$ Supported on $B_{=1}$}\label{sec:B=1}

We now use \Cref{lem:escape probs} to bound $\abs{\left\langle f, (R_{m,m-1,k}-R_{m,m,k})f\right\rangle}\leq \abs{\left\langle f, R_{m,m-1,k}f\right\rangle}$ when $f$ is supported only on $B_{= 1}$ and the cross terms contributed. We first define a partition of $B_{=1}$, and apply \Cref{lem:escape probs} to these different parts. One of our main observations to bound $f$ supported on $B_{=1}$ is the observation that when $k$ is small, it is highly unlikely that two vectors out of any $k$ are close to each other. Thus, a random walk beginning in $B_{=1}$ and obeying the transition probabilities given by $B_{m,m-1,k}$ will rarely remain in $B_{=1}$. Formally, \Cref{lem:escape probs} bounds the contributions to the spectral norm by these transition probabilities.
\begin{definition}
    For $S\subseteq[m]$ define the set $\mathcal{I}_S\subseteq\{\pm1\}^{mk}$ by
    \begin{align*}
        \mathcal{I}_S=\cbra{X\in B_{=1}: \forall a\in S ,\exists i,j\in[k]:\Delta\pbra{X^i,X^j}=\{a\}}.
    \end{align*}
\end{definition}
Unfortunately, the sets $\mathcal{I}_S$ for different $S\subseteq[m]$ do not form a partition of $B_{=1}$, since there is overlap between $\mathcal{I}_S$ and $\mathcal{I}_T$ for $S\neq T$. However, we can artificially make this into a partition.
\begin{definition}
    For $S\subseteq[m]$ with $|S|=1$ (so $S=\{a\}$ for some $a\in[m]$) define the set $\widetilde{\mathcal{I}}_S\subseteq\{\pm1\}^{nk}$ by
    \begin{align*}
        \widetilde{\mathcal{I}}_S=\mathcal{I}_S\setminus \bigcup_{a'\neq a}\mathcal{I}_{\{a'\}}.
    \end{align*}
    Now place an arbitrary ordering $\preceq $ on the set $\{S\subseteq[m]:|S|=2\}$ and define
    \begin{align*}
        \widetilde{\mathcal{I}}_S = \mathcal{I}_S\setminus \pbra{ \bigcup_{S'\subseteq [m]:|S'| = 2,S'\preceq S}\mathcal{I}_{S'}}.
    \end{align*}
\end{definition}

\begin{observation}\label{obs:Is partition}
    The collection of sets $\{\widetilde{\mathcal{I}}_S:S\subseteq[m],|S|\leq 2\}$ forms a partition of $B_{=1}$.
\end{observation}

\begin{lemma}\label{lem:square terms B=1}
    Let $S\subseteq[m]$ be such that $|S|\leq 2$. If $k\geq 2$ and $f:\{\pm1\}^{mk}\to\R$ is supported on $\widetilde{\mathcal{I}}_S$ then 
    \begin{align*}
        \abs{\left\langle f,(R_{m,m-1,k}-R_{m,m,k})f\right\rangle }\leq\pbra{\frac1m+\frac{k}{2^{m/2-2}}}\left\langle f,f\right\rangle.
    \end{align*}
\end{lemma}
\begin{proof}
    We bound $\abs{ \left\langle f,R_{m,m-1,k}f\right\rangle }\geq \abs{ \left\langle f,(R_{m,m-1,k}-R_{m,m,k})f\right\rangle }$. This inequality is true because $R_{m,m,k}$ and $R_{m,m-1,k}-R_{m,m,k}$ are PSD by. Suppose first that $|S|=1$ so that $S=\{a\}$ for some $a\in[m]$. Let $X\in\widetilde{\mathcal{I}}_S$. Then for every $W\subseteq[m]$ with $|W|=m-1$ and $a\in W$ we have
    \begin{align*}
        \Pr\sbra{X\to_{R_{m,W,k}} \widetilde{\mathcal{I}}_S}
        \leq& \sum_{i,j\in[k]}\Pr_{\mathbf{Y}\sim \mathcal{D}_X^{m,W,k}}\sbra{\Delta\pbra{\mathbf{Y}^i,\mathbf{Y}^j}=\{a\}}
        \leq  \frac{k^2}{2^{m-1}}.
    \end{align*}
    This is because $a\in W$, and so it must be the case that for all $i\neq j\in[k]$ we have $X^i_{[m]\setminus\{b\}}\neq X^j_{[m]\setminus\{b\}}$. Otherwise we would have $\widetilde{\mathcal I}_{\{a,b\}}$, contradicting that $X\in \widetilde{\mathcal I}_{\{a\}}$. Using this bound, we have that
    \begin{align*}
        &\abs{\left\langle f,(R_{m,m-1,k}-R_{m,m,k})f\right\rangle }\\
        \leq&\frac1m\sum_{W\subseteq [m],|W|=m-1}\abs{\left\langle f,R_{m,W,k}f\right\rangle }+\abs{\left\langle f,R_{m,m,k}f\right\rangle }\\
        \leq&\frac1m\abs{\left\langle f,R_{m,[m]\setminus\{a\},k}f\right\rangle }+ \frac1m\sum_{W\subseteq [m],|W|=m-1,a\in W}\abs{\left\langle f,R_{m,W,k}f\right\rangle }+\abs{\left\langle f,R_{m,m,k}f\right\rangle }\\
        \leq&\frac1m\langle f,f\rangle+  \frac{k}{2^{m/2-1}}\langle f,f\rangle+ \frac{k}{2^{m/2}}\langle f,f\rangle.
    \end{align*}
    The last inequality follows from the \Cref{lem:escape probs} and our previous calculation, while noticing that any $R_{m,W,k}$ has the uniform distribution over $\{\pm1\}^{mk}$ as a stationary distribution. A similar calculation shows that $\Pr\sbra{X\to_{R_{m,m,k}} \widetilde{\mathcal{I}}_S}\leq \frac{k^2}{2^{m}}$, which is used to bound the third term $\abs{\left\langle f,R_{m,m,k}f\right\rangle }$ by \Cref{lem:escape probs}. This completes the proof for the case $S=\{a\}$.

    Now assume that $|S|=2$ so that $S=\{a,b\}$ and $X\in\widetilde{\mathcal{I}}_S$. Fix $W\subseteq[m]$ with $|W|=m-1$. Assume that $a\not\in W$. Then
    \begin{align*}
        \Pr\sbra{X\to_{R_{m,W,k}} \widetilde{\mathcal{I}}_S}\leq&\sum_{i,j\in[k]}\Pr_{\mathbf{Y}\sim\mathcal{D}^{m,W,k}_X}\sbra{\Delta\pbra{\mathbf{Y}^i,\mathbf{Y}^j}=b}\leq\sum_{i,j\in[k]}\frac{1}{2^{m-1}}\leq \frac{k^2}{2^{m-1}}.
    \end{align*}
    This is because if $\Delta(X^i,X^j)=\{a\}$ then $\Delta(\mathbf{Y}^i,\mathbf{Y}^j)=\{a\}\neq \{b\}$ and otherwise $X^i_{[m]\setminus\{a\}}=X^i_W \neq X^j_W= X^j_{[m]\setminus\{a\}}$. Using this we find that
    \begin{align*}
        \Pr\sbra{X\to_{R_{m,m-1,k}} \widetilde{\mathcal{I}}_S}=&\frac1m\sum_{W\subseteq[m],|S|=m-1}\Pr\sbra{X\to_{R_{m,W,k}} \widetilde{\mathcal{I}}_S}\leq \frac{k^2}{2^{m-1}}.
    \end{align*}
    If $b\in W$ or $a,b\not\in W$ a similar proof shows the same bound. Using the bound $\Pr\sbra{X\to_{R_{m,m,k}} \widetilde{\mathcal{I}}_T}\leq \frac{k^2}{2^{m-1}}$ completes the proof:
    \begin{align*}
        &\abs{\left\langle f,(R_{m,m-1,k}-R_{m,m,k})f\right\rangle}
        \leq\abs{\left\langle f,R_{m,m-1,k}f\right\rangle}+\abs{\left\langle f,R_{m,m,k}f\right\rangle}
        \leq\frac{k}{2^{m/2-1}}\langle f,f\rangle+ \frac{k}{2^{m/2-1}}\langle f,f\rangle.
    \end{align*}
    The second inequality is an application of \Cref{lem:escape probs}.
\end{proof}

\begin{lemma}\label{lem:cross terms B=1 to B=1}
    Let $S\neq T\subseteq[m]$ be such that $|S|,|T|\leq 2$. If $k\geq 2$ and $f:\{\pm1\}^{mk}\to\R$ is supported on $\widetilde{\mathcal{I}}_S$ and $g:\{\pm1\}^{mk}\to\R$ is supported on $\widetilde{\mathcal{I}}_T$ then 
    \begin{align*}
       \abs{ \left\langle f,(R_{m,m-1,k}-R_{m,m,k})g\right\rangle }\leq\frac{k}{2^{m/2-2}}\norm{f}_2\norm{g}_2.
    \end{align*}
\end{lemma}
\begin{proof}
    Let $X\in\widetilde{\mathcal{I}}_S$. Let $a\in T\setminus S$ be such that there does not exist $i,j\in[k]$ such that $X\in \mathcal{I}_{\{a\}}$. Then 
    \begin{align*}
        &\Pr\sbra{X\to_{R_{m,m-1,k}} \widetilde{\mathcal{I}}_T}
        \leq \Pr\sbra{X\to_{R_{m,m-1,k}} \widetilde{\mathcal{I}}_{\{a\}}}
        \leq \Pr\sbra{X\to_{R_{m,m-1,k}} {\mathcal{I}}_{\{a\}}}\\
        \leq & \sum_{i,j\in[k]}\Pr_{\mathbf{Y}\sim \mathcal{D}_X^{m,m-1,k}}\sbra{\Delta\pbra{\mathbf{Y}^i,\mathbf{Y}^j}=\{a\}}
        \leq \frac{k^2}{2^{m-1}}.\tag{$X\not\in\mathcal{I}_{\{a\}}$}
    \end{align*}
    Then applying \Cref{lem:escape probs} gives
    \begin{align*}
        \abs{ \left\langle f,R_{m,m-1,k}g\right\rangle }\leq\frac{k}{2^{m/2-1}}\norm{f}_2\norm{g}_2.
    \end{align*}
    A similar calculation shows that $\Pr\sbra{X\to_{R_{m,m,k}} \widetilde{\mathcal{I}}}\leq \frac{mk^2}{2^{m-1}}$. Then \Cref{lem:escape probs} gives
    \begin{align*}
        \abs{ \left\langle f,R_{m,m,k}g\right\rangle }\leq\frac{k}{2^{m/2-1}}\norm{f}_2\norm{g}_2.
    \end{align*}
    Applying the triangle inequality completes the proof.
\end{proof}

\begin{corollary}\label{cor:f1-f1}
    Assume $k\geq 2$. Let $f:\{\pm1\}^{mk}\to\R$ be supported on $B_{=1}$. Then 
    \begin{align*}
        \abs{\left\langle f,(R_{m,m-1,k}-R_{m,m,k})f\right\rangle }\leq \pbra{\frac1m+\frac{m^5k}{2^{m/2-2}}}\left\langle f,f\right\rangle.
    \end{align*}
\end{corollary}
\begin{proof}
    Write $f=\sum_{S\subseteq[m]:|S|\leq 2}f_{S}$ where each $f_S$ is supported on $\mathcal{I}_S$. Then
    \begin{align*}
        & \abs{\left\langle f,(R_{m,m-1,k}-R_{m,m,k})f\right\rangle}\\
        \leq &\abs{\left\langle f,R_{m,m-1,k}f\right\rangle} \tag{$R_{m,m,k}$ and $R_{m,m-1,k}-R_{m,m,k}$ both PSD}\\
        = &\sum_{S,T\subseteq [m]:|S|,|T|\leq 2}\abs{\left\langle f_S,R_{m,m-1,k}f_T\right\rangle}\\
        \leq &\sum_{S\subseteq [m]:|S|\leq 2}\abs{\left\langle f_S,R_{m,m-1,k}f_S\right\rangle}+ \sum_{S\neq T\subseteq [m]:|S|,|T|\leq 2}\abs{\left\langle f_S,R_{m,m-1,k}f_T\right\rangle} \\
        \leq & \pbra{\frac1m+\frac{k}{2^{m/2-2}}}\sum_{S\subseteq [m]:|S|\leq 2}\left\langle f_S,f_S\right\rangle + \frac{k}{2^{m/2-2}}\sum_{S\neq T\subseteq[m]:|S|,|T|\leq 2}\norm{f_S}_2\norm{f_T}_2\tag{\Cref{lem:square terms B=1}, \Cref{lem:cross terms B=1 to B=1}}\\
        = & \pbra{\frac1m+\frac{k}{2^{m/2-2}}}\left\langle f,f\right\rangle + \frac{k}{2^{m/2-2}}\sum_{S\neq T\subseteq[m]:|S|,|T|\leq 2}\norm{f_S}_2\norm{f_T}_2 \\
        \leq & \pbra{\frac1m+\frac{k}{2^{m/2-2}}}\left\langle f,f\right\rangle + \frac{k}{2^{m/2-2}}\sum_{S\neq T\subseteq[m]:|S|,|T|\leq 2}\left\langle f,f\right\rangle\\
        \leq & \pbra{\frac1m+\frac{k}{2^{m/2-2}}}\left\langle f,f\right\rangle + \frac{m^5k^2}{2^{m/2-2}}\left\langle f,f\right\rangle\\
        \leq & \pbra{\frac1m+\frac{m^5k}{2^{m/2-2}}}\left\langle f,f\right\rangle .
    \end{align*}
    Note that we can apply \Cref{lem:square terms B=1} and \Cref{lem:cross terms B=1 to B=1} because $k\geq 2$.
\end{proof}
\subsection{Cross Terms}\label{sec:cross terms}
We can use the same idea to bound the contributions from the cross terms.
\begin{lemma}\label{lem:cross terms B=1 to B=2}
    Let $f_1:\{\pm1\}^{mk}\to\R$ be supported on $B_{=1}$ and let $f_2:\{\pm1\}^{mk}\to\R$ be supported on $B_{\geq 2}$. Then 
    \begin{align*}
        \abs{\left\langle f_1,(R_{m,m-1,k}-R_{m,m,k})f_2\right\rangle} \leq \frac{\sqrt{m}k}{2^{m/2-2}}\norm{f_1}_2\norm{f_2}_2.
    \end{align*}
\end{lemma}
\begin{proof}
    For $X\in B_{\geq 2}$ we have $\Pr\sbra{X\to_{R_{m,m-1,k}} B_{=1}}\leq \frac{k^2m}{2^{m-1}}$. Apply \Cref{lem:escape probs} to find that
    \begin{align*}
        \abs{\left\langle f_1,R_{m,m-1,k}f_2\right\rangle} \leq \frac{\sqrt{m}k}{2^{m/2-1}}\norm{f_1}_2\norm{f_2}_2.
    \end{align*}
    For $X\in B_{\geq 2}$ we have $\Pr\sbra{X\to_{R_{m,m,k}} B_{=1}}\leq \frac{k^2m}{2^{m-1}}$. Apply \Cref{lem:escape probs} to find that
    \begin{align*}
        \abs{\left\langle f_1,R_{m,m,k}f_2\right\rangle} \leq \frac{\sqrt{m}k}{2^{m/2-1}}\norm{f_1}_2\norm{f_2}_2.
    \end{align*}
    Applying the triangle inequality completes the proof.
\end{proof}

\begin{lemma}\label{lem:cross terms B0}
    Let $f_0:\{\pm1\}^{mk}\to\R$ be supported on $B_{=0}$ and let $f_1:\{\pm1\}^{mk}\to\R$ be supported on $B_{=1}\cup B_{\geq 2}$. Then 
    \begin{align*}
        \abs{\left\langle f_0,(R_{m,m-1,k}-R_{m,m,k})f_1\right\rangle} = 0.
    \end{align*}
\end{lemma}
\begin{proof}
    For $X\in B_{\geq 2}\cup B_{=1}$ we have $\Pr\sbra{X\to_{R_{m,m-1,k}} B_{=0}}=\Pr\sbra{X\to_{R_{m,m,k}} B_{=0}}=0$. Apply \Cref{lem:escape probs} to bound 
    \begin{align*}
        &\abs{\left\langle f_0,(R_{m,m-1,k}-R_{m,m,k})f_1\right\rangle} \leq \abs{\left\langle f_0,R_{m,m-1,k}f_1\right\rangle} + \abs{\left\langle f_0,R_{m,m,k}f_1\right\rangle} \leq0.\qedhere
    \end{align*}
\end{proof}

\subsection{A Hybrid Argument for $f$ Supported on $B_{\geq 2}$}\label{sec:B>=2}

In this section we bound the square terms $\left\langle f_2,(R_{m,m-1,k}-R_{m,m,k})f_2\right\rangle$ for $f_2$ supported on $B_{\geq 2}$. As mentioned in \Cref{sec:overview}, our key idea is that when $k$ is small compared to $m$, the fraction of $X\in\{\pm1\}^{mk}$ with two identical columns is so small that applying the noise by randomly permuting the rows is almost the same as randomly replacing the rows with completely random rows. That is, sampling without replacement resembles sampling with replacement closely.

This observation allows us to pass from the random walk described by $R_{m,m-1,k}$ to a different random walk described by a nicer noise model described by operators we will call $Q_{m,m-1,k}$. The key tool we use is the bound given by \Cref{lem:TV distance bound} for relating total-variation distances between Markov chain transition probabilities to a more linear-algebraic notion of closeness, stated in terms of their transition matrices. Fourier-analytic techniques will then be useful to bound the spectral norm of these $Q$-operators.

\begin{definition}
    We define four random walk operators\footnotemark\footnotetext{$R_{m,m-1,k}$ and $R_{m,m,k}$ have already been defined, but we define them again here for ease of comparison.} $R_{m,m-1,k}$, $Q_{m,m-1,k}$, $R_{m,m,k}$, and $Q_{m,m,k}$ on $\R^{\{\pm1\}^{mk}}$. 
    \begin{itemize}     
        \item Define
        \begin{align*}
            (R_{m,m-1,k}f)(X)&=\underset{\mathbf{Y} \sim\mathcal{D}^{m,m-1,k}_{X}}{\E}\sbra{f(\mathbf{Y})}.
        \end{align*}
           
        \item To define $Q_{m,m-1,k}$, for any $X\in \cbra{\pm1}^{mk}$ we define the distribution $\mathcal{C}^{m,m-1,k}_{X}$ as follows. To sample $\mathbf{Y}$ from $\mathcal{C}^{m,m-1,k}_X$, we sample $\mathbf{a}\in[m]$ uniformly randomly and set $\mathbf{Y}^i_{\mathbf{a}}=X^i_{\mathbf{a}}$ for all $i\in[k]$. Then set $\mathbf{Y}^i_{a}$ uniformly randomly for $a\neq \mathbf{a}$. Then
        \begin{align*}
            (Q_{m,m-1,k}f)(X)&=\underset{\mathbf{Y}\sim\mathcal{C}^{m,m-1,k}_X}{\E}\sbra{f(\mathbf{Y})}.
        \end{align*}
        
        \item Define
        \begin{align*}
            (R_{m,m,k}f)(X)&=\Ex_{\mathbf{Y}\sim\mathcal{D}^{m,m,k}_X}\sbra{f(\mathbf Y)}.
        \end{align*}

        \item The operator $Q_{m,m,k}$ is defined by setting for each $f:\cbra{\pm1}^{mk}\to\R$
        \begin{align*}
            (Q_{m,m,k}f)(X)&=\Ex_{\mathbf{Y}\sim \mathrm{Unif}\pbra{\cbra{\pm1}^{mk}}}\sbra{f(\mathbf{Y})}.
        \end{align*}
    \end{itemize}
\end{definition}
\begin{fact}\label{fact:Q self-adjoint}
    The matrices $Q_{m,m-1,k}$ and $Q_{m,m,k}$ are self-adjoint and PSD for any $m,k$.
\end{fact}

We intend to show that $\abs{\left\langle f,(R_{m,m-1,k}-R_{m,m,k})f\right\rangle}$ is small for $f$ supported on $B_{\geq 2}$. We do so by a hybrid argument. In the following inequality, the first (and last) term on the RHS will be bounded by a simple bound on the total variation distance between the distribution $\mathcal{C}_X^{m,m-1,k}$ and $\mathcal{D}_X^{m,m-1,k}$ ($\mathcal{C}_X^{m,m,k}$ and $\mathcal{D}_X^{m,m,k}$). As mentioned before, the second term on the RHS will be bounded using Fourier analysis.
\begin{align*}\label{eq:hybrid}
    &\abs{\left\langle f,(R_{m,m-1,k}-R_{m,m,k})f\right\rangle} \\
        \leq& \abs{\left\langle f,\pbra{R_{m,m-1,k}-Q_{m,m-1,k}}f\right\rangle}+ \abs{\left\langle f,\pbra{Q_{m,m,k}-Q_{m,m-1,k}}f\right\rangle}+\abs{\left\langle f,\pbra{R_{m,m,k}-Q_{m,m,k}}f\right\rangle}.
\end{align*}

\subsubsection{The First Hybrid: $R_{m,m-1,k}$ to $Q_{m,m-1,k}$}

\begin{lemma}\label{lem:hybrid 1}
    Assume that $k\leq 2^{m/3}$. For any $f:\{\pm1\}^{mk}\to\R$ supported on $B_{\geq 2}$ and we have that
    \begin{align*}
        \abs{\left\langle f,(R_{m,m-1,k}-Q_{m,m-1,k})f\right\rangle}\leq \frac{k^2}{2^{m-1}}\left\langle f,f\right\rangle.
    \end{align*}
\end{lemma}
\begin{proof}
We directly compute (using the appropriate definitions of $p_0$ and $p_1$ given in \Cref{lem:TV distance bound}):
\begin{align*}
        &\abs{\left\langle f,(R_{m,m-1,k}-Q_{m,m-1,k})f\right\rangle}\\
        \leq & \sum_{X\in B_{\geq 2}}f(X)^2\sum_{Y\in B_{\geq 2}}\abs{\Pr\sbra{X\to_{R_{m,m-1,k}} Y}-\Pr\sbra{X\to_{Q_{m,m-1,k}}Y}}\tag{\Cref{lem:TV distance bound} + self-adjointness (\Cref{fact:self-adjoint}, \Cref{fact:Q self-adjoint})}\\
        \leq & \frac{k^2}{2^{m-1}}\sum_{X\in B_{\geq 2}}f(X)^2\tag{\Cref{eq:D0 D1 TV} below}\\
        \leq & \frac{k^2}{2^{m-1}}\left\langle f,f\right\rangle.
    \end{align*}
    It suffices to establish \Cref{eq:D0 D1 TV}. Assume that $X\in B_{\geq 2}$. Then
    \begin{align*}
        &\sum_{Y\in B_{\geq 2}}\abs{\Pr\sbra{X\to_{R_{m,m-1,k}} Y}-\Pr\sbra{X\to_{Q_{m,m-1,k}} Y}}\\
        =&\frac1m\sum_{a\in[m]}\sum_{Y\in B_{\geq 2}}\abs{\Pr_{\mathbf{Y}\sim\mathcal{D}^{m,m-1,k}_X}\sbra{\mathbf{Y}=Y|a\text{ fixed}}-\Pr_{\mathbf{Y}\sim\mathcal{C}^{m,m-1,k}_X}\sbra{\mathbf{Y}=Y|a\text{ fixed}}}\\
        =&\frac1m\sum_{a\in[m]}\sum_{Y\in B_{\geq 2},Y_a=X_a}\abs{\Pr_{\mathbf{Y}\sim\mathcal{D}^{m,m-1,k}_X}\sbra{\mathbf{Y}=Y|a\text{ fixed}}-\Pr_{\mathbf{Y}\sim\mathcal{C}^{m,m-1,k}_X}\sbra{\mathbf{Y}=Y|a\text{ fixed}}}\\
        =&\frac1{m}\sum_{a\in[m]}\sum_{Y\in B_{\geq 2},Y_a=X_a}\abs{ \prod_{i=0}^{k-1}\frac{1}{2^{m-1}-i}- \frac1{2^{(m-1)k}}}\tag{$k\leq 2^{m/3}\leq 2^m-2$ and $X\in B_{\geq 2}$}\\
        \leq &\frac1{m}\sum_{a\in[m]}\sum_{Y\in B_{\geq 2},Y_a=X_a}\abs{ \frac1{2^{(m-1)k}}\pbra{\prod_{i=0}^{k-1}\frac{2^{m-1}}{2^{m-1}-i}- 1}}\\
        \leq &\frac1{m}\sum_{a\in[m]}\sum_{Y\in B_{\geq 2},Y_a=X_a}\abs{ \frac1{2^{(m-1)k}}\pbra{1+\frac{k^2}{2^{m-1}}- 1}}\tag{$k\leq 2^{m/3}$, \Cref{fact:k^2}}\\
        = &\frac1{m}\sum_{a\in[m]}\sum_{Y\in B_{\geq 2},Y_a=X_a} \frac{k^2}{2^{m-1}2^{(m-1)k}}\\
        = &\sum_{Y\in B_{\geq 2},Y_1=X_1} \frac{k^2}{2^{m-1}2^{(m-1)k}}\\
        \leq & \frac{2^{mk}}{2^k}\cdot \frac{k^2}{2^{m-1}2^{(m-1)k}}\\
        =& \frac{k^2}{2^{m-1}}.\numberthis\label{eq:D0 D1 TV}
    \end{align*}
    Note that our computations for $\Pr_{\mathbf{Y}\sim\mathcal{D}^{m,m-1,k}_X}\sbra{\mathbf{Y}=Y|a\text{ fixed}}$ and $\Pr_{\mathbf{Y}\sim\mathcal{C}^{m,m-1,k}_X}\sbra{\mathbf{Y}=Y|a\text{ fixed}}$ relied on the fact that $Y\in B_{\geq 2}$.
\end{proof}

\subsubsection{The Second Hybrid: $Q_{m,m-1,k}$ to $Q_{m,m,k}$}
We use Fourier analysis to analyze the spectrum of the operator $Q_{m,m-1,k}-Q_{m,m,k}$, which will prove that $\norm{Q_{m,m-1,k}-Q_{m,m,k}}_{\mathrm{op}}$ is small. We use the Fourier characters as an eigenbasis for $Q_{m,m-1,k}-Q_{m,m,k}$.
\begin{fact}\label{fact:characters under A1-R1}
    Fix $S_1,\dots,S_k\subseteq[n]$. Then
    \begin{align*}
        \pbra{(Q_{m,m-1,k}-Q_{m,m,k})\chi_{S_1,\dots,S_k}}=&
        \begin{cases}
            \frac1m\chi_{S_1,\dots,S_k} &\text{ if }S_1\cup \dots\cup S_k=\{a\} \text{ for some $a\in [m]$}.\\
            0 &\text{ otherwise.}
        \end{cases}
    \end{align*}
\end{fact}
\begin{proof}
    If $\abs{\bigcup_i S_i}=1$ then $S_1=\dots=S_k=\emptyset$ and it is clear that 
    \begin{align*}
        Q_{m,m-1,k}\chi_{S_1,\dots,S_k}=1=Q_{m,m,k}\chi_{S_1,\dots,S_k}.
    \end{align*}
    Now assume $\abs{\bigcup_i S_i}\geq 2$. Then for any $X=(X^1,\dots,X^k)\in \{\pm1\}^{mk}$, 
    \begin{align*}
        &Q_{m,m-1,k}\chi_{S_1,\dots,S_k}(X^1,\dots,X^k)\\
        =&\Ex_{(\mathbf{Y}^1,\dots,\mathbf{Y}^k)\sim\mathcal{C}^{m,m-1,k}_{X}}\sbra{\chi_{S_1,\dots,S_k}(\mathbf{Y}^1,\dots,\mathbf{Y}^k)}\\
        =&\frac1m\sum_{a'\in[m]}\Ex_{\substack{(\mathbf{Y}^1,\dots,\mathbf{Y}^k)\sim\mathcal{C}^{m,m-1,k}_{X}\\\mathbf{a}=a'}}\sbra{\chi_{S_1,\dots,S_k}(\mathbf{Y}^1,\dots,\mathbf{Y}^k)}\\
        =&\frac1m\sum_{a'\in[m]}\Ex_{\substack{(\mathbf{Y}^1,\dots,\mathbf{Y}^k)\sim\mathcal{C}^{m,m-1,k}_{X}\\\mathbf{a}=a'}}\sbra{\prod_{i\in[k]}\prod_{a\in S_i} \mathbf{Y}^i_a}\\
        =&\frac1m\sum_{a'\in[m]}\Ex_{\substack{(\mathbf{Y}^1,\dots,\mathbf{Y}^k)\sim\mathcal{C}^{m,m-1,k}_{X}\\\mathbf{a}=a'}}\sbra{\pbra{\prod_{\substack{i\in[k]\\a'\in S_i}}\mathbf{Y}^i_{a'}}\pbra{\prod_{i\in[k]}\prod_{a\in S_i\setminus\{a'\}} \mathbf{Y}^i_a }}\\
        =&\frac1m\sum_{a'\in[m]}\pbra{\prod_{\substack{i\in[k]\\a'\in S_i}}\mathbf{Y}^i_{a'}}\Ex_{\substack{(\mathbf{Y}^1,\dots,\mathbf{Y}^k)\sim\mathcal{C}^{m,m-1,k}_{X}\\\mathbf{a}=a'}}\sbra{\prod_{i\in[k]}\prod_{a\in S_i\setminus\{a'\}} \mathbf{Y}^i_a }\\
        =&\frac1m\sum_{a'\in[m]}\pbra{\prod_{\substack{i\in[k]\\a'\in S_i}}\mathbf{Y}^i_{a'}}\Ex_{(\mathbf{Y}^1,\dots,\mathbf{Y}^k)\in\pbra{\{\pm1\}^n}^{[k]\setminus\{a'\}}}\sbra{\chi_{S_1\setminus\{a'\},\dots,S_k\setminus\{a'\}}(\mathbf{Y}^1,\dots,\mathbf{Y}^k) }\\
        =&0.\tag{since at least one of the $S_i\setminus\{a'\}$ is nonempty}
    \end{align*}
    It is easy to see that $Q_{m,m,k}\chi_{S_1,\dots,S_k}=0$, since $S_1\cup\dots\cup S_k\neq\emptyset$.

    Now assume that $\bigcup_i S_i=\{a'\}$ for some $a'\in[m]$. Then it is clear that $Q_{m,m,k}\chi_{S_1,\dots,S_k}=0$, since $S_1\cup\dots\cup S_k\neq\emptyset$. We now compute
    \begin{align*}
        &Q_{m,m-1,k}\chi_{S_1,\dots,S_k}(X^1,\dots,X^k)\\
        =&\Ex_{(\mathbf{Y}^1,\dots,\mathbf{Y}^k)\sim\mathcal{C}^{m,m-1,k}_{X}}\sbra{\chi_{S_1,\dots,S_k}(\mathbf{Y}^1,\dots,\mathbf{Y}^k)}\\
        =&\frac1m\sum_{a\in[m]}\Ex_{\substack{(\mathbf{Y}^1,\dots,\mathbf{Y}^k)\sim\mathcal{C}^{m,m-1,k}_{X}\\\mathbf{a}=a}}\sbra{\prod_{\substack{i\in[k]\\S_i=\{a'\}}}\mathbf{Y}^k_{a'}}\\
        =&\frac1m\Ex_{\substack{(\mathbf{Y}^1,\dots,\mathbf{Y}^k)\sim\mathcal{C}^{m,m-1,k}_{X}\\\mathbf{a}=a'}}\sbra{\prod_{\substack{i\in[k]\\S_i=\{a'\}}}\mathbf{Y}^k_{a'}}\tag{If $\mathbf{a}\neq a'$ then the $\mathbf{Y}^i_{a'}$ are uniformly random.}\\
        =&\frac1m\Ex_{\substack{(\mathbf{Y}^1,\dots,\mathbf{Y}^k)\sim\mathcal{C}^{m,m-1,k}_{X}\\\mathbf{a}=a'}}\sbra{\prod_{\substack{i\in[k]\\S_i=\{a'\}}}X^k_{a'}}\\
        =&\frac1m\chi_{S_1,\dots,S_k}(X^1,\dots,X^k).
    \end{align*}
    Therefore $\pbra{Q_{m,m-1,k}-Q_{m,m,k}}\chi_{S_1,\dots,S_k}=\frac1m\chi_{S_1,\dots,S_k}(X^1,\dots,X^k)-0=\frac1m\chi_{S_1,\dots,S_k}(X^1,\dots,X^k)$.
\end{proof}

\begin{corollary}\label{cor:hybrid 2}
    Let $f:\{\pm1\}^{mk}\to\R$. Then
    \begin{align*}
        \abs{\left\langle f,(Q_{m,m-1,k}-Q_{m,m,k})f\right\rangle} \leq \frac1m \left\langle f,f\right\rangle.
    \end{align*}
\end{corollary}
\begin{proof}
    The $\chi_{S_1,\dots,S_k}$ form an orthonormal eigenbasis (\Cref{fact:fourier characters}) for $Q_{m,m-1,k}-Q_{m,m,k}$, and by \Cref{fact:characters under A1-R1} each basis element has eigenvalue with absolute value at most $\frac1m$.
\end{proof}

\subsubsection{The Third Hybrid: $Q_{m,m,k}$ to $R_{m,m,k}$}

\begin{lemma}\label{lem:hybrid 3}
    Assume that $k\leq 2^{m/3}$. For any $f:\{\pm1\}^{mk}\to\R$ supported on $B_{\geq 2}$ we have
    \begin{align*}
        \abs{\left\langle f,(R_{m,m,k}-Q_{m,m,k})f\right\rangle} \leq \frac{k^2}{2^{m}}\norm{f}_2^2 .
    \end{align*}
\end{lemma}
\begin{proof}
    As in the proof of \Cref{lem:hybrid 1} we find that (with the appropriate definitions of $p_0$ and $p_1$ for use of \Cref{lem:TV distance bound},
    \begin{align*}
       &\abs{\left\langle f,(R_{m,m,k}-Q_{m,m,k})f\right\rangle}\\
       \leq & \sum_{X\in B_{\geq 2}}f(X)^2\sum_{Y\in B_{\geq 2}}\abs{\Pr\sbra{X\to_{R_{m,m,k}} Y}-\Pr\sbra{X\to_{Q_{m,m,k}} Y}}\tag{\Cref{lem:TV distance bound} + self-adjointness (\Cref{fact:self-adjoint}, \Cref{fact:Q self-adjoint})}\\
        \leq & \frac{k^2}{2^m}\sum_{X\in B_{\geq 2}}f(X)^2\tag{\Cref{eq:TV for hybrid 3} below}\\
        \leq & \frac{k^2}{2^m}\left\langle f,f\right\rangle.
    \end{align*}
    It remains to prove \Cref{eq:TV for hybrid 3}:
    \begin{align*}
        &\sum_{Y\in B_{\geq 2}}\abs{\Pr\sbra{X\to_{R_{m,m,k}} Y}-\Pr\sbra{X\to_{Q_{m,m,k}} Y}}\\
        =&\sum_{Y\in B_{\geq 2}}\abs{\Pr_{\mathbf{Y}\sim\mathcal{D}^{m,m,k}_X}\sbra{\mathbf{Y}=Y}-\Pr_{\mathbf{Y}\sim\mathcal{C}^{m,m,k}_X}\sbra{\mathbf{Y}=Y}}\\
        =&\sum_{Y\in B_{\geq 2}}\abs{ \prod_{i=0}^{k-1}\frac{1}{2^{m}-i}- \frac1{2^{mk}}}\tag{$k\leq 2^{m/3}\leq 2^m-2$ and $X\in B_{\geq 2}$}\\
        \leq &\sum_{Y\in B_{\geq 2}}\abs{ \frac1{2^{mk}}\pbra{\prod_{i=0}^{k-1}\frac{2^{m}}{2^{m}-i}- 1}}\\
        \leq &\sum_{Y\in B_{\geq 2}}\abs{ \frac1{2^{mk}}\pbra{1+\frac{k^2}{2^m}- 1}}\tag{$k\leq 2^{m/3}$, \Cref{fact:k^2}}\\
        = &\sum_{Y\in B_{\geq 2}} \frac{k^2}{2^m2^{mk}}\\
        \leq &{2^{mk}}\cdot \frac{k^2}{2^m2^{mk}}\\
        =& \frac{k^2}{2^m}.\numberthis\label{eq:TV for hybrid 3}
    \end{align*}
    Note that our computations relied on the fact that $Y\in B_{\geq 2}$.
\end{proof}

\subsubsection{Putting Hybrids Together}
\begin{proposition}\label{prop:B>=2 square term}
    Assume that $k\leq 2^{m/3}$. For any $f:\{\pm1\}^{mk}\to\R$ supported on $B_{\geq 2}$, we have
    \begin{align*}
        \abs{\left\langle f,(R_{m,m-1,k}-R_{m,m,k})f\right\rangle} \leq \pbra{\frac1m+\frac{k^2}{2^{m/2}}}\left\langle f,f\right\rangle
    \end{align*}
\end{proposition}
\begin{proof}
    By the triangle inequality,
    \begin{align*}
        &\abs{\left\langle f,(R_{m,m-1,k}-R_{m,m,k})f\right\rangle} \\
        \leq& \abs{\left\langle f,\pbra{R_{m,m-1,k}-Q_{m,m-1,k}}f\right\rangle}+ \abs{\left\langle f,\pbra{Q_{m,m,k}-Q_{m,m-1,k}}f\right\rangle}+\abs{\left\langle f,\pbra{R_{m,m,k}-Q_{m,m,k}}f\right\rangle}\\
        \leq & \frac{k^2}{2^{m-1}}\left\langle f,f\right\rangle+\frac1m\left\langle f,f\right\rangle+\frac{k^2}{2^{m}}\left\langle f,f\right\rangle \tag{\Cref{lem:hybrid 1}, \Cref{cor:hybrid 2}, \Cref{lem:hybrid 3}}\\
        \leq  &\pbra{\frac1m+\frac{k^2}{2^{m-2}}}\left\langle f,f\right\rangle .
    \end{align*}
    Note we can apply \Cref{lem:hybrid 1} and \Cref{lem:hybrid 3} because $f$ is supported on $B_{\geq 2}$.
\end{proof}

\section{Proof of \Cref{thm:nachtergaele hypothesis}}\label{sec:nachtergaele hypothesis}

In this section let $m$ and $k\leq 2^m-2$ be fixed positive integers. As in the proof of \Cref{thm:small k} (via its restatement \Cref{thm:small k internal}) we will break $\R^{\{\pm1\}^{mk}}$ into orthogonal components and bound the contributions of each of these cross terms to evaluation on the quadratic form given by $R_{m,[m-\ell-1,m],k}\pbra{R_{m,[m-1],k}-R_{m,[m],k}}$. 

Define the following subsets of $\{\pm1\}^{mk}$:
\begin{itemize}
    \item $B_{=0}=\cbra{X\in\{\pm1\}^{mk}:\exists i\neq j\in[k], X^i=X^j}$.\footnotemark\footnotetext{$B_{=0}$ was already defined in \Cref{sec:small k} but we define it again here for convenience.} 
    \item $B_{\geq1}^{\mathrm{coll}}=\cbra{X\in\{\pm1\}^{mk}\setminus B_{=0}:\exists i\neq j\in[k], X^i_{[m-\ell-1,m-1]}=X^j_{[m-\ell-1,m-1]}}$.
    \item $B_{\geq1}^{\mathrm{safe}}=\cbra{X\in\{\pm1\}^{mk}:\forall i\neq j\in[k], X^i_{[m-\ell-1,m-1]}\neq X^j_{[m-\ell-1,m-1]}}$.
\end{itemize}
Note that these sets form a partition of $\{\pm1\}^{mk}$.

As in the proof of \Cref{thm:small k}, the contribution from the parts of functions supported on $B_{=0}$ will be bounded by induction. The component $B_{\geq1}^{\mathrm{safe}}$ will play the role that $B_{\geq2}$ did: the part of the domain on which the noise model induced by random permutations (sampling without replacement) to the bits in $[m-\ell-1,m-1]$ is close to the noise model induced by completely random replacement of bits (sampling with replacement). $B_{\geq1}^{\mathrm{coll}}$ is the component on which these two noise models are not similar, but as in the case of $B_{=1}$, this set will already show good expansion.

The following is equivalent to \Cref{thm:nachtergaele hypothesis} because for any $f:\{\pm1\}^{mk}\to \R$ we have that
\begin{align*}
    \norm{R_{m,[m-\ell-1,m],k}\pbra{R_{m,[m-1],k}-R_{m,[m],k}}}_{\mathrm{op}}=&
    \norm{\pbra{R_{m,[m-\ell-1,m],k}\pbra{R_{m,[m-1],k}-R_{m,[m],k}}}^*}_{\mathrm{op}}\\
    =&\norm{\pbra{R_{m,[m-1],k}-R_{m,[m],k}}R_{m,[m-\ell-1,m],k}}_{\mathrm{op}}.
\end{align*}
To bound this quantity, we can use that
\begin{align*}
    &\norm{\pbra{R_{m,[m-1],k}-R_{m,[m],k}}R_{m,[m-\ell-1,m],k}f}^2_2\\
    =&\left\langle \pbra{R_{m,[m-1],k}-R_{m,[m],k}}R_{m,[m-\ell-1,m],k}f,\pbra{R_{m,[m-1],k}-R_{m,[m],k}}R_{m,[m-\ell-1,m],k}f\right\rangle\\
    =&\left\langle R_{m,[m-\ell-1,m],k}\pbra{R_{m,[m-1],k}-R_{m,[m],k}}\pbra{R_{m,[m-\ell-1,m],k}\pbra{R_{m,[m-1],k}-R_{m,[m],k}}}^*f,f\right\rangle\\
    =&\left\langle R_{m,[m-\ell-1,m],k}\pbra{R_{m,[m-1],k}-R_{m,[m],k}}^2R_{m,[m-\ell-1,m],k}f,f\right\rangle\\
    =&\left\langle R_{m,[m-\ell-1,m],k}\pbra{R_{m,[m-1],k}-R_{m,[m],k}-R_{m,[m],k}+R_{m,[m],k}}R_{m,[m-\ell-1,m],k}f,f\right\rangle\\
    =&\left\langle f,R_{m,[m-\ell-1,m],k}\pbra{R_{m,[m-1],k}-R_{m,[m],k}}R_{m,[m-\ell-1,m],k}f\right\rangle.
\end{align*}

\begin{theorem}[\Cref{thm:nachtergaele hypothesis} restated]\label{thm:nachtergaele hypothesis internal}
    Fix any $m\geq100$ and set $100\leq\ell\leq m$. Suppose $k\leq 2^{\ell/10}$. Then we have for any $f:\{\pm1\}^{mk}\to\R$ that
    \begin{align*}
        \abs{\left\langle f,R_{m,[m-\ell-1,m],k}\pbra{R_{m,[m-1],k}-R_{m,[m],k}}R_{m,[m-\ell-1,m],k}f\right\rangle}\leq \frac{k^3}{2^{\ell/4-60}}\left\langle f,f\right\rangle\leq \frac1{\ell^2} \left\langle f,f\right\rangle.
    \end{align*}
\end{theorem}
\begin{proof}
    We prove by induction on $k$. In the base case $k=1$ and the result holds by the following argument when we write $f=f_2$. Now assume that the result holds for real functions on $\{\pm1\}^{m(k-1)}$. 
    
    Let $f:\{\pm1\}^{mk}\to\R$. Write $f=f_0+f_1+f_2$ where $f_0$ is supported on $B_{=0}$, $f_1$ is supported on $B_{\geq1}^{\mathrm{coll}}$, and $f_2$ is supported on $B_{\geq1}^{\mathrm{safe}}$. By \Cref{lem:cross terms B0} applied with $R_{m,[m-\ell-1,m],k}$, $R_{m,[m-1],k}$ and $R_{m,[m],k}$ and $B_{=0}$ and $\{\pm1\}^{mk}\setminus B_{=0}=B_{\geq1}^{\mathrm{coll}}\cup B_{\geq1}^{\mathrm{safe}}$, the other cross terms vanish, and we have
    \begin{align*}
        &\abs{ \left\langle f,R_{m,[m-\ell-1,m],k}\pbra{R_{m,[m-1],k}-R_{m,[m],k}}R_{m,[m-\ell-1,m],k}f\right\rangle }\\
        =&\abs{ \left\langle f,Af\right\rangle }\tag{Defining $A$ for convenience of notation}\\
        \leq&\abs{ \left\langle f_0,Af_0\right\rangle }+\abs{ \left\langle f_2,Af_2\right\rangle }+\abs{ \left\langle f_1,A(f_1+f_2)\right\rangle }+\abs{ \left\langle f_2,Af_1\right\rangle }\\
        =&\abs{ \left\langle f_0,Af_0\right\rangle }+\abs{ \left\langle f_2,Af_2\right\rangle }+\abs{ \left\langle f_1,A(f_1+f_2)\right\rangle }+\abs{ \left\langle f_1,A^*f_2\right\rangle }\\
        \leq&\frac{(k-1)^3}{2^{\ell/4-60}}\left\langle f_0,f_0 \right\rangle+\frac{k^2}{2^{\ell/4-20}}\langle f_2,f_2\rangle  +\frac{k^2}{2^{\ell/2-4}}\norm{f_1+f_2}_2\norm{f_2}_2+\frac{k^2}{2^{\ell/2-4}}\norm{f_1}_2\norm{f_2}_2\tag{\Cref{lemma:f supported on B0} + Induction, \Cref{cor:square terms safe to safe}, \Cref{cor:cross terms safe to coll}, \Cref{cor:cross terms safe to coll}, in that order}\\
        \leq & \frac{k^3}{2^{\ell/4-60}}\langle f,f\rangle.
    \end{align*}
    In the last step we used that $k\leq 2^{\ell/10}$.
\end{proof}

Similar to \Cref{sec:B>=2} we will again argue that because $\ell$ is large, applying a random permutation to the same indices of all elements of a tuple of binary strings is similar to replacing those indices with random binary strings in all elements of a tuple. To model this situation, we similarly use the matrix $Q_{m,S,k}$ to denote the random walk matrix induced by the following distribution $\mathcal{C}_X^{m,S,k}$ for each $X\in\{\pm1\}^{mk}$. To draw $\mathbf{Y}\sim\mathcal{C}_X^{m,S,k}$, set $\mathbf{Y}^i_{[m]\setminus S}=X^i_{[m]\setminus S}$ and set $\mathbf{Y}^i_S$ uniformly at random for all $i\in[k]$.

\subsection{$f$ Supported on $B_{\geq 1}^{\mathrm{coll}}$ and Cross Terms}\label{sec:coll and cross terms}

The first step in comparing the noise model generated by the application of a random permutation to a substring of each string in a $k$-tuple of $n$-bit strings is showing that the set of strings to which this comparison fails already exhibits good expansion.

In this section we leverage that $B_{\geq1}^{\mathrm{coll}}$ is a set in $\{\pm1\}^{mk}$ with very good expansion. For intuition, $B_{\geq1}^{\mathrm{coll}}$ plays the role that $B_{=1}$ played in \Cref{sec:small k}. The following lemma formalizes the expansion property that we need.

\begin{lemma}\label{lem:anything transition into coll}
    For any $X\in B_{\geq 1}^{\mathrm{coll}}\cup B_{\geq 1}^{\mathrm{safe}}$ we have that for any $I,J\subseteq[m]$ such that $I\cup J=[m]$ and $I,J\supseteq [m-\ell-1,m-1]$,
    \begin{align*}
        \Pr\sbra{X\to_{R_{m,I,k}R_{m,J,k}} B_{\geq 1}^{\mathrm{coll}}}
        \leq\frac{k^2}{2^{\ell-3}}.
    \end{align*}
    Moreover, we also have
    \begin{align*}
        \Pr\sbra{X\to_{R_{m,I,k}R_{m,J,k}R_{m,I,k}} B_{\geq 1}^{\mathrm{coll}}}
        \leq\frac{k^2}{2^{\ell-3}}.
    \end{align*}
    In fact, this holds with $Q$ replacing $R$ anywhere in this inequality.
\end{lemma}
\begin{proof}
    We first prove the first statement (with two steps) for the case of drawing from the distribution $\mathcal{D}$ (corresponding to the random walk operator $R$), but the proof translates easily to the case for $\mathcal{C}$ (corresponding to the random wak operator $Q$). Model the process as first drawing $\bm{Y}\sim \mcD^{m,I,k}_X$ and then $\bm{Z}\sim \mcD^{m,J,k}_{\bm{Y}}$. Fix any $X\in B_{\geq1}^{\mathrm{coll}}\cup B_{\geq1}^{\mathrm{safe}}$. Define the following subset of pairs of indices in $\binom{[k]}{2}$:
    \begin{align*}
        P_1=& \cbra{\{i,j\}\in\binom{[k]}{2}:X^i_{I}=X^j_{I}}.
    \end{align*}
    Assume that $\{i,j\}\in P_1$. Then $\mathbf{Y}^i_I=\mathbf{Y}^j_I$ with probability 1. But now note that $X^i_{J}\neq X^j_{J}$ since otherwise $X^i=X^j$, which contradicts that $X\not\in B_{=0}$. Therefore, 
    \begin{align*}
        &\Pr_{\substack{\mathbf{Y}\sim \mathcal{D}_{X}^{m,I,k}\\\mathbf{Z}\sim\mathcal{D}_{\mathbf{Y}}^{m,J,k}}}\sbra{\mathbf{Z}^i_{[m-\ell-1,m-1]}=\mathbf{Z}^j_{[m-\ell-1,m-1]}}\\
        =&\Pr_{\substack{\mathbf{Z}\sim \mathcal{D}_{X}^{m,J,k}}}\sbra{\mathbf{Z}^i_{[m-\ell-1,m-1]}=\mathbf{Z}^j_{[m-\ell-1,m-1]}}\\
        =& \frac1{2^{\ell}-1}.\numberthis\label{eq:P1 contribution}
    \end{align*}
    Now assume that $\{i,j\}\not\in P_1$. Let $E$ be the event that $\mathbf{Y}^i_{[m-\ell-1,m-1]}=\mathbf{Y}^j_{[m-\ell-1,m-1]}$. Then 
    \begin{align*}
        &\Pr_{\substack{\mathbf{Y}\sim \mathcal{D}_{X}^{m,I,k}\\\mathbf{Z}\sim\mathcal{D}_{\mathbf{Y}}^{m,J,k}}}\sbra{\mathbf{Z}^i_{[m-\ell-1,m-1]}=\mathbf{Z}^j_{[m-\ell-1,m-1]}}\\
        =&\Pr_{\mathbf{Y}\sim \mathcal{D}_{X}^{m,I,k}}\sbra{E}\Pr_{\substack{\mathbf{Y}\sim \mathcal{D}_{X}^{m,I,k}\\\mathbf{Z}\sim\mathcal{D}_{\mathbf{Y}}^{m,J,k}}}\sbra{\mathbf{Z}^i_{[m-\ell-1,m-1]}=\mathbf{Z}^j_{[m-\ell-1,m-1]}|E}\\
        &+\Pr_{\mathbf{Y}\sim \mathcal{D}_{X}^{m,I,k}}\sbra{\overline{E}}\Pr_{\substack{\mathbf{Y}\sim \mathcal{D}_{X}^{m,I,k}\\\mathbf{Z}\sim\mathcal{D}_{\mathbf{Y}}^{m,J,k}}}\sbra{\mathbf{Z}^i_{[m-\ell-1,m-1]}=\mathbf{Z}^j_{[m-\ell-1,m-1]}|\overline{E}}\\
        \leq& \frac1{2^\ell} + \Pr_{\substack{\mathbf{Y}\sim \mathcal{D}_{X}^{m,I,k}\\\mathbf{Z}\sim\mathcal{D}_{\mathbf{Y}}^{m,J,k}}}\sbra{\mathbf{Z}^i_{[m-\ell-1,m-1]}=\mathbf{Z}^j_{[m-\ell-1,m-1]}|\overline{E}}\\
        =& \frac1{2^\ell}+\frac1{2^\ell-1}\\
        \leq&\frac1{2^{\ell-2}}.\numberthis\label{eq:not P1 contribution}
    \end{align*}
    To complete the proof, we use a union bound and find
    \begin{align*}
        &\Pr_{\substack{\mathbf{Y}\sim \mathcal{D}_{X}^{m,I,k}\\\mathbf{Z}\sim\mathcal{D}_{\mathbf{Y}}^{m,J,k}}}\sbra{\mathbf{Z}\in B_{\geq1}^{\mathrm{coll}}}\\
        \leq& \sum_{\{i,j\}\in\binom{[k]}{2}}\Pr_{\substack{\mathbf{Y}\sim \mathcal{D}_{X}^{m,I,k}\\\mathbf{Z}\sim\mathcal{D}_{\mathbf{Y}}^{m,J,k}}}\sbra{\mathbf{Z}^i_{[m-\ell-1,m-1]}=\mathbf{Z}^j_{[m-\ell-1,m-1]}}\\
        \leq&\sum_{\{i,j\}\in P_1}\Pr_{\substack{\mathbf{Y}\sim \mathcal{D}_{X}^{m,I,k}\\\mathbf{Z}\sim\mathcal{D}_{\mathbf{Y}}^{m,J,k}}}\sbra{\mathbf{Z}^i_{[m-\ell-1,m-1]}=\mathbf{Z}^j_{[m-\ell-1,m-1]}}\\
        &+\sum_{\{i,j\}\not\in P_1}\Pr_{\substack{\mathbf{Y}\sim \mathcal{D}_{X}^{m,I,k}\\\mathbf{Z}\sim\mathcal{D}_{\mathbf{Y}}^{m,J,k}}}\sbra{\mathbf{Z}^i_{[m-\ell-1,m-1]}=\mathbf{Z}^j_{[m-\ell-1,m-1]}}\\
        \leq& \sum_{\{i,j\}\in P_1}\frac1{2^{\ell}}+\sum_{\{i,j\}\not\in P_1}\frac1{2^{\ell-1}}\\
        \leq&\sum_{\{i,j\}\in \binom{[k]}{2}}\frac1{2^{\ell-2}}\\
        \leq& \frac{k^2}{2^{\ell-2}}\numberthis\label{eq:two step bound}.
    \end{align*}
    To prove the second statement (with three steps taken), note that the probability that a random permutation applied to the bits in $I$ sends an element of $B_{\geq1}^{\mathrm{safe}}$ to an element of $B_{\geq1}^{\mathrm{coll}}$ is at most $\frac{k^2}{2^{\ell}-1}$. Applying a union bound and the bound \Cref{eq:two step bound} completes the proof.
\end{proof}

The expansion property just proved results in the following linear-algebraic statements, which show that the contributions from $B_{\geq1}^{\mathrm{coll}}$ to our operator norm bound is extremely small.

\begin{lemma}\label{lem:randomize everything}
    Let $k\geq 2$ and $f:\{\pm1\}^{mk}$ be supported on $B_{\geq 1}^{\mathrm{coll}}$ and $g:\{\pm1\}^{mk}$ be supported on $B_{\geq 1}^{\mathrm{safe}}\cup B_{\geq 1}^{\mathrm{coll}}$. Let $I$ and $J$ be such that $I\cup J=[m]$. Then
    \begin{align*}
        &\abs{\left\langle f,R_{m,J,k}R_{m,I,k}R_{m,J,k}g\right\rangle}      \leq\sqrt{\frac{k^2}{2^{\ell-3}}}\norm{f}_2\norm{g}_2,\\
        &\abs{\left\langle f,R_{m,J,k}R_{m,I,k}g\right\rangle}\leq\sqrt{\frac{k^2}{2^{\ell-3}}}\norm{f}_2\norm{g}_2.
    \end{align*}
    Moreover, this holds with $Q$ in place of $R$ anywhere.
\end{lemma}
\begin{proof}
    The inequality directly follows from \Cref{lem:escape probs} and \Cref{lem:anything transition into coll}. We can apply \Cref{lem:escape probs} because the uniform distribution on $\{\pm1\}^{mk}$ is indeed a stationary distribution under $R_{m,S,k}$ for any $S$ (\Cref{fact:uniform is stationary}). The $Q$ case follows from the fact that \Cref{lem:anything transition into coll} applies for the distributions $\mathcal{C}$ too.
\end{proof}

\begin{corollary}\label{cor:cross terms safe to coll}
    Let $k\geq 2$ and $f:\{\pm1\}^{mk}$ be supported on $B_{\geq 1}^{\mathrm{coll}}$ and $g:\{\pm1\}^{mk}$ be supported on $B_{\geq 1}^{\mathrm{safe}}\cup B_{\geq 1}^{\mathrm{coll}}$. Then 
    \begin{align*}
        \abs{\left\langle f,R_{m,[m-\ell-1,m],k}\pbra{R_{m,[m-1],k}-R_{m,[m],k}}R_{m,[m-\ell-1,m],k}g\right\rangle }\leq \frac{k}{2^{\ell/2-4}}\norm{f}_2\norm{g}_2,\\
        \abs{\left\langle f,R_{m,[m-\ell-1,m],k}\pbra{R_{m,[m-1],k}-R_{m,[m],k}}g\right\rangle }\leq \frac{k}{2^{\ell/2-4}}\norm{f}_2\norm{g}_2.
    \end{align*}
\end{corollary}
\begin{proof}
    We compute
    \begin{align*}
        &\abs{\left\langle f,R_{m,[m-\ell-1,m],k}\pbra{R_{m,[m-1],k}-R_{m,[m],k}}R_{m,[m-\ell-1,m],k}g\right\rangle }\\
        \leq&\abs{\left\langle f,{R_{m,[m-\ell-1,m],k}R_{m,[m-1],k}}R_{m,[m-\ell-1,m],k}g\right\rangle}+\abs{\left\langle f,{R_{m,[m-\ell-1,m],k}R_{m,[m],k}}R_{m,[m-\ell-1,m],k}g\right\rangle }\\
        =&\abs{\left\langle f,{R_{m,[m-\ell-1,m],k}R_{m,[m-1],k}}R_{m,[m-\ell-1,m],k}g\right\rangle}+\abs{\left\langle f,{R_{m,[m],k}}g\right\rangle }\\
        \leq&\frac{k}{2^{\ell/2-3}}\norm{f}_2\norm{g}_2+\frac{k}{2^{\ell/2}}\norm{f}_2\norm{g}_2\tag{\Cref{lem:randomize everything}}\\
        \leq& \frac{k}{2^{\ell/2-4}}\norm{f}_2\norm{g}_2.
    \end{align*}
    The proof of the second statement is similar.
\end{proof}

\subsection{A Hybrid Argument for $f$ Supported on $B_{\geq 1}^{\mathrm{safe}}$}

The role that $B_{\geq 1}^{\mathrm{safe}}$ plays in this section is similar to the role played by $B_{\geq2}$ in \Cref{sec:small k}. It is the region of $\{\pm1\}^{mk}$ that is ``well-behaved" in the sense that the nicer noise model given by the $Q$ operators is similar to the noise model given by the $R$ operators on this region of $\{\pm1\}^{mk}$. 

In particular, because these noise models behave similarly when restricted to these $k$-tuples, we can bound the difference between the corresponding transition matrices as follows:

\begin{lemma}\label{lem:R to Q hybrid local}
    Assume that $k\leq 2^{\ell/10}$ and $f,g:\{\pm1\}^{mk}$ be supported on $B_{\geq 1}^{\mathrm{safe}}$. Then for any $S\supseteq [m-\ell-1,m-1]$, we have
    \begin{align*}
        \abs{\left\langle f,\pbra{Q_{m,S,k}-R_{m,S,k}}g\right\rangle}\leq \frac{k^2}{2^{\ell-1}}\norm{f}_2\norm{g}_2.
    \end{align*}
\end{lemma}
\begin{proof}
    By assumption $f$ and $g$ are supported on $B_{\geq1}^{\mathsf{safe}}$. Therefore, by \Cref{lem:TV distance bound} and self-adjointness of $Q_{m,S,k}$ and $R_{m,S,k}$ (\Cref{fact:self-adjoint}, \Cref{fact:Q self-adjoint}) we have
    \begin{align*}
        &\abs{\left\langle f,(R_{m,S,k}-Q_{m,S,k})g\right\rangle}\\
        \leq & \sqrt{\sum_{X\in B_{\geq 1}^{\mathrm{safe}}}f(X)^2\sum_{Y\in B_{\geq 1}^{\mathrm{safe}}}\abs{\Pr\sbra{X\to_{R_{m,S,k}} Y}-\Pr\sbra{X\to_{Q_{m,S,k}} Y}}}\\
        &\;\;\;\;\;\cdot\sqrt{\sum_{X\in B_{\geq 1}^{\mathrm{safe}}}g(X)^2\sum_{Y\in B_{\geq 1}^{\mathrm{safe}}}\abs{\Pr\sbra{X\to_{R_{m,S,k}} Y}-\Pr\sbra{X\to_{Q_{m,S,k}} Y}}}\\
        \leq & \frac{k^2}{2^{\ell-1}}\sqrt{\sum_{X\in B_{\geq 1}^{\mathrm{safe}}}f(X)^2}\sqrt{\sum_{X\in B_{\geq 1}^{\mathrm{safe}}}g(X)^2}\tag{\Cref{eq:tv for safe} below}\\
        =& \frac{k^2}{2^{\ell-1}}\norm{f}_2\norm{g}_2.
    \end{align*}
Here $p_0$ and $p_1$ are as used below. Now it suffices to establish \Cref{eq:tv for safe}. Assume $X\in B_{\geq 1}^{\mathrm{safe}}$. Then because $S\supseteq [m-\ell-1,m-1]$, we know that for all $i\neq j$ we have $X^i_{S}\neq X^j_S$.
\begin{align*}
    &\sum_{Y\in B_{\geq 1}^{\mathrm{safe}}}\abs{\Pr\sbra{X\to_{R_{m,S,k}} Y}-\Pr\sbra{X\to_{Q_{m,S,k}} Y}}\\
    =&\sum_{\substack{Y\in B_{\geq 1}^{\mathrm{safe}}\\\forall i\in[k], a\in [m]\setminus S, Y^i_a=X^i_a}}\abs{\Pr_{\mathbf{Y}\sim\mathcal{D}^{m,S,k}_X}\sbra{\mathbf{Y}=Y}- \Pr_{\mathbf{Y}\sim\mathcal{C}^{m,S,k}_X}\sbra{\mathbf{Y}=Y}}\\
    =&\sum_{\substack{Y\in B_{\geq 1}^{\mathrm{safe}}\\\forall i\in[k], a\in [m]\setminus S, Y^i_a=X^i_a}}\abs{ \prod_{j=0}^{k-1}\frac{1}{2^{|S|}-j}- \frac1{2^{|S|k}}}\tag{$k\leq 2^{m/3}\leq 2^m-2$ and $X,Y\in B_{\geq 1}^{\mathrm{safe}}$}\\
    \leq &\sum_{\substack{Y\in B_{\geq 1}^{\mathrm{safe}}\\\forall i\in[k], a\in [m]\setminus S, Y^i_a=X^i_a}}\abs{ \frac1{2^{|S|k}}\pbra{\prod_{j=0}^{k-1}\frac{2^{|S|}}{2^{|S|}-j}- 1}}\\
    \leq &\sum_{\substack{Y\in B_{\geq 1}^{\mathrm{safe}}\\\forall i\in[k], a\in [m]\setminus S, Y^i_a=X^i_a}}\abs{ \frac1{2^{|S|k}}\pbra{1+\frac{k^2}{2^{|S|}}- 1}}\tag{$k^2\leq 2^{\ell}\leq 2^{|S|}$, \Cref{fact:k^2}}\\
    = &\sum_{\substack{Y\in B_{\geq 1}^{\mathrm{safe}}\\\forall i\in[k], a\in [m]\setminus S, Y^i_a=X^i_a}}\frac{k^2}{2^{|S|}2^{|S|k}}\\
    = &\cdot 2^{|S|k}\cdot \frac{k^2}{2^{|S|}2^{|S|k}}\\
    =& \frac{k^2}{2^{|S|}}\\
    \leq& \frac{k^2}{2^{\ell-1}}.\numberthis\label{eq:tv for safe}
\end{align*}
The last inequality follows because $|S|\geq\abs{[m-\ell-1,m-1]}=\ell-1$. Having established \Cref{eq:tv for safe}, we have completed the proof.
\end{proof}

At this point it may seem like we are essentially finished with the proof, since we should just replace the $R$ operators in the expression $\left\langle f,{R_{m,[m-\ell-1,m],k}\pbra{R_{m,[m-1],k}-R_{m,[m],k}}}R_{m,[m-\ell-1,m],k}f\right\rangle$ with the corresponding $Q$ operators and finish the proof. However, we don't quite show an upper bound on the \textit{operator norm} of $R-Q$. Rather, we simply show that they are close on the well-behaved region. A priori, this gives us no information about how products of these operators may behave, since the first term in the product may ``rotate" vectors into the badly-behaved region. So a straightforward application of the triangle inequality fails.

However, we observe that we have already shown in \Cref{sec:coll and cross terms} that random walks starting outside the badly-behaved region rarely transition into it, so that we almost can pretend as if all of these operators are operators on $\R^{B_{\geq1}^{\mathrm{safe}}}$. We formalize this by inserting projections to the space of functions supported on $B_{\geq1}^{\mathrm{safe}}$, and showing that this move does very little quantitatively.

\begin{lemma}\label{lem:product of operators quadratic form}
    We have for any $f,g:\{\pm1\}^{mk}\to\R$ supported on $B_{\geq1}^{\mathrm{safe}}$ that
    \begin{align*}
        \abs{\left\langle f,\pbra{R_{m,[m-\ell-1,m],k}-Q_{m,[m-\ell-1,m],k}}\pbra{R_{m,[m-1],k}-R_{m,[m],k}}g\right\rangle }\leq \frac{k^2}{2^{\ell/4-20}}\norm{f}_2\norm{g}_2.
    \end{align*}
    Moreover, we have for such $f$ and $g$ that
    \begin{align*}
        \abs{\left\langle f,Q_{m,[m-\ell-1,m],k}\pbra{R_{m,[m-1],k}-Q_{m,[m-1],k}-R_{m,[m],k}+Q_{m,[m],k}}g\right\rangle }\leq \frac{k^2}{2^{\ell/4-20}}\norm{f}_2\norm{g}_2.
    \end{align*}
\end{lemma}
\begin{proof}
    Let $\Pi_{\mathrm{safe}}$ be the projection to $\{h:\{\pm1\}^{mk}\to\R:h\text{ supported on }B_{\geq1}^{\mathrm{safe}}\}$. We directly compute
    \begin{align*}
        &\abs{\left\langle f,\pbra{R_{m,[m-\ell-1,m],k}-Q_{m,[m-\ell-1,m],k}}R_{m,[m-1],k}g\right\rangle }\\
        \leq&\abs{\left\langle f,\pbra{R_{m,[m-\ell-1,m],k}-Q_{m,[m-\ell-1,m],k}}\pbra{\Id-\Pi_{\mathrm{safe}}}R_{m,[m-1],k}g\right\rangle }\\
        &\;\;+\abs{\left\langle f,\pbra{R_{m,[m-\ell-1,m],k}-Q_{m,[m-\ell-1,m],k}}\Pi_{\mathrm{safe}}R_{m,[m-1],k}g\right\rangle}\\
        \leq&\abs{\left\langle f,\pbra{R_{m,[m-\ell-1,m],k}-Q_{m,[m-\ell-1,m],k}}\pbra{\Id-\Pi_{\mathrm{safe}}}R_{m,[m-1],k}g\right\rangle }+\frac{k^2}{2^{\ell-1}}\norm{f}_2\norm{\Pi_{\mathrm{safe}}R_{m,[m-1],k}g}_2\tag{\Cref{lem:R to Q hybrid local}}\\
        \leq&\abs{\left\langle f,R_{m,[m-\ell-1,m],k}\pbra{\Id-\Pi_{\mathrm{safe}}}R_{m,[m-1],k}g\right\rangle}+\abs{\left\langle f,Q_{m,[m-\ell-1,m],k}\pbra{\Id-\Pi_{\mathrm{safe}}}R_{m,[m-1],k}g\right\rangle}+\frac{k^2}{2^{\ell-1}}\norm{f}_2\norm{g}_2\\
        \leq&2\norm{f}_2\norm{\pbra{\Id-\Pi_{\mathrm{safe}}}R_{m,[m-1],k}g}_2+\frac{k^2}{2^{\ell-1}}\norm{f}_2\norm{g}_2\\
        \leq&\frac{k^2}{2^{\ell/4-10}}\norm{f}_2\norm{g}_2+\frac{k^2}{2^{\ell-1}}\norm{f}_2\norm{g}_2.\tag{\Cref{eq:safe to not safe R} below}
    \end{align*}
    Here our application of the \Cref{lem:R to Q hybrid local} depended on the fact that $\mathrm{Supp}(\Pi_{\mathrm{safe}}R_{m,[m-1],k}f)\subseteq B_{\geq1}^{\mathrm{safe}}$ and $\mathrm{Supp}((\Id-\Pi_{\mathrm{safe}})R_{m,[m-1],k}f)\subseteq \{\pm1\}^{mk}\setminus B_{\geq1}^{\mathrm{safe}}$.
    
    To establish \Cref{eq:safe to not safe R} we compute 
    \begin{align*}
        &\norm{\pbra{\Id-\Pi_{\mathrm{safe}}}R_{m,[m-1],k}g}_2^2\\
        =&\left\langle \pbra{\Id-\Pi_{\mathrm{safe}}}R_{m,[m-1],k}g,\pbra{\Id-\Pi_{\mathrm{safe}}}R_{m,[m-1],k}g\right\rangle\\
        =&\left\langle R_{m,[m-1],k}g,\pbra{\Id-\Pi_{\mathrm{safe}}}^2R_{m,[m-1],k}g\right\rangle\\
        =&\left\langle f,R_{m,[m-1],k}\pbra{\Id-\Pi_{\mathrm{safe}}}R_{m,[m-1],k}g\right\rangle\tag{self-adjointness (\Cref{fact:self-adjoint})}\\
        \leq &\frac{k}{2^{\ell/2-1}}\norm{g}_2\norm{\pbra{\Id-\Pi_{\mathrm{safe}}}R_{m,[m-1],k}g}_2\\
        \leq&\frac{k}{2^{\ell/2-1}}\norm{g}_2^2\numberthis\label{eq:safe to not safe R}.
    \end{align*}
    The first inequality follows from \Cref{lem:escape probs} the fact that $\mathrm{Supp}(f)\subseteq B_{\geq1}^{\mathrm{safe}}$ and for any $X\in B_{\geq1}^{\mathrm{safe}}$, we have $\Pr\sbra{X\to_{R_{m,[m-1],k}}B_{\geq1}^{\mathrm{coll}}}\leq \frac{k^2}{2^{\ell-1}}$ (using the same proof as \Cref{lem:anything transition into coll}), and that $\mathrm{Supp}\pbra{\pbra{\Id-\Pi_{\mathrm{safe}}}R_{m,[m-1],k}f}\subseteq \{\pm1\}^{mk}\setminus B_{\geq1}^{\mathrm{safe}}$. Bounding the similar quantity but with $R_{m,[m],k}$ instead of $R_{m,[m-1],k}$ is the same and the first part of the lemma statement follows from the triangle inequality. 
    
    The second part of the lemma statement follows from the same argument and an application of the triangle inequality:
    \begin{align*}
        &\abs{\left\langle f,Q_{m,[m-\ell-1,m],k}\pbra{R_{m,[m-1],k}-Q_{m,[m-1],k}-R_{m,[m],k}+Q_{m,[m],k}}g\right\rangle}\\
        =&\abs{\left\langle f,\pbra{R_{m,[m-1],k}-Q_{m,[m-1],k}-R_{m,[m],k}+Q_{m,[m],k}}Q_{m,[m-\ell-1,m],k}g\right\rangle}\tag{\Cref{fact:self-adjoint}, \Cref{fact:Q self-adjoint}}\\
        \leq&\abs{\left\langle f,\pbra{R_{m,[m-1],k}-Q_{m,[m-1],k}}Q_{m,[m-\ell-1,m],k}g\right\rangle}+\abs{\left\langle f,\pbra{R_{m,[m],k}-Q_{m,[m],k}}Q_{m,[m-\ell-1,m],k}g\right\rangle}.\numberthis\label{eq:hybrid the difference}
    \end{align*}
    We show how to bound the first term in this sum:
    \begin{align*}
        &\abs{\left\langle f,\pbra{R_{m,[m-1],k}-Q_{m,[m-1],k}}Q_{m,[m-\ell-1,m],k}g\right\rangle}\\
        \leq&\abs{\left\langle f,\pbra{R_{m,[m-1],k}-Q_{m,[m-1],k}}\pbra{\Id-\Pi_{\mathrm{safe}}}Q_{m,[m-\ell-1,m],k}g\right\rangle }\\
        &\;\;+\abs{\left\langle f,\pbra{R_{m,[m-1],k}-Q_{m,[m-1],k}}\Pi_{\mathrm{safe}}Q_{m,[m-\ell-1,m],k}g\right\rangle}\\
        \leq&\abs{\left\langle f,\pbra{R_{m,[m-1],k}-Q_{m,[m-1],k}}\pbra{\Id-\Pi_{\mathrm{safe}}}Q_{m,[m-\ell-1,m],k}g\right\rangle }+\frac{k^2}{2^{\ell-1}}\norm{f}_2\norm{\Pi_{\mathrm{safe}}Q_{m,[m-\ell-1,m],k}g}_2\tag{\Cref{lem:R to Q hybrid local}}\\
        \leq&\abs{\left\langle f,R_{m,[m-1],k}\pbra{\Id-\Pi_{\mathrm{safe}}}Q_{m,[m-\ell-1,m],k}g\right\rangle }\\
        &+\abs{\left\langle f,Q_{m,[m-1],k}\pbra{\Id-\Pi_{\mathrm{safe}}}Q_{m,[m-\ell-1,m],k}g\right\rangle }+\frac{k^2}{2^{\ell-1}}\norm{f}_2\norm{g}_2\\
        \leq&2\norm{f}_2\norm{\pbra{\Id-\Pi_{\mathrm{safe}}}Q_{m,[m-\ell-1,m],k}g}_2+\frac{k^2}{2^{\ell-1}}\norm{f}_2\norm{g}_2\\
        \leq&\frac{k}{2^{\ell/4-10}}\norm{f}_2\norm{g}_2+\frac{k^2}{2^{\ell-1}}\norm{f}_2\norm{g}_2.\tag{\Cref{eq:safe to not safe Q} below}
    \end{align*}
    To establish \Cref{eq:safe to not safe Q} we compute 
    \begin{align*}
        &\norm{\pbra{\Id-\Pi_{\mathrm{safe}}}Q_{m,[m-\ell-1,m],k}g}_2^2\\
        =&\left\langle g,Q_{m,[m-\ell-1,m],k}\pbra{\Id-\Pi_{\mathrm{safe}}}Q_{m,[m-\ell-1,m],k}g\right\rangle\tag{self-adjointness (\Cref{fact:Q self-adjoint}}\\
        \leq &\frac{k}{2^{\ell/2-1}}\norm{g}_2\norm{\pbra{\Id-\Pi_{\mathrm{safe}}}Q_{m,[m-\ell-1,m],k}g}_2\\
        \leq&\frac{k}{2^{\ell/2-1}}\norm{g}_2^2\numberthis\label{eq:safe to not safe Q}.
    \end{align*}
    The first inequality follows from \Cref{lem:escape probs} the fact that $\mathrm{Supp}(f)\subseteq B_{\geq1}^{\mathrm{safe}}$ and for any $X\in B_{\geq1}^{\mathrm{safe}}$, we have $\Pr\sbra{X\to_{Q_{m,[m-1],k}}B_{\geq1}^{\mathrm{coll}}}\leq \frac{k^2}{2^{\ell-1}}$ (using the same proof as \Cref{lem:anything transition into coll}). The bound on the second term in the sum \Cref{eq:hybrid the difference} follows from the same argument.
\end{proof}

\begin{corollary}\label{cor:square terms safe to safe}
    Let $k\geq 2$ and $f:\{\pm1\}^{mk}\to\R$ be supported on $B_{\geq 1}^{\mathrm{safe}}$. Then 
    \begin{align*}
        \abs{\left\langle f,R_{m,[m-\ell-1,m],k}\pbra{R_{m,[m-1],k}-R_{m,[m],k}}R_{m,[m-\ell-1,m],k}f\right\rangle }\leq \frac{k^2}{2^{\ell/4-50}}\langle f,f\rangle.
    \end{align*}
\end{corollary}
\begin{proof}
    Let $g=R_{m,[m-\ell-1,m],k}f$ and we can write $g=g_1+g_2$ where $g_1=\Pi_{\mathrm{safe}}g$ and $g_2=\pbra{\mathrm{Id}-\Pi_{\mathrm{safe}}}g$. Then
    \begin{align*}
        &\abs{\left\langle f,R_{m,[m-\ell-1,m],k}\pbra{R_{m,[m-1],k}-R_{m,[m],k}}R_{m,[m-\ell-1,m],k}f\right\rangle }\\
        \leq&\abs{\left\langle f,R_{m,[m-\ell-1,m],k}\pbra{R_{m,[m-1],k}-R_{m,[m],k}}g_1\right\rangle }+\abs{\left\langle f,R_{m,[m-\ell-1,m],k}\pbra{R_{m,[m-1],k}-R_{m,[m],k}}g_2\right\rangle }\\
        \leq&\abs{\left\langle f,R_{m,[m-\ell-1,m],k}\pbra{R_{m,[m-1],k}-R_{m,[m],k}}g_1\right\rangle }+\frac{k^2}{2^{\ell/2-3}}\norm{f}_2\norm{g_2}_2.\tag{\Cref{cor:cross terms safe to coll}}
    \end{align*}
    To bound the first term, we note that $g_1$ is supported on $B_{\geq1}^{\mathrm{safe}}$ and directly compute
    \begin{align*}
        &\abs{\left\langle f,R_{m,[m-\ell-1,m],k}\pbra{R_{m,[m-1],k}-R_{m,[m],k}}g_1\right\rangle }\\
        \leq&\abs{\left\langle f,Q_{m,[m-\ell-1,m],k}\pbra{R_{m,[m-1],k}-R_{m,[m],k}}g_1\right\rangle }+\\
        &\;\;\abs{\left\langle f,\pbra{R_{m,[m-\ell-1,m],k}-Q_{m,[m-\ell-1,m],k}}\pbra{R_{m,[m-1],k}-R_{m,[m],k}}g_1\right\rangle }\\
        \leq &\abs{\left\langle f,Q_{m,[m-\ell-1,m],k}\pbra{R_{m,[m-1],k}-R_{m,[m],k}}g_1\right\rangle }+\frac{k^2}{2^{\ell/4-20}}\norm{f}_2\norm{g_1}_2\tag{\Cref{lem:product of operators quadratic form}, first part}\\
        \leq &\frac{k^2}{2^{\ell/2-20}}\norm{f}_2\norm{g_1}_2+\abs{\left\langle f,Q_{m,[m-\ell-1,m],k}\pbra{Q_{m,[m-1],k}-Q_{m,[m],k}}f\right\rangle }\\
        &\;\;\;+\abs{\left\langle f,Q_{m,[m-\ell-1,m],k}\pbra{R_{m,[m-1],k}-Q_{m,[m-1],k}-R_{m,[m],k}+Q_{m,[m],k}}g_1\right\rangle }\\
        = &\frac{k^2}{2^{\ell/4-20}}\norm{f}_2\norm{g_1}_2+\abs{\left\langle f,Q_{m,[m-\ell-1,m],k}\pbra{R_{m,[m-1],k}-Q_{m,[m-1],k}-R_{m,[m],k}+Q_{m,[m],k}}g_1\right\rangle}\\
        =&\frac{k^2}{2^{\ell/4-20}}\norm{f}_2\norm{g_1}_2+\frac{k^2}{2^{\ell/4-20}}\norm{f}_2\norm{g_1}_2\tag{\Cref{lem:product of operators quadratic form}, second part}.
    \end{align*}
    For the second-to-last equality we used that $Q_{m,[m-\ell-1,m],k}\pbra{Q_{m,[m-1],k}-Q_{m,[m],k}}=0$ because $Q_{m,S,k}Q_{m,T,k}=Q_{m,S\cup T,k}$ for any $S,T\subseteq[m]$.

    To complete the proof we use that $\norm{g_1}_2,\norm{g_2}_2\leq \norm{f}_2$.
\end{proof}

\section{Reduction from Two-Dimensional to One-Dimensional Constructions: Proof of \Cref{thm:2D to 1D reduction}}\label{sec:proof of 2D to 1D reduction}

\subsection{Definitions}\label{sec:definitions}
To prove \Cref{thm:2D to 1D reduction} we introduce some new notation required for dealing with bit arrays on a higher-dimensional lattice architecture. 
\subsubsection{Bit Arrays}
We regard an element $x\in \{\pm1\}^{n}$ as a function $x : \sbra{\sqrt{n}\,} \times \sbra{\sqrt{n}\,} \to \{\pm1\}$. 
Similarly, we regard an element $X \in \{\pm1\}^{nk}$ as a function $X : \sbra{\sqrt{n}\,} \times \sbra{\sqrt{n}\,} \times \sbra{k} \to \{\pm1\}$. 
For $X \in \{\pm1\}^{nk}$, and $i,j \in \sbra{\sqrt{n}\,}$, and $\ell \in [k]$, we use the notation:
\begin{itemize}
    \item $X^\ell_{i, j} = X(i, j, \ell) \in \{\pm1\}$
    \item $X^\ell = X \mid_{\sbra{\sqrt{n}\,} \times \sbra{\sqrt{n}\,} \times \{\ell\}} \in \{\pm1\}^{n}$
    \item $X^\ell_{i, \cdot} = X \mid_{\{i\} \times \sbra{\sqrt{n}\,} \times \{\ell\}} \in \{\pm1\}^{\sqrt{n}}$
    \item $X^\ell_{\cdot, j} = X \mid_{\sbra{\sqrt{n}\,} \times \{j\} \times \{\ell\}} \in \{\pm1\}^{\sqrt{n}}$
    \item $X_{i, \cdot} = X \mid_{\{i\} \times \sbra{\sqrt{n}\,} \times [k]} \in \{\pm1\}^{\sqrt{n}k}$
    \item $X_{\cdot, j} = X \mid_{\sbra{\sqrt{n}\,} \times \{j\} \times [k]} \in \{\pm1\}^{\sqrt{n}k}$
\end{itemize}
We will use $\sD_n^{(k)}$ to denote the set of all $X \in \{\pm1\}^{nk}$ such that $X^i \ne X^j$ when $i \ne j$. Unless otherwise specified, $\sD$ refers to $\sD_{n}^{(k)}$.

\subsubsection{Color Classes}
\label{sec:colorclasses}
We partition $\{\pm1\}^{nk}$ into ``color classes'' via the following relation. Let $R^{(\sqrt{n})}$ be a tuple of $\sqrt{n}$ equivalence relations on $[k]$. That is, there is one equivalence relation for each row in $[\sqrt{n}]$. Then we define:
\begin{equation*}
        B_{R^{(\sqrt{n})}} = \left\{ X \in \{\pm1\}^{nk} : \forall i \in \sqrt{n}, X_{i, \cdot}^\ell = X_{i, \cdot}^m \text{ if and only if } \ell \,R_i\, m \right\}.
\end{equation*}

Informally, $X^{\ell}_{i, \cdot}$ and $X^m_{i, \cdot}$ share a color if $X^{\ell}_{i, \cdot} = X^m_{i, \cdot}$. This relation induces a coloring on the rows of $X$. We then say that $X$ and $Y$ are colored the same if all of their rows are colored the same. Since $R^{(\sqrt{n})}$ is an equivalence relation itself, the sets $\{B_{R^{(\sqrt{n})}}\}_{R^{(\sqrt{n})} \in \sR^{\otimes \sqrt{n}}}$ partition $\{\pm1\}^{nk}$ for $\sR$ the set of equivalence relations on $[k]$. Additionally, each $X$ has a unique color class we will denote as $B(X)$. We will also use $\sB$ to denote the set of color classes. This partition is useful in part due to its size.
\begin{fact}
    \label{fact:numofcolorclasses}
    There are $\leq k^{k\sqrt{n}}$ color classes, that is, $\abs{\sB} \leq k^{k\sqrt{n}}$.
\end{fact}

\begin{proof}
    We can count each color class by identifying the partition of each row, of which there are $\sqrt{n}$. Each row consists of $k$ elements, so we can overcount the number as putting the $k$ elements into $k$ partitions, $k^k$.
\end{proof}

We will define a simpler partition that will facilitate much of our analysis:
\begin{align*}
    &B_\text{safe} := \cbra{X \in \sD : \forall \ell \neq  m \in [k], i \in [\sqrt{n}], X^\ell_{i, \cdot} \neq X^m_{i, \cdot}},\\
    &B_\text{coll} := \sD \setminus B_\text{safe},\\
    &B_{=0} := \{\pm1\}^{nk} \setminus \sD.
\end{align*}
That is, $B_\text{safe}$ is the color class determined by the $\sqrt{n}$-wise product of the identity relation. $B_\text{coll}$ then consists of all other color classes within $\sD$, whereas $B_{=0}$ consists of all elements outside of $\sD$. We will often use the following result on the size of $B_\text{coll}$:
\begin{fact}
\label[fact]{fact:colorclasssizes}
    $\frac{\abs{B_{\mathrm{coll}}}}{\abs{\sD}} \leq \frac{2\sqrt{n}k^2}{2^{\sqrt{n}}}$.
\end{fact}

\begin{proof}
We may write:
\begin{equation*}
    \frac{\abs{B_{\mathrm{coll}}}}{\abs{\sD}} = \frac{\abs{B_{\mathrm{coll}}}}{\abs{\{\pm1\}^{nk}}} \cdot \frac{\abs{\{\pm1\}^{nk}}}{\abs{\sD}}.
\end{equation*}
The first can be viewed as the probability of sampling an element of $B_{\mathrm{coll}}$ when sampling from $\{\pm1\}^{nk}$. The process of sampling from $\{\pm1\}^{nk}$ can be seen as sampling $\sqrt{n}k$ rows from $\{\pm1\}^{\sqrt{n}}$. Under this view, a simple union bound tells us that there are at most $\sqrt{n}k^2$ possible ``collisions'' that would induce a non-distinct color class, allowing us to bound the probability by $\frac{\sqrt{n}k^2}{2^{\sqrt{n}}}$.

For the other term, we will prove simply that $\frac{\abs{B_{=0}}}{\abs{\{\pm1\}^{nk}}} \leq \frac{1}{2}$. Note that our analysis above actually bounds the probability $X \sim \{\pm1\}^{nk}$ is not in $B_{\mathrm{safe}}$, which is more than sufficient for this. \end{proof}

\subsubsection{Distributions}
\label{subsec:distributions}
Because in this section we are dealing with permutations of $n$-bit strings where the bits lie on a higher-dimensional lattice, it will be convenient to introduce completely new notation that will supplant the random walk operators $R_{n,S,k}$ from before.

If $\pi \in \mfS_{\{\pm1\}^{n}}$, then let $\pi^{\otimes k} \in \mfS_{\{\pm1\}^{nk}}$ be such that $\pi^{\otimes k}(X)^\ell = \pi(X^\ell)$ for all $X \in \{\pm1\}^{nk}$ and $\ell \in [k]$. 

\begin{itemize}
    \item Let $\mcB$ be a distribution on $\mfS_{\{\pm1\}^{\sqrt{n}}}$.
    
    \item Let $\mcP_R$ be a distribution on $\mfS_{\{\pm1\}^{n}}$ such that $\pi \sim \mcP_R$ is sampled as follows: 
          Sample $\sigma_i \sim \mcB$ independently for each $i \in \sbra{\sqrt{n}\,}$ and define $\pi$ such that $\pi(x)_{i, \cdot} = \sigma_i(x_{i, \cdot})$ for all $x \in \{\pm1\}^{n}$ and all $i \in \sbra{\sqrt{n}\,}$. 
    \item Let $\mcP_C$ be a distribution on $\mfS_{\{\pm1\}^{n}}$ such that $\pi \sim \mcP_C$ is sampled as follows: 
          Sample $\sigma_i \sim \mcB$ independently for each $i \in \sbra{\sqrt{n}\,}$ and define $\pi$ such that $\pi(x)_{\cdot, i} = \sigma_i(x_{\cdot, i})$ for all $x \in \{\pm1\}^{n}$ and all $i \in \sbra{\sqrt{n}\,}$. 
    \item Let $\mcP^0 = \mcP_R$. For all $t \ge 1$, let $\mcP^t$ be a distribution on $\mfS_{\{\pm1\}^{n}}$ such that $\pi \sim \mcP^t$ is sampled as follows:
          Sample $\sigma_1 \sim \mcP^{t-1}$, $\sigma_2 \sim \mcP_C$, and $\sigma_3 \sim \mcP_R$ and define $\pi$ such that $\pi(x) = (\sigma_3 \circ \sigma_2 \circ \sigma_1)(x)$ for all $x \in \{\pm1\}^{n}$. It is worth noting that this construction is exactly that of our circuit model above.
    \item Let $\mcG_R$ be a distribution on $\mfS_{\{\pm1\}^{n}}$ such that $\pi \sim \mcG_R$ is sampled as follows:
          Sample $\sigma_i \sim \mcU(\mfS_{\{\pm1\}^{\sqrt{n}}})$ independently for each $i \in \sbra{\sqrt{n}\,}$ and define $\pi$ such that $\pi(x)_{i, \cdot} = \sigma_i(x_{i,\cdot})$ for all $x \in \{\pm1\}^{n}$. 
    \item Let $\mcG_C$ be a distribution on $\mfS_{\{\pm1\}^{n}}$ such that $\pi \sim \mcG_C$ is sampled as follows:
          Sample $\sigma_i \sim \mcU(\mfS_{\{\pm1\}^{\sqrt{n}}})$ independently for each $i \in \sbra{\sqrt{n}\,}$ and define $\pi$ such that $\pi(x)_{\cdot, i} = \sigma_i(x_{\cdot, i})$ for all $x \in \{\pm1\}^{n}$. 
    \item Let $\mcG^0 = \mcG_R$. For all $t \ge 1$, let $\mcG^t$ be a distribution on $\mfS_{\{\pm1\}^{n}}$ such that $\pi \sim \mcG^t$ is sampled as follows:
          Sample $\sigma_1 \sim \mcG^{t-1}$, $\sigma_2 \sim \mcG_C$, and $\sigma_3 \sim \mcG_R$ and define $\pi$ such that $\pi(x) = (\sigma_3 \circ \sigma_2 \circ \sigma_1)(x)$ for all $x \in \{\pm1\}^{n}$.
    \item Let $\mcG$ be $\mcU(\mfS_{\{\pm1\}^{n}})$, i.e. the uniform distribution on $\mfS_{\{\pm1\}^{n}}$.
    \item If $\mcD$ is a distribution on $\mfS_{\{\pm1\}^{n}}$ and $X \in \{\pm1\}^{nk}$, let $\mcD^{(k)}_X$ be a distribution on $\{\pm1\}^{nk}$ such that $Y \sim \mcD^{(k)}_X$ is sampled as follows:
          Sample $\pi \sim \mcD$ and define $Y$ such that $Y = \pi^{\otimes k}(X)$. 
          Note that if $X \in \sD$, then $\mcG^{(k)}_X$ is $\mcU\pbra{\sD}$. If the superscript is understood from the context of $X$ to be $k$, we will often drop it and just write $\mcD_X$.
\end{itemize}

It will be helpful to think of the distribution $\mcD$ as defining a Markov chain on $\{\pm1\}^{nk}$ for any choice $k \geq 1$. More specifically, $\mcD^{(k)}_X$ can be thought of as specifying the transition probabilities out of $X$ in the corresponding Markov chain on $\{\pm1\}^{nk}$. The transition operators corresponding to these Markov chains are given by:
\begin{align*}
    (T_{\mcD}f)(X)=\Ex_{\mathbf{Y} \sim \mathcal{D}_X^{(k)}}\sbra{f(\mathbf{Y})} = \sum_{Y \in \{\pm1\}^{nk}} \Pr_{\bm{Y}\sim\mcD_{X}^{(k)} }[\bm{Y}=Y] \cdot f(Y).
\end{align*}
When the value of $k$ is not clear from context, we will write $T_{\mcD^{(k)}}$ for this operator.

\subsection{Linear Algebra}
Recall that $\mcD_X$ can be thought of as the distribution after a 1-step random walk from $X$ according to a permutation drawn from $\mcD$.
\begin{fact}
    \label[fact]{fact:selfadjoint}
    $T_{\mcG_R}$, $T_{\mcG_C}$, and $T_{\mcG}$ are self-adjoint w.r.t our inner product.
\end{fact}
\begin{proof}
We will state the proof for $T_{\mcG_R}$, the reasoning for the others being symmetric. Let $f, g : \{\pm1\}^{nk} \to \R$.
\begin{align*}
    \langle f, T_{\mcG_R} g \rangle &= \sum_{X \in \{\pm1\}^{nk}} f(X)(T_{\mcG_R}g)(X)\\
    &= \sum_{X \in \{\pm1\}^{nk}} f(X) \sum_{Y \in \{\pm1\}^{nk}} \Pr[X \to_{T_{\mcG_R}} Y] \cdot g(Y)\\
    &= \sum_{X \in \{\pm1\}^{nk}}\sum_{Y \in \{\pm1\}^{nk}} f(X) \cdot g(Y) \cdot \Pr[X \to_{T_{\mcG_R}} Y].
\end{align*}
Note the symmetry in above: we are done if we prove that $\Pr[X \to_{T_{\mcG_R}} Y] = \Pr[Y \to_{T_{\mcG_R}} X]$. This fact is quite observable from the definition of $\mcG_R$, since for any $\sigma \in \mfS_{\{\pm1\}^{n}}$ the probability it is drawn from $\mcG_R$ is the same as the probability of drawing $\sigma^{-1}$. 
\end{proof}

\begin{fact}
    \label[fact]{fact:TGabsorbs}
    Let $\mcD$ be a distribution on $\mfS_{\{\pm1\}^{n}}$. Then:
    \begin{equation*}
        T_\mcD T_\mcG = T_\mcG = T_\mcG T_\mcD.
    \end{equation*}
\end{fact}

\begin{fact}
    \label{fact:differenceofproducts}
    Let $U_1, ..., U_s$ and $W_1,..., W_s$ be operators. Then we have:
    \begin{equation*}
        \prod_{i= 1}^s U_i - \prod_{i=1}^s W_i = \sum_{i=1}^s \prod_{j=1}^{i-1} U_j \cdot \pbra{U_i - W_i} \cdot \prod_{j = i+1}^s W_j.
    \end{equation*}
\end{fact}
\begin{proof}
We prove this by induction on $s$. The base case $s=1$ is immediate. For the inductive hypothesis note that:
\begin{align*}
    \prod_{i= 1}^{s+1} U_i - \prod_{i=1}^{s+1} W_i &= \pbra{\prod_{i= 1}^s U_i - \prod_{i=1}^s W_i}W_{s+1} + \prod_{i=1}^s U_i \pbra{U_{s+1} - W_{s+1}}\\
    &= \sum_{i=1}^s \prod_{j=1}^{i-1} U_j \cdot \pbra{U_i - W_i} \cdot \prod_{j = i+1}^{s+1} W_j + \prod_{i=1}^s U_i \pbra{U_{s+1} - W_{s+1}} \tag{Induction}\\
    &= \sum_{i=1}^{s+1} \prod_{j=1}^{i-1} U_j \cdot \pbra{U_i - W_i} \cdot \prod_{j = i+1}^{s+1} W_j.
\end{align*}
This completes the induction and the proof.
\end{proof}

\subsection{Proof of \Cref{thm:2D to 1D reduction}}
\label{sec:2D to 1D proof}

In this section we will prove \Cref{thm:2D to 1D reduction}, reducing the result to a spectral norm bound to be proved in \Cref{sec:spectralproof}. Our key insight will be that given the distribution and operator framework outlined in the previous section, we can now restate \Cref{thm:2D to 1D reduction} in a more ``palatable'' way by translating statements about the TV distance between distributions as quantities of their corresponding operators. More concretely we have the following:
\begin{claim}
    \label{tvdistancetolinearform}
    For $X \in \sD$, $d_{\mathrm{TV}}((\mcD_1)_X, (\mcD_2)_X) = \frac{1}{2} \sum_{Y \in \sD} \abs{\ip{e_X}{(T_{\mcD_1}-T_{\mcD_2}) e_Y}}$
\end{claim}

\begin{proof}We directly compute:
    \begin{align*}
        d_{\textrm{TV}}((\mcD_1)_X, (\mcD_2)_X) 
        &= \frac{1}{2} \sum_{Y \in \{\pm1\}^{nk}} \abs{\Pr[X \to_{T_{\mcD_1}} Y] - \Pr[X \to_{T_{\mcD_2}} Y]}\\
        &= \frac{1}{2} \sum_{Y \in \{\pm1\}^{nk}} \abs{\ip{e_X}{T_{\mcD_1} e_Y} - \ip{e_X}{T_{\mcD_2} e_Y}}\\
        &= \frac{1}{2} \sum_{Y \in \sD} \abs{\ip{e_X}{(T_{\mcD_1}-T_{\mcD_2}) e_Y}}.
    \end{align*}
    Here we used \Cref{fact:colorclasssizes} and that under a permutation elements in $\sD$ only map to $\sD$. 
\end{proof}

We apply this to prove \Cref{thm:2D to 1D reduction}:
\begin{theorem}[\Cref{thm:2D to 1D reduction} restated]
\label{thm:2D to 1D reduction technical}
    Given $k \leq 2^{\sqrt{n}/500}$, then for any $t \geq 500\pbra{k \log k + \frac{\log(1/\varepsilon)}{\sqrt{n}}}$, the following holds. Let $\mcP^t$ be the distribution on $\mathcal{S}_{\{\pm1\}^{n}}$ defined from our circuit model with a circuit family computing an $\frac{\varepsilon}{2\sqrt{n} \cdot (2t+1)}$-approximate $k$-wise independent distribution on $\mfS_{\{\pm1\}^{\sqrt{n}}}$ as the base. Then $\mcP^t$ is $\varepsilon$-approximate $k$-wise independent. That is, for all $X \in \sD$ we have $d_{\textrm{TV}}\pbra{\mcP^t_X, \mcG_X} \leq \varepsilon$ when $n$ is large enough.
\end{theorem}

\begin{proof}
We will bound the two quantities arising from the following application of the triangle inequality:
\begin{equation*}
    d_{\textrm{TV}}\pbra{\mcP^t_X, \mcG_X} \leq d_{\textrm{TV}}\pbra{\mcP^t_X, \mcG^t_X} + d_{\textrm{TV}}\pbra{\mcG^t_X, \mcG_X}.
\end{equation*}
Here we introduce the intermediate distribution $\mcG^t$ (defined in \Cref{subsec:distributions}) which will facilitate our analysis. The two steps are then to bound each of the latter terms separately by $\frac{\varepsilon}{2}$.

\begin{lemma}
    \label{lem:reduction}
    Assume the hypotheses of \Cref{thm:2D to 1D reduction}. Then for any $X \in \sD$:
\begin{equation*}
    \sum_{Y \in \sD} \abs{\ip{e_X}{(T_{\mcP^t}-T_{\mcG^t}) e_Y}} \leq  \max_{f \in \mcF} \abs{\ip{e_X}{(T_{\mcP^t}-T_{\mcG^t}) f}} \leq \frac\varepsilon2
\end{equation*}
where $\mcF$ is the set of functions $f : \{\pm1\}^{nk} \to [-1, 1]$ with $\mathrm{Supp}(f) \subseteq \sD$.
\end{lemma}

This lemma is a just a quantitative result capturing the intuition that if we replace all of the black-box $1$D pseudorandom permutation circuits in our design with truly random permutations, the TV distance does not change much.

\begin{lemma}
    \label{lem:maintrick}
    Assume the hypotheses of \Cref{thm:2D to 1D reduction}. Then for any $X \in \sD$:
\begin{equation*}
    \sum_{Y \in \sD} \abs{\ip{e_X}{(T_{\mcG^t}-T_{\mcG}) e_Y}} \leq \frac\varepsilon2.
\end{equation*}
\end{lemma}

This lemma claims that random columns and random row permutations approximate truly random permutations in low depth $t = \poly(k)$. In light of \Cref{tvdistancetolinearform}, we have shown that
\begin{equation*}
    d_{\textrm{TV}}\pbra{\mcP^t_X, \mcG^t_X} + d_{\textrm{TV}}\pbra{\mcG^t_X, \mcG_X} \leq \varepsilon,
\end{equation*}
which finishes the proof of \Cref{thm:2D to 1D reduction technical}, given the two lemmas.
\end{proof}

\subsection{Proof of \Cref{lem:reduction}}
\label{subsec:reduction}

In this proof, we roughly want to decompose the operators of our circuit as 1) sequential pieces corresponding to each layer of the circuit and 2) for each layer, since the pieces act independently on either the columns or rows, we can represent them as tensor products of simpler operators acting only on one row or column. Since these individual pieces are assumed to be approximately $k$-wise independent, we just show that this is preserved through tensorization and sequential application. 

Recall that $T_{\mcP^t} = T_{\mcP_R}\pbra{T_{\mcP_C}T_{\mcP_R}}^t$ and $T_{\mcG^t} = T_{\mcG_R}\pbra{T_{\mcG_C}T_{\mcG_R}}^t$, that is, they are products of $2t-1$ operators corresponding to the sequential pieces in the circuit. We can then rewrite the difference of the two operators as a telescoping sum following \Cref{fact:differenceofproducts}. For clarity let $T_{\mcP^{(i)}}$ denote the $i$th operator in the product $T_{\mcP^t}$ and likewise for $T_{\mcG^t}$.
\begin{equation*}
    T_{\mcP_R}\pbra{T_{\mcP_C}T_{\mcP_R}}^t - T_{\mcG_R}\pbra{T_{\mcG_C}T_{\mcG_R}}^t = \sum_{i=1}^{2t+1} \prod_{j=1}^{i-1} T_{\mcP^{(j)}} \cdot \pbra{T_{\mcP^{(i)}} - T_{\mcG^{(i)}}} \cdot \prod_{j = i+1}^s T_{\mcG^{(j)}}
\end{equation*}
To simplify this sum, we will make use of the following claim:

\begin{claim}\label{claim:sandwiching}
    Let $\mcD_1, \mcD_2$ be distributions on $\mfS_{\{\pm1\}^{n}}$ and $A$ an operator on the space $\R^{\{\pm1\}^{nk}}$. Let $X \in \sD$ and $f \in \mcF$. Then there exists some $X^* \in \sD$ and $f^* \in \mcF$ s.t.:
    \begin{equation*}
        \abs{\ip{e_X}{T_{\mcD_1}AT_{\mcD_2} f}} \leq \abs{\ip{e_{X^*}}{Af^*}}.
    \end{equation*}
\end{claim}
\begin{proof} Observe:
\begin{align*}
    \abs{\ip{e_X}{T_{\mcD_1}AT_{\mcD_2} f}} &= \abs{T_{\mcD_1}AT_{\mcD_2} f(X)}\\
    &= \abs{\sum_{Y \in \{\pm1\}^{nk}} \Pr[X \to_{T_{\mcD_1}} Y] \cdot AT_{\mcD_2} f(Y)}\\
    &= \abs{\sum_{Y \in \sD} \Pr[X \to_{T_{\mcD_1}} Y] \cdot \ip{e_Y}{AT_{\mcD_2} f}}\\
    &\leq \max_{X^* \in \sD} \abs{\ip{e_{X^*}}{AT_{\mcD_2} f}}.
\end{align*}
In the last line we are using triangle inequality. For the second part of the proof, it suffices to claim $T_{\mcD_2}f \in \mcF$. To see this just note:
\begin{equation*}
    T_{\mcD_2}f(Z) = \sum_{Y \in \mathrm{Supp}(\mcD_2(Z)))} \Pr[Z \to_{T_{\mcD_2}} Y] \cdot f(Y)
\end{equation*}
Since $f$ maps to $[-1, 1]$ which is a convex set, $T_{\mcD_2}f$ maps to $[-1, 1]$ as well. Additionally, for any $Z \notin \sD$, the support of $(\mcD_2)_Z$ cannot intersect $\sD$ so the term is 0, thus $T_{\mcD_2}f$ is supported on a subset of $\sD$ so is in $\mcF$. \end{proof}

\begin{claim}
    \begin{equation*}
        \max_{X \in \sD, f \in \mcF} \abs{\ip{e_X}{\pbra{T_{\mcP^t} - T_{\mcG^t}} f}} \leq \pbra{2t+1} \cdot \max_{X \in \sD, f \in \mcF} \abs{\ip{e_X}{\pbra{T_{\mcP_R} - T_{\mcG_R}} f}}.
    \end{equation*}
    Noting importantly that the latter term is interchangeable for $T_{\mcP_R} - T_{\mcG_R}$ and $T_{\mcP_C} - T_{\mcG_C}$. 
\end{claim}

\begin{proof}We directly compute:
    \begin{align*}
        \max_{X \in \sD, f \in \mcF} \abs{\ip{e_X}{\pbra{T_{\mcP^t} - T_{\mcG^t}} f}}
        &= \max_{X \in \sD, f \in \mcF} \abs{\ip{e_X}{\pbra{\sum_{i=1}^{2t+1} \prod_{j=1}^{i-1} T_{\mcP^{(j)}} \cdot \pbra{T_{\mcP^{(i)}} - T_{\mcG^{(i)}}} \cdot \prod_{j = i+1}^s T_{\mcG^{(j)}}} f}} \tag{\Cref{fact:differenceofproducts}}\\
        &\leq \max_{X \in \sD, f \in \mcF} \sum_{i=1}^{2t+1} \abs{\ip{e_X}{\pbra{ \prod_{j=1}^{i-1} T_{\mcP^{(j)}} \cdot \pbra{T_{\mcP^{(i)}} - T_{\mcG^{(i)}}} \cdot \prod_{j = i+1}^s T_{\mcG^{(j)}}} f}}\\
        &\leq \max_{X \in \sD, f \in \mcF} \sum_{i=1}^{2t+1} \abs{\ip{e_X}{\pbra{T_{\mcP^{(i)}} - T_{\mcG^{(i)}}} f}} \tag{\Cref{claim:sandwiching}}\\
        &\leq (2t+1) \cdot \max_{X \in \sD, f \in \mcF} \abs{\ip{e_X}{\pbra{T_{\mcP_R} - T_{\mcG_R}} f}}.
    \end{align*}
    In the last line we are using that $T_{\mcP_R}$ and $T_{\mcP_C}$ are symmetric, and likewise for $T_{\mcG_R}$ and $T_{\mcG_C}$.
\end{proof}

It is worth pointing out at this point we could convert back to the total variation distance making the above statement:
\begin{equation*}
    d_{\textrm{TV}}\pbra{\mcP^t_X, \mcG^t_X} \leq (2t+1) \cdot d_{\textrm{TV}}\pbra{(\mcP_R)_X, (\mcG_R)_X}.
\end{equation*}
We have in essence reduced the TV distance bound on our sequential circuit to just a single layer. Our next move will be to reduce the distance further to the individual parallel gates making up each layer, which is what we assumed black box is $\varepsilon'$-approximate $k$-wise independent. Towards this end we write:
\begin{equation*}
    \sum_{Y \in \sD} \abs{\Pr[X \to_{T_{\mcP_R}} Y] - \Pr[X \to_{T_{\mcG_R}} Y]} = \sum_{Y \in \sD} \abs{\prod_{i=1}^{\sqrt{n}}\Pr[X_{i, \cdot} \to_{T_\mcB} Y_{i ,\cdot}] - \prod_{i=1}^{\sqrt{n}}\Pr[X_{i, \cdot} \to_{T_{\mcG_{\sqrt{n}}}} Y_{i ,\cdot}]}.
\end{equation*}
We denote here $\mcG_{\sqrt{n}} = \mcU(\mfS_{\{\pm1\}^{\sqrt{n}}})$. The key fact here is that the operators correspond to product distributions on individual rows. We can again utilize \Cref{fact:differenceofproducts} to simplify the difference of products:
\begin{align*}
    &\sum_{Y \in \sD} \abs{\sum_{j = 1}^{\sqrt{n}} \prod_{i=1}^{j-1}\Pr[X_{i, \cdot} \to_{T_\mcB} Y_{i ,\cdot}] \cdot \pbra{\Pr[X_{j, \cdot} \to_{T_\mcB} Y_{j ,\cdot}] - \Pr[X_{j, \cdot} \to_{T_{\mcG_{\sqrt{n}}}} Y_{j ,\cdot}]} \cdot \prod_{i=j+1}^{\sqrt{n}}\Pr[X_{i, \cdot} \to_{T_{\mcG_{\sqrt{n}}}} Y_{i ,\cdot}]}\\
    \leq&  \sum_{j = 1}^{\sqrt{n}} \sum_{Y \in \sD} \abs{\prod_{i=1}^{j-1}\Pr[X_{i, \cdot} \to_{T_\mcB} Y_{i ,\cdot}] \cdot \pbra{\Pr[X_{j, \cdot} \to_{T_\mcB} Y_{j ,\cdot}] - \Pr[X_{j, \cdot} \to_{T_{\mcG_{\sqrt{n}}}} Y_{j ,\cdot}]} \cdot \prod_{i=j+1}^{\sqrt{n}}\Pr[X_{i, \cdot} \to_{T_{\mcG_{\sqrt{n}}}} Y_{i ,\cdot}]}\\
    =& \sum_{j = 1}^{\sqrt{n}} \sum_{y \in \{\pm1\}^{\sqrt{n}k}} \abs{\Pr[X_{j, \cdot} \to_{T_\mcB} y] - \Pr[X_{j, \cdot} \to_{T_{\mcG_{\sqrt{n}}}} y]} \sum_{\substack{Y \in \sD\\ Y_{j, \cdot} = y}} \prod_{i=1}^{j-1}\Pr[X_{i, \cdot} \to_{T_\mcB} Y_{i ,\cdot}] \prod_{i=j+1}^{\sqrt{n}}\Pr[X_{i, \cdot} \to_{T_{\mcG_{\sqrt{n}}}} Y_{i ,\cdot}]\\
    \leq& \sum_{j = 1}^{\sqrt{n}} \sum_{y \in \{\pm1\}^{\sqrt{n}k}} \abs{\Pr[X_{j, \cdot} \to_{T_\mcB} y] - \Pr[X_{j, \cdot} \to_{T_{\mcG_{\sqrt{n}}}} y]} \sum_{\substack{Y \in \{\pm1\}^{nk}\\ Y_{j, \cdot} = y}} \prod_{i=1}^{j-1}\Pr[X_{i, \cdot} \to_{T_\mcB} Y_{i ,\cdot}] \prod_{i=j+1}^{\sqrt{n}}\Pr[X_{i, \cdot} \to_{T_{\mcG_{\sqrt{n}}}} Y_{i ,\cdot}]\\
    =& \sum_{j = 1}^{\sqrt{n}} \sum_{y \in \{\pm1\}^{\sqrt{n}k}} \abs{\Pr[X_{j, \cdot} \to_{T_\mcB} y] - \Pr[X_{j, \cdot} \to_{T_{\mcG_{\sqrt{n}}}} y]}.
\end{align*}
Here we are partitioning the sum over $Y$ into its fixed row $Y_{j, \cdot}$. The large sum of products we get is just the probability the ``free'' rows map to different elements of $\{\pm1\}^{\sqrt{n}k}$, which marginalizes to 1 when we sum over the entire region. We are nearly done, as the term now looks very close to that which shows up in the definition of $\varepsilon$-approximate $k$-wise independent. The only difference is that we sum over the entire set of rows $\{\pm1\}^{\sqrt{n}k}$, whereas in the definition of $\varepsilon$-approximate $k$-wise independence, the sum is over ``distinct'' $k$-tuples. Distinct $k$-tuples of ``grids'' $X$ may share rows $X_{j, \cdot}$ that are not distinctly colored. Nonetheless, our result still follows from the definition of $\varepsilon$-approximate $k$-wise independence.

Using \Cref{kwiseimplies} below, we compute
\begin{equation*}
    \max_{f \in \mcF} \abs{\ip{e_X}{(T_{\mcP^t}-T_{\mcG^t}) f}} \leq (2t+1) \cdot \sum_{j=1}^{\sqrt{n}} \frac{\varepsilon}{2\sqrt{n} \cdot (2t+1)} \leq \frac\varepsilon2.
\end{equation*}
This completes the proof of \Cref{lem:reduction}.

\begin{lemma}
    \label{kwiseimplies}
    For every $x \in \{\pm1\}^{\sqrt{n}k}$ we have:
    \begin{equation*}
        \sum_{y \in \{\pm1\}^{\sqrt{n}k}} \abs{\Pr[x \to_{T_\mcB} y] - \Pr[x \to_{T_{\mcG_{\sqrt{n}}}} y]} \leq \frac{2\varepsilon}{\sqrt{n} \cdot (2t+1)}.
    \end{equation*}
\end{lemma}

\begin{proof}[Proof of \Cref{kwiseimplies}]
For $x$ corresponding to a distinct $k$-tuple, this reduces to the term in $\varepsilon'$-approximate $k$-wise independence, for which $\mcB$ is assumed to fulfill. The proof is then a matter of showing that ``$\varepsilon$-approximate $k$-wise independence'' implies ``$\varepsilon$-approximate $\tau$-wise independence'' for $\tau < k$. For this, we will need a notion of color class for elements in $\{\pm1\}^{\sqrt{n}k}$ analogous to the one defined in \Cref{sec:colorclasses}. We will define for an equivalence relation $R$ on $[k]$:
\begin{equation*}
    B_R = \left\{x \in \{\pm1\}^{\sqrt{n}k} \mid x^i = x^j \iff i R j\right\}.
\end{equation*}
That is, we think of $x$ as a $k$-tuple of rows and take the corresponding coloring.

First note that we only need to consider terms such that $y \in B(x)$, as if they are not colored the same then the transition probability under any permutation becomes 0. Now assume $x$ is $\tau$-colored for $\tau < k$ and $B = B(x)$. Let $T$ be the set of indices corresponding to the first instances of a color appearing in $x$. For example, if $x$ was colored with $k-1$ colors, with the first and last elements of the $k$-tuple colored the same, then $T$ would be $[k-1]$. Importantly, $T$ is the same across the color class $B(x)$ and $\abs{T} = \tau$. We will then create a function $\varphi_B : \{\pm1\}^{\sqrt{n}k} \to \{\pm1\}^{\sqrt{n}\tau}$ that projects out the indices outside of $T$. As a result we have for all $y \in B$, $\varphi_B(y) \in \sD_{\sqrt{n}}^{(\tau)}$, the set of distinct tuples in $\{\pm1\}^{\sqrt{n}\tau}$ and moreover the image of $\varphi_B$ under $B$ is entirely $\sD_{\sqrt{n}}^{(\tau)}$. The key observation is then that $[k] \setminus T$, the indices not in $T$, can be ignored across transitions since they are completely fixed:
\begin{align*}
    \sum_{y \in B(x)} \abs{\Pr[x \to_{T_\mcB} y] - \Pr[x \to_{T_{\mcG_{\sqrt{n}}}} y]} =& \sum_{y \in B} \abs{\Pr[\varphi_B(x) \to_{T_\mcB} \varphi_B(y)] - \Pr[\varphi_B(x) \to_{T_{\mcG_{\sqrt{n}}}} \varphi_B(y)]}\\
    =& \sum_{\varphi_B(y) \in \sD_{\sqrt{n}}^{(\tau)}} \abs{\Pr[\varphi_B(x) \to_{T_\mcB} \varphi_B(y)] - \Pr[\varphi_B(x) \to_{T_{\mcG_{\sqrt{n}}}} \varphi_B(y)]}.
\end{align*}

Note the end formula above has no dependence on the fixed indices $[k] \setminus T$. This allows us to pretend they are distinct, writing the above sum over elements of $\sD$ instead:
\begin{align*}
    &\sum_{y \in B(x)} \abs{\Pr[x \to_{T_\mcB} y] - \Pr[x \to_{T_{\mcG_{\sqrt{n}}}} y]}\\
    =& \sum_{\varphi_B(y) \in \sD_{\sqrt{n}}^{(\tau)}} \abs{\sum_{y_{[k] \setminus T} \in \sD_{\sqrt{n}}^{(k-\tau)}}\Pr[(\varphi_B(x), \cdot) \to_{T_\mcB} (\varphi_B(y), y_{[k] \setminus T})] - \Pr[\varphi_B(x) \to_{T_{\mcG_{\sqrt{n}}}} (\varphi_B(y), y_{[k] \setminus T})]}\\
    \leq& \sum_{\substack{\varphi_B(y) \in \sD_{\sqrt{n}}^{(\tau)} \\ y_{[k] \setminus T} \in \sD_{\sqrt{n}}^{(k-\tau)}}} \abs{\Pr[(\varphi_B(x), \cdot) \to_{T_\mcB} (\varphi_B(y), y_{[k] \setminus T})] - \Pr[\varphi_B(x) \to_{T_{\mcG_{\sqrt{n}}}} (\varphi_B(y), y_{[k] \setminus T})]}\\
    =& \sum_{y \in \sD^{(k)}_{\sqrt{n}}} \abs{\Pr[(\varphi_B(x), \cdot) \to_{T_\mcB} (\varphi_B(y), y_{[k] \setminus T})] - \Pr[\varphi_B(x) \to_{T_{\mcG_{\sqrt{n}}}} (\varphi_B(y), y_{[k] \setminus T})]}.
\end{align*}

In this last line we can choose $(\varphi_B(x), \cdot)$ to be from $\sD^{(k)}_{\sqrt{n}}$ so any $y$ outside of this class contributes nothing to the sum. Appealing to the approximate $k$-wise independence of $\mcB$ finishes the proof.
\end{proof}

\subsection{Proof of \Cref{lem:maintrick}}
\label{subsec:inductiontrick}

The following lemma will help us achieve the bound in \Cref{lem:maintrick}.
\begin{lemma}
\label[lemma]{lem:offdiagonalmoment}
    Assume the hypothesis of \Cref{lem:maintrick}. Then for any $Y \in\sD$, we have $\abs{\ip{e_X}{(T_{\mcG^t}-T_{\mcG}) e_Y}} \leq \frac{t+1}{2^{\sqrt{n}(t-1)/128}} \cdot \frac{1}{\abs{B(Y)}}$.
\end{lemma}

To see why the lemma is sufficient, observe:
\begin{equation*}
    \sum_{Y \in \sD} \abs{\ip{e_X}{(T_{\mcG^t}-T_{\mcG}) e_Y}} 
    \leq \frac{t+1}{2^{\sqrt{n}(t-1)/128}} \sum_{Y \in \sD} \frac{1}{\abs{B(Y)}} 
    \leq \frac{t+1}{2^{\sqrt{n}(t-1)/128}}\sum_{B \in \sB} \sum_{Y \in B} \frac{1}{\abs{B}} \leq \frac{\abs{\sB} \cdot (t+1)}{2^{\sqrt{n}(t-1)/128}}.
\end{equation*}
Here we partition the sum based on color classes, and note that each color class contributes a total of 1 to the sum. We can use \Cref{fact:numofcolorclasses} bounding the number of color classes and the fact that $\frac{t+1}{2^{t-1}}$ very quickly to write:
\begin{equation*}
    \sum_{Y \in \sD} \abs{\ip{e_X}{(T_{\mcG^t}-T_{\mcG}) e_Y}} 
    \leq \frac{k^{k\sqrt{n}} \cdot (t+1)}{2^{\sqrt{n}(t-1)/128}} \leq \frac{k^{k\sqrt{n}}}{2^{\pbra{\sqrt{n}/128-1} (t-1)}} \leq \frac{\varepsilon}{2}.
\end{equation*}

This is then bounded by $\frac\varepsilon2$ for $t \geq  \frac{\sqrt{n}k\log_2 k + \log_2 2/\varepsilon}{\sqrt{n}/128-1}+1$. When $n$ is large enough, the bound holds when $t\geq 500\pbra{k \log_2 k + \frac{\log_2 1/\varepsilon}{\sqrt{n}}}$.

\begin{proof}[Proof of \Cref{lem:offdiagonalmoment}.]
Recall that $T_{\mcG^t} = T_{\mcG_R}(T_{\mcG_C}T_{\mcG_R})^t$. We can then write $T_{\mcG^t}-T_{\mcG} = (T_{\mcG_R}T_{\mcG_C})^t(T_{\mcG_R} - T_\mcG)$ and prove the claim by induction on $t$. Consider first when $t = 0$.
\begin{equation*}
    \abs{\ip{e_X}{(T_{\mcG_R}-T_{\mcG}) e_Y}} = \abs{\Pr[X \to_{T_{\mcG_R}} Y] - \Pr[X \to_{T_{\mcG}} Y]} = \abs{\Pr[Y \to_{T_{\mcG_R}} X] - \Pr[Y \to_{T_{\mcG}} X]}.
\end{equation*}
We use the self-adjointness of the two operators here. Observe that under the action of $T_{\mcG_R}$, $Y$ goes to a uniform element of $B(Y)$, and under $T_\mcG$ goes to a uniform element of $\sD$. Thus, the quantity is either $\frac{1}{\abs{B(Y)}} - \frac{1}{\abs{\sD}}$ or $\frac{1}{\abs{\sD}}$. Either way it is below $\frac{1}{\abs{B(Y)}}$. For the induction step, we assume the lemma for fixed $t \geq 0$. Then we compute
\begin{align*}
    \abs{\ip{e_X}{(T_{\mcG_R}T_{\mcG_C})^{t+1}(T_{\mcG_R}-T_{\mcG}) e_Y}} &= \abs{\ip{e_X}{(T_{\mcG_R}T_{\mcG_C})(T_{\mcG_R}T_{\mcG_C})^t(T_{\mcG_R}-T_{\mcG}) e_Y}}\\
    &= \abs{T_{\mcG_R}\pbra{T_{\mcG_C}(T_{\mcG_R}T_{\mcG_C})^t(T_{\mcG_R}-T_{\mcG}) e_Y}(X)}\\
    &= \abs{\sum_{Z \in \sD} \Pr[X \to_{T_{\mcG_R}} Z]\pbra{T_{\mcG_C}(T_{\mcG_R}T_{\mcG_C})^t(T_{\mcG_R}-T_{\mcG}) e_Y}(Z)}\\
    &= \abs{\sum_{Z \in \sD} \Pr[X \to_{T_{\mcG_C}T_{\mcG_R}} Z]\pbra{(T_{\mcG_R}T_{\mcG_C})^t(T_{\mcG_R}-T_{\mcG}) e_Y}(Z)}\\
    &= \abs{\sum_{Z \in \sD} \Pr[X \to_{T_{\mcG_C}T_{\mcG_R}} Z] \ip{e_Z}{(T_{\mcG_R}T_{\mcG_C})^t(T_{\mcG_R}-T_{\mcG}) e_Y}}.
\end{align*}
We have managed to write the $(t+1)$ case inner product as a convex combination of the case with $t$. However, if we try to apply the induction hypothesis here we will make no progress. Instead, we will break the sum up and handle only one half with induction. The other half we will bound ``from scratch'', and it is here we will make progress. Recall that we may partition $\sD$ into two regions, $B_\text{safe}$ and $B_\text{coll}$:
\begin{align*}
    \abs{\ip{e_X}{(T_{\mcG_R}T_{\mcG_C})^{t+1}(T_{\mcG_R}-T_{\mcG}) e_Y}} \leq& \abs{\sum_{Z \in B_\text{safe}} \Pr[X \to_{T_{\mcG_C}T_{\mcG_R}} Z] \ip{e_Z}{(T_{\mcG_R}T_{\mcG_C})^t(T_{\mcG_R}-T_{\mcG}) e_Y}}\\
    &\;\;+ \abs{\sum_{Z \in B_\text{coll}} \Pr[X \to_{T_{\mcG_C}T_{\mcG_R}} Z] \ip{e_Z}{(T_{\mcG_R}T_{\mcG_C})^t(T_{\mcG_R}-T_{\mcG}) e_Y}}\\
    \leq& \abs{\sum_{Z \in B_\text{safe}} \Pr[X \to_{T_{\mcG_C}T_{\mcG_R}} Z] \ip{e_Z}{(T_{\mcG_R}T_{\mcG_C})^t(T_{\mcG_R}-T_{\mcG}) e_Y}}\\
    &\;\;+ \Pr[X \to_{T_{\mcG_C}T_{\mcG_R}} B_{\text{coll}}] \cdot \frac{t+1}{2^{\sqrt{n}(t-1)/128}} \cdot \frac{1}{\abs{B(Y)}}.
\end{align*}
In the last line we used the inductive hypothesis. We will then show that the first term is smaller than is demanded by the induction due to a straightforward spectral norm argument. The second term is small because the probability of ``collision'', or that a walk transitions to $B_{\mathrm{coll}}$, is small. More specifically we will need the following two lemmas which we will prove in \Cref{sec:spectralproof}.
\begin{lemma}\label{lem:spectralnorm}
    Assuming $k \leq 2^{\sqrt{n}/500}$ and $n$ large enough, $\norm{T_{\mcG_R}T_{\mcG_C}T_{\mcG_R}-T_{\mcG}}_{\mathrm{op}} \leq \frac{1}{2^{\sqrt{n}/128}}$.
\end{lemma}

\begin{lemma}
\label{lem:lowcollprob}
    Assuming $k \leq 2^{\sqrt{n}/500}$, for all $X \in \sD$, $\Pr[X \to_{T_{\mcG_C}T_{\mcG_R}} B_\text{coll}] \leq \frac{1}{2^{\sqrt{n}/128}}$.
\end{lemma}
To use \Cref{lem:spectralnorm} we write for $Z \in B_{\mathrm{safe}}$:
\begin{align*}
    \abs{\ip{e_Z}{(T_{\mcG_R}T_{\mcG_C})^t(T_{\mcG_R}-T_{\mcG}) e_Y}} =& \abs{\ip{T_{\mcG_R}e_Z}{(T_{\mcG_R}T_{\mcG_C}T_{\mcG_R}-T_{\mcG})^t T_{\mcG_R} e_Y}} \\
    \leq& \norm{T_{\mcG_R}T_{\mcG_C}T_{\mcG_R}-T_{\mcG}}{2}^t \norm{T_{\mcG_R}e_Z}_{2} \norm{T_{\mcG_R}e_Y}_{2}\\
    \leq& \frac{1}{2^{\sqrt{n}t/128}} \cdot \frac{1}{\abs{B(Y)}^{1/2}\abs{B_\text{safe}}^{1/2}}\\
    \leq& \frac{1}{2^{\sqrt{n}t/128}} \cdot \frac{1}{\abs{B(Y)}}.
\end{align*}
The first step uses the self-adjointness of $T_{\mcG_R}$, the fact that $T_{\mcG_R}^2 = T_{\mcG_R}$, and \Cref{fact:TGabsorbs}. The inequality is an application of Cauchy-Schwarz and submultiplicativity of the operator norm. The second to last step uses \Cref{lem:spectralnorm} and \Cref{TGR eU 2 norm} below, and the last step uses \Cref{fact:colorclasssizes}, namely that $B_{\mathrm{safe}}$ is larger than every other color class given our choice of $k$ and large enough $n$.

\begin{claim}\label{TGR eU 2 norm}
    For arbitrary $U \in \{\pm1\}^{nk}$:
\begin{equation*}
    \norm{T_{\mcG_R}e_U}_{2} = \frac{1}{\abs{B(U)}^{1/2}}.
\end{equation*}
\end{claim}

Continuing from the equation above we have:
\begin{align*}
    \abs{\ip{e_X}{(T_{\mcG_R}T_{\mcG_C})^{t+1}(T_{\mcG_R}-T_{\mcG}) e_Y}} \leq& \abs{\sum_{Z \in B_\text{safe}} \Pr[X \to_{T_{\mcG_C}T_{\mcG_R}} Z] \ip{e_Z}{(T_{\mcG_R}T_{\mcG_C})^t(T_{\mcG_R}-T_{\mcG}) e_Y}}\\
    &\;\;+ \Pr[X \to_{T_{\mcG_C}T_{\mcG_R}} B_{\text{coll}}] \cdot \frac{t+1}{2^{\sqrt{n}(t-1)/128}} \cdot \frac{1}{\abs{B(Y)}}\\
    \leq& \frac{1}{2^{\sqrt{n}t/128}} \cdot \frac{1}{\abs{B(Y)}} \tag{\Cref{lem:spectralnorm}}\\
    &\;\;+ \Pr[X \to_{T_{\mcG_C}T_{\mcG_R}} B_{\text{coll}}] \cdot \frac{t+1}{2^{\sqrt{n}(t-1)/128}} \cdot \frac{1}{\abs{B(Y)}}\\
    \leq& \frac{1}{2^{\sqrt{n}t/128}} \cdot \frac{1}{\abs{B(Y)}} + \frac{t+1}{2^{\sqrt{n}t/128}} \cdot \frac{1}{\abs{B(Y)}} \tag{\Cref{lem:lowcollprob}}\\
    \leq& \frac{t+2}{2^{\sqrt{n}t/128}} \cdot \frac{1}{\abs{B(Y)}}.
\end{align*}
This completes the induction. We finish by proving the claim above:

\begin{proof}[Proof of \Cref{TGR eU 2 norm}]
    Observe:
    \begin{align*}
        &\norm{T_{\mcG_R}e_U}_{2} = \sqrt{\sum_{W \in B(U)} \pbra{T_{\mcG_R}e_U(W)}^2} = \sqrt{\sum_{W \in B(U)} \Pr[W \to_{\mcG_R} U]^2} \\
        =& \sqrt{\sum_{W \in B(U)} \pbra{\frac{1}{\abs{B(U)}}}^2} = \frac{1}{\abs{B(U)}^{1/2}}.\qedhere
    \end{align*}
\end{proof}
This completes the proof of \Cref{lem:offdiagonalmoment}.
\end{proof} 

\section{Proof of Spectral Properties of $T_{\mcG_R}$ and $T_{\mcG_C}$}
\label{sec:spectralproof}

In this section we prove \Cref{lem:spectralnorm}, which is a spectral norm bound on the difference between two operators related to our constructions above. As an intermediate result in the proof we will show \Cref{lem:lowcollprob} as well. This will then conclude the proof of \Cref{thm:2D to 1D reduction}.

\begin{lemma}[\Cref{lem:spectralnorm} restated]
    Assuming $k \leq 2^{\sqrt{n}/500}$, we have for large enough $n$,
    \begin{equation*}
        \norm{T_{\mcG_R}T_{\mcG_C}T_{\mcG_R} - T_\mcG}_{\mathrm{op}} \leq \frac{1}{2^{\sqrt{n}/128}},
    \end{equation*}
    or rather for $f : \{\pm 1 \}^{nk} \to \R$:
    \begin{equation*}
        \ip{f}{(T_{\mcG_R}T_{\mcG_C}T_{\mcG_R} - T_\mcG) f} \leq \frac{1}{2^{\sqrt{n}/128}}\cdot \ip{f}{f}.
    \end{equation*}
\end{lemma}

Note that it suffices to prove maximization across symmetric linear forms because the operator is self-adjoint. We will proceed by decomposing $f = f_{B_{\mathrm{safe}}} + f_{B_{\mathrm{coll}}} + f_{B_{=0}}$ where $f_S$ is supported on $S \subseteq \{\pm1\}^{nk}$. Note that these regions form a partition of $\{\pm1\}^{nk}$, so these functions are orthogonal to one another.
\begin{align*}
    \abs{\ip{f}{(T_{\mcG_R}T_{\mcG_C}T_{\mcG_R} - T_\mcG) f}} \leq& \abs{\ip{f_{B_{\mathrm{safe}}}}{(T_{\mcG_R}T_{\mcG_C}T_{\mcG_R} - T_\mcG) f_{\sD}}}\\
    &+ \abs{\ip{f_{B_{\mathrm{coll}}}}{(T_{\mcG_R}T_{\mcG_C}T_{\mcG_R} - T_\mcG) f_\sD}}\\
    &+ \abs{\ip{f_{B_{=0}}}{(T_{\mcG_R}T_{\mcG_C}T_{\mcG_R} - T_\mcG) f_{B_{=0}}}}.
\end{align*}
Note that the cross terms involving $B_{=0}$ are all zero, as a permutation will not cross between these regions. Our proof will bound each of these terms separately.

\subsection{$f$ Supported on $B_{\textrm{safe}}$}

\begin{lemma}
    $\abs{\ip{f_{B_{\mathrm{safe}}}}{(T_{\mcG_R}T_{\mcG_C}T_{\mcG_R} - T_\mcG) f_{\sD}}} \leq \frac{4\sqrt{n}k^2}{2^{\sqrt{n}}} \cdot \ip{f}{f}$.
\end{lemma}
\begin{proof}
Let $X \in B_{\mathrm{safe}}$, $g : \{\pm1\}^{nk} \to \R$.
\begin{align*}
    (T_{\mcG_R} - T_\mcG)g(X ) &= \sum_{Y \in \{\pm1\}^{nk}} \Pr[X \to_{T_{\mcG_R}} Y] \cdot g(Y) - \sum_{Y \in \{\pm1\}^{nk}} \Pr[X \to_{T_{\mcG}} Y] \cdot g(Y)\\
    &= \frac{1}{\abs{B_{\mathrm{safe}}}}\sum_{Y \in B_{\mathrm{safe}}} g(Y) - \frac{1}{\abs{\sD}}\sum_{Y \in \sD} g(Y)\\
    &= \pbra{\frac{1}{\abs{B_{\mathrm{safe}}}} - \frac{1}{\abs{\sD}}}\sum_{Y \in B_{\mathrm{safe}}} g(Y) - \frac{1}{\abs{\sD}}\sum_{Y \in B_{\mathrm{coll}}} g(Y)\\
    &= \pbra{1 - \frac{\abs{B_{\mathrm{safe}}}}{\abs{\sD}}} \cdot \frac{1}{\abs{B_{\mathrm{safe}}}}\sum_{Y \in B_{\mathrm{safe}}} g(Y) - \frac{\abs{B_{\mathrm{coll}}}}{\abs{\sD}} \cdot \frac{1}{\abs{B_{\mathrm{coll}}}}\sum_{Y \in B_{\mathrm{coll}}} g(Y)\\
    &= \frac{\abs{B_{\mathrm{coll}}}}{\abs{\sD}} \pbra{T_{\mcG_R} - \mcH}g(X).
\end{align*}

Our definition of $\mcH g(X) = \frac{1}{\abs{B_{\mathrm{coll}}}}\sum_{Y \in B_{\mathrm{coll}}} g(Y)$ corresponds to the random walk operator that puts all probability weight into $B_{\mathrm{coll}}$ uniformly. Note that $\mcH$ does not correspond to any random walk induced by a distribution on $\mfS_{\{\pm1\}^{nk}}$ (so it cannot be written as $T_{\mcH}$), but is still a random walk operator. With this in hand we may write:
\begin{align*}
    \abs{\ip{f_{B_{\mathrm{safe}}}}{(T_{\mcG_R}T_{\mcG_C}T_{\mcG_R} - T_\mcG) f_{\sD}}} &= \sum_{X \in \{\pm1\}^{nk}} f_{B_{\mathrm{safe}}}(X) \cdot (T_{\mcG_R}-T_\mcG)(T_{\mcG_C}T_{\mcG_R}f_{\sD})(X)\\
    &= \sum_{X \in B_{\mathrm{safe}}} f_{B_{\mathrm{safe}}}(X) \cdot (T_{\mcG_R}-T_\mcG)(T_{\mcG_C}T_{\mcG_R}f_{\sD})(X)\\
    &= \frac{\abs{B_{\mathrm{coll}}}}{\abs{\sD}} \sum_{X \in B_{\mathrm{safe}}} f_{B_{\mathrm{safe}}}(X) \cdot (T_{\mcG_R}-\mcH)(T_{\mcG_C}T_{\mcG_R}f_{\sD})(X)\\
    &= \frac{\abs{B_{\mathrm{coll}}}}{\abs{\sD}} \abs{\ip{f_{B_{\mathrm{safe}}}}{(T_{\mcG_R}T_{\mcG_C}T_{\mcG_R} - \mcH T_{\mcG_C}T_{\mcG_R}) f_{\sD}}}.
\end{align*}

The only important fact about $\mcH$ is that it is a valid random walk operator, which allows us to use \Cref{lem:escape probs} to bound this final inner product crudely as:
\begin{align*}
    \abs{\ip{f_{B_{\mathrm{safe}}}}{(T_{\mcG_R}T_{\mcG_C}T_{\mcG_R} - T_\mcG) f_{\sD}}} &\leq \frac{\abs{B_{\mathrm{coll}}}}{\abs{\sD}} \pbra{\abs{\ip{f_{B_{\mathrm{safe}}}}{(T_{\mcG_R}T_{\mcG_C}T_{\mcG_R} f_{\sD}}} + \abs{\ip{f_{B_{\mathrm{safe}}}}{\mcH T_{\mcG_C}T_{\mcG_R}) f_{\sD}}}}\\
    &\leq \frac{2\abs{B_{\mathrm{coll}}}}{\abs{\sD}} \norm{f_{B_{\mathrm{safe}}}}_{2} \norm{f_\sD}_{2} \tag{\Cref{lem:escape probs}}\\ 
    &\leq \frac{2\abs{B_{\mathrm{coll}}}}{\abs{\sD}} \langle f, f \rangle.
\end{align*}
\Cref{fact:colorclasssizes} then suffices to prove the claim. \end{proof}

\subsection{$f$ Supported on $B_{\textrm{coll}}$}

\begin{lemma}
    \label{lem:collcase}
    $\abs{\ip{f_{B_{\mathrm{coll}}}}{(T_{\mcG_R}T_{\mcG_C}T_{\mcG_R} - T_\mcG) f_\sD}} \leq \frac{8\sqrt{n}k^2}{2^{\sqrt{n}/32}} \ip{f}{f}$.
\end{lemma}

\begin{proof}
First, we can decompose $f_\sD = f_{B_{\mathrm{safe}}} + f_{B_{\mathrm{coll}}}$ and bound:
\begin{align*}
    &\abs{\ip{f_{B_{\mathrm{coll}}}}{(T_{\mcG_R}T_{\mcG_C}T_{\mcG_R} - T_\mcG) f_\sD}} \\
    \leq &\abs{\ip{f_{B_{\mathrm{coll}}}}{(T_{\mcG_R}T_{\mcG_C}T_{\mcG_R} - T_\mcG) f_{B_{\mathrm{safe}}}}} + \abs{\ip{f_{B_{\mathrm{coll}}}}{(T_{\mcG_R}T_{\mcG_C}T_{\mcG_R} - T_\mcG) f_{B_{\mathrm{coll}}}}}.
\end{align*}
By the self-adjointness of the operator, the first term is bounded by the case above, so it suffices to bound the latter. For this term, we can appeal directly to \Cref{lem:escape probs} and the triangle inequality to get:
\begin{align*}
    &\abs{\ip{f_{B_{\mathrm{coll}}}}{(T_{\mcG_R}T_{\mcG_C}T_{\mcG_R} - T_\mcG) f_{B_{\mathrm{coll}}}}} \\
    \leq &\sqrt{\max_{X \in B_{\mathrm{coll}}} \Pr[X \to_{T_{\mcG_R}T_{\mcG_C}T_{\mcG_R}} B_{\mathrm{coll}}] + \max_{X \in B_{\mathrm{coll}}} \Pr[X \to_{T_\mcG} B_{\mathrm{coll}}]} \ip{f_{B_{\mathrm{coll}}}}{f_{B_{\mathrm{coll}}}}.
\end{align*}
Note that regardless of choice of $X$, the latter probability $\Pr[X \to_{T_\mcG} B_{\mathrm{coll}}] = \frac{\abs{B_{\mathrm{coll}}}}{\abs{\sD}}$ which is less than $\frac{2\sqrt{n}k^2}{2^{\sqrt{n}}}$ by \Cref{fact:colorclasssizes}. For the former, we will need a slightly more detailed analysis which also serves as the proof of \Cref{lem:lowcollprob} in the previous section:

\begin{lemma}[Restatement of \Cref{lem:lowcollprob}]
    \label{lem:relowcollprob}
    For all $X \in \sD$, $\Pr[X \to_{T_{\mcG_C}T_{\mcG_R}} B_{\mathrm{coll}}] \leq \frac{2\sqrt{n}k^2}{2^{\sqrt{n}/16}}$.
\end{lemma}

The lemma is immediately sufficient to achieve the bound in \Cref{lem:collcase}.
\end{proof}

\begin{proof}[Proof of \Cref{lem:relowcollprob}]
We will apply a union bound over the probability of any pair of rows ``colliding'', which would put them in $B_{\mathrm{coll}}$. Let $X \in \sD$. We will model our process as:
\begin{equation*}
    X \to_{T_{\mcG_R}} \bm{Y} \to_{T_{\mcG_C}} \bm{Z}.
\end{equation*}
We then fix $\bm{Z}^\ell_{i, \cdot}$ and $\bm{Z}^m_{i, \cdot}$ for $i \in [\sqrt{n}]$, $\ell \neq m \in [k]$. The only fact we will use about $\bm{Y}$ is that for some $j \in [k]$ (potentially equal to $i$), we have $(\bm{Y}^\ell_{j, \cdot}, \bm{Y}^m_{j, \cdot})$ are uniform from ${\{\pm1\}^{\sqrt{n}} \choose 2}$. To see this note that there must exist some $j$ s.t. $X^\ell_{j, \cdot} \neq X^m_{j,\cdot}$, otherwise $X \notin \sD$. Since the permutation applied to these two rows is uniform from $\mfS_{\{\pm1\}^{\sqrt{n}}}$, the resulting rows in $Y$ look like a uniform unequal pair.

With this in mind, we will now condition on the event that $d(\bm{Y}^\ell_{j, \cdot}, \bm{Y}^m_{j, \cdot}) \geq \sqrt{n}/4$ (distance here is Hamming distance), allowing us to split our analysis into two cases:
\begin{align*}
    \Pr[\bm{Z}^\ell_{i, \cdot} = \bm{Z}_{i, \cdot}^m]  \leq&\Pr[\bm{Z}^\ell_{i, \cdot} = \bm{Z}^m_{i, \cdot} \mid d(\bm{Y}^\ell_{j, \cdot}, \bm{Y}^m_{j, \cdot}) > \sqrt{n}/4] + \Pr[d(\bm{Y}^\ell_{j, \cdot}, \bm{Y}^m_{j, \cdot}) \leq \sqrt{n}/4]\\
    \leq&\frac1{2^{\sqrt{n}/4}}+\frac1{e^{\sqrt{n}/16}}\tag{\Cref{lem:given Y far 2D}, \Cref{lem:Y probably far 2D}}\\
    \leq& \frac2{2^{\sqrt{n}/16}}.
\end{align*}
Applying a union bound over $\leq \sqrt{n}k^2$ rows completes the proof.
\end{proof}

\begin{lemma}\label{lem:given Y far 2D}
    $\Pr[\bm{Z}^\ell_{i, \cdot} = \bm{Z}^m_{i, \cdot} \mid d(\bm{Y}^\ell_{j, \cdot}, \bm{Y}^m_{j, \cdot}) > \sqrt{n}/4] \leq \frac{1}{2^{\sqrt{n}/4}}$.
\end{lemma}
\begin{proof}
The probability that $\bm{Z}^\ell_{i, \cdot}$ and $\bm{Z}^m_{i, \cdot}$ are equal can be viewed as the probability that all of their individual bits are equal, and they are all independent since they come from independently sampled column permutations. Since $\bm{Y}^\ell_{j, \cdot}$ and $\bm{Y}^m_{j, \cdot}$ differ in at least $\sqrt{n}/4$ places, $\bm{Y}^\ell$ and $\bm{Y}^m$ must differ in at least that many columns. In these columns, it can be seen that the corresponding bits in $\bm{Z}^\ell_{i, \cdot}$ and $Z^m_{i, \cdot}$ are the same with probability $\leq \frac{1}{2}$ (the probability is exactly one half when the columns are sampled uniformly independently, conditioning that they are unequal only lowers this probability). By independence the probability is less than $\frac{1}{2^{\sqrt{n}/4}}$. \end{proof}

\begin{lemma}\label{lem:Y probably far 2D}
    $\Pr[d(\bm{Y}^\ell_{j, \cdot}, \bm{Y}^m_{j, \cdot}) \leq \sqrt{n}/4] \leq \frac{1}{e^{\sqrt{n}/16}}$.
\end{lemma}
\begin{proof}
Note that $\Pr[d(\bm{Y}^\ell_{j, \cdot},\bm{Y}^m_{j, \cdot}) \leq \sqrt{n}/4] \leq \Pr_{\bx, \by \sim \{\pm1\}^{\sqrt{n}}}[d(\bx, \by) \leq \sqrt{n}/4]$. For uniform $x,y$, the random variable $d(x,y)$ is the sum of $\sqrt{n}$ independent Bernoulli$(1/2)$ random variables. This has expectation $\sqrt{n}/2$ and thus by Hoeffding's Inequality:
\begin{equation*}
    \Pr_{\bx,\by \sim \{\pm1\}^{\sqrt{n}}}[d(\bx,\by) \leq \sqrt{n}/4] \leq e^{-\sqrt{n}/16}.\qedhere
\end{equation*}
\end{proof}

\subsection{The Induction Case}

\begin{lemma} \label{lemmaInductionCase}
Let $f : \{\pm1\}^{nk} \to \R$ be supported on $B_{=0}$ for $k \geq 2$. Then, we have
\[ 
\abs{\ip{f}{\pbra{T_{\mcG_R}^{(k)}T_{\mcG_C}^{(k)}T_{\mcG_R}^{(k)} - T_\mcG^{(k)}}f}} \le \norm{T_{\mcG_R}^{(k-1)}T_{\mcG_C}^{(k-1)}T_{\mcG_R}^{(k-1)} - T_{\mcG}^{(k-1)}}_{\mathrm{op}} \ip{f}{f}. 
\]
\end{lemma}

\begin{proof}
    The proof follows from the same proof as the proof of \Cref{lemma:f supported on B0}.
\end{proof}

\subsection{Wrapping Up}

Putting together all three cases we have:
\begin{equation*}
    \abs{\ip{f}{(T_{\mcG_R}T_{\mcG_C}T_{\mcG_R} - T_\mcG) f}} \leq \frac{4\sqrt{n}k^2}{2^{\sqrt{n}}}\ip{f}{f} + \frac{8\sqrt{n}k^2}{2^{\sqrt{n}/32}} \ip{f}{f} + \norm{T_{\mcG_R}^{(k-1)}T_{\mcG_C}^{(k-1)}T_{\mcG_R}^{(k-1)} - T_\mcG^{(k-1)}}_{\mathrm{op}}\ip{f}{f}.
\end{equation*}

Since $\norm{T_{\mcG_R}^{(1)}T_{\mcG_C}^{(1)}T_{\mcG_R}^{(1)} - T_\mcG^{(1)}}_{\mathrm{op}} = 0$ and by assumption $k \leq 2^{\sqrt{n}/500}$, it follows by induction that:
\[ \norm{T_{\mcG_R}^{(k)}T_{\mcG_C}^{(k)}T_{\mcG_R}^{(k)} - T_\mcG^{(k)}}_{\mathrm{op}} \le \sum_{\ell = 2}^k \pbra{\frac{4\sqrt{n}\ell^2}{2^{\sqrt{n}}} + \frac{8\sqrt{n}\ell^2}{2^{\sqrt{n}/32}}} \leq \frac{k^3}{2^{\sqrt{n}/64}} \leq \frac{1}{2^{\sqrt{n}/128}}. \]

\section*{Acknowledgments}
W.H. thanks Angelos Pelecanos and Lucas Gretta for helpful discussions on the comparison method, and for allowing him to reuse results from their discussion and \cite{gretta2024more} in \Cref{appendix:comparison}.

\bibliographystyle{alpha}
\bibliography{references}

\newcommand{\etalchar}[1]{$^{#1}$}
\begin{thebibliography}{HMMH{\etalchar{+}}23}

\bibitem[AALV09]{aharonov2009detectability}
Dorit Aharonov, Itai Arad, Zeph Landau, and Umesh Vazirani.
\newblock {The Detectability Lemma and Quantum Gap Amplification}.
\newblock In {\em Proceedings of the Forty-First Annual ACM Symposium on Theory of Computing}, pages 417--426, 2009.

\bibitem[BH08]{brodsky2008simple}
Alex Brodsky and Shlomo Hoory.
\newblock {Simple Permutations Mix Even Better}.
\newblock {\em Random Structures \& Algorithms}, 32(3):274--289, 2008.

\bibitem[BHH16]{brandao2016local}
Fernando~G.S.L. Brandao, Aram~W. Harrow, and Micha{\l} Horodecki.
\newblock {Local Random Quantum Circuits are Approximate Polynomial-Designs}.
\newblock {\em Communications in Mathematical Physics}, 346:397--434, 2016.

\bibitem[CDX{\etalchar{+}}24]{chen2024efficient}
Chi-Fang Chen, Jordan Docter, Michelle Xu, Adam Bouland, and Patrick Hayden.
\newblock {Efficient Unitary T-designs from Random Sums}.
\newblock {\em arXiv preprint arXiv:2402.09335}, 2024.

\bibitem[CHH{\etalchar{+}}24]{chen2024incompressibility}
Chi-Fang Chen, Jeongwan Haah, Jonas Haferkamp, Yunchao Liu, Tony Metger, and Xinyu Tan.
\newblock Incompressibility and spectral gaps of random circuits.
\newblock {\em arXiv preprint arXiv:2406.07478}, 2024.

\bibitem[DCEL09]{dankert2009exact}
Christoph Dankert, Richard Cleve, Joseph Emerson, and Etera Livine.
\newblock Exact and approximate unitary 2-designs and their application to fidelity estimation.
\newblock {\em Physical Review A—Atomic, Molecular, and Optical Physics}, 80(1):012304, 2009.

\bibitem[FI24]{feng2024dynamics}
Xiaozhou Feng and Matteo Ippoliti.
\newblock Dynamics of pseudoentanglement.
\newblock {\em arXiv preprint arXiv:2403.09619}, 2024.

\bibitem[GHP24]{gretta2024more}
Lucas Gretta, William He, and Angelos Pelecanos.
\newblock More efficient approximate $ k $-wise independent permutations from random reversible circuits via log-sobolev inequalities.
\newblock {\em Cryptology ePrint Archive}, 2024.

\bibitem[Gow96]{gowers1996almost}
W.T. Gowers.
\newblock {An Almost m-wise Independent Random Permutation of the Cube}.
\newblock {\em Combinatorics, Probability and Computing}, 5(2):119--130, 1996.

\bibitem[HHJ21]{haferkamp2021improved}
Jonas Haferkamp and Nicholas Hunter-Jones.
\newblock {Improved Spectral Gaps for Random Quantum Circuits: Large Local Dimensions and All-to-All Interactions}.
\newblock {\em Physical Review A}, 104(2):022417, 2021.

\bibitem[HJ19]{hunter2019unitary}
Nicholas Hunter-Jones.
\newblock {Unitary Designs from Statistical Mechanics in Random Quantum Circuits}.
\newblock {\em arXiv preprint arXiv:1905.12053}, 2019.

\bibitem[HKP20]{huang2020predicting}
Hsin-Yuan Huang, Richard Kueng, and John Preskill.
\newblock Predicting many properties of a quantum system from very few measurements.
\newblock {\em Nature Physics}, 16(10):1050--1057, 2020.

\bibitem[HM23]{harrow2023approximate}
Aram~W Harrow and Saeed Mehraban.
\newblock Approximate unitary t-designs by short random quantum circuits using nearest-neighbor and long-range gates.
\newblock {\em Communications in Mathematical Physics}, 401(2):1531--1626, 2023.

\bibitem[HMMH{\etalchar{+}}23]{haferkamp2023efficient}
Jonas Haferkamp, Felipe Montealegre-Mora, Markus Heinrich, Jens Eisert, David Gross, and Ingo Roth.
\newblock {Efficient Unitary Designs with a System-Size Independent Number of non-Clifford Gates}.
\newblock {\em Communications in Mathematical Physics}, 397(3):995--1041, 2023.

\bibitem[HMMR05]{hoory2005simple}
Shlomo Hoory, Avner Magen, Steven Myers, and Charles Rackoff.
\newblock {Simple Permutations Mix Well}.
\newblock {\em Theoretical Computer Science}, 348(2-3):251--261, 2005.

\bibitem[HP07]{hayden2007black}
Patrick Hayden and John Preskill.
\newblock {Black Holes as Mirrors: Quantum Information in Random Subsystems}.
\newblock {\em {Journal of High Energy Physics}}, 2007(09):120, 2007.

\bibitem[Kas07]{kassabov2007symmetric}
Martin Kassabov.
\newblock {Symmetric Groups and Expander Graphs}.
\newblock {\em Inventiones Mathematicae}, 170(2):327--354, 2007.

\bibitem[KNR09]{kaplan2009derandomized}
Eyal Kaplan, Moni Naor, and Omer Reingold.
\newblock {Derandomized Constructions of k-wise (almost) Independent Permutations}.
\newblock {\em Algorithmica}, 55(1):113--133, 2009.

\bibitem[LPTV23]{liu2023layout}
Tianren Liu, Angelos Pelecanos, Stefano Tessaro, and Vinod Vaikuntanathan.
\newblock {Layout Graphs, Random Walks and the t-Wise Independence of SPN Block Ciphers}.
\newblock In {\em Annual International Cryptology Conference}, pages 694--726. Springer, 2023.

\bibitem[LTV21]{liu2021t}
Tianren Liu, Stefano Tessaro, and Vinod Vaikuntanathan.
\newblock The t-wise independence of substitution-permutation networks.
\newblock In {\em Advances in Cryptology--CRYPTO 2021: 41st Annual International Cryptology Conference, CRYPTO 2021, Virtual Event, August 16--20, 2021, Proceedings, Part IV 41}, pages 454--483. Springer, 2021.

\bibitem[MH24]{ma2024construct}
Fermi Ma and Hsin-Yuan Huang.
\newblock How to construct random unitaries.
\newblock {\em arXiv preprint arXiv:2410.10116}, 2024.

\bibitem[MOP20]{mohanty2020explicit}
Sidhanth Mohanty, Ryan O'Donnell, and Pedro Paredes.
\newblock {Explicit Near-Ramanujan Graphs of Every Degree}.
\newblock In {\em Proceedings of the 52nd Annual ACM SIGACT Symposium on Theory of Computing}, pages 510--523, 2020.

\bibitem[MP04]{maurer2004composition}
Ueli Maurer and Krzysztof Pietrzak.
\newblock {Composition of Random Systems: When Two Weak Make One Strong}.
\newblock In {\em Theory of Cryptography Conference}, pages 410--427. Springer, 2004.

\bibitem[MPSY24]{metger2024simple}
Tony Metger, Alexander Poremba, Makrand Sinha, and Henry Yuen.
\newblock Simple constructions of linear-depth t-designs and pseudorandom unitaries.
\newblock {\em arXiv preprint arXiv:2404.12647}, 2024.

\bibitem[Nac96]{nachtergaele1996spectral}
Bruno Nachtergaele.
\newblock {The Spectral Gap for Some Spin Chains with Discrete Symmetry Breaking}.
\newblock {\em Communications in Mathematical Physics}, 175:565--606, 1996.

\bibitem[O'D14]{ODonnell2014}
Ryan O'Donnell.
\newblock {\em Analysis of {B}oolean {F}unctions}.
\newblock Cambridge University Press, 2014.

\bibitem[OSP23]{o2023explicit}
Ryan O’Donnell, Rocco~A. Servedio, and Pedro Paredes.
\newblock {Explicit Orthogonal and Unitary Designs}.
\newblock In {\em 2023 IEEE 64th Annual Symposium on Foundations of Computer Science (FOCS)}, pages 1240--1260. IEEE, 2023.

\bibitem[Sco08]{scott2008optimizing}
Andrew~James Scott.
\newblock Optimizing quantum process tomography with unitary 2-designs.
\newblock {\em Journal of Physics A: Mathematical and Theoretical}, 41(5):055308, 2008.

\bibitem[WLP09]{wilmer2009markov}
E.L. Wilmer, David~A. Levin, and Yuval Peres.
\newblock {Markov Chains and Mixing Times}.
\newblock {\em American Mathematical Soc., Providence}, 2009.

\end{thebibliography}

\appendix
\section{Comparison Method}\label{appendix:comparison}

\begin{theorem}[\cite{wilmer2009markov}, Theorem 13.23]
    \label{thm:wilmer comparison}
    Let $\wt{P}$ and $P$ be transition matrices for two ergodic Markov chains on the same state space $V$. Assume that for each $(x,y)\in V^2$ there exists a random path 
    \begin{align*}
        \bm{\Delta}(x,y)=\pbra{(x,\bm{u}_1),(\bm{u}_1,\bm{u}_2),(\bm{u}_2,\bm{u}_3),\dots,(\bm{u}_{\ell},y)}.
    \end{align*}
    Then we have that
    \begin{align*}
        \lambda_2(L)\geq \pbra{\max_{v\in V}\frac{\pi(v)}{\widetilde{\pi}(v)}}A(\bm\Delta)\lambda_2(\wt{L}).
    \end{align*}
    where the comparison constant of $\bm{\Delta}$ is defined to be
    \begin{align*}
        A(\bm{\Delta}) \coloneq \max_{\substack{(a,b)\in V^2 \\\wt{P}(a,b)>0}} \cbra{\frac{1}{\wt{\pi}(x)\wt{P}(a,b)}\sum_{(x,y)\in V^2}\Ex_{\bm{\Delta}}\sbra{\mathbf{1}_{(a, b) \in \bm{\Delta}(x,y)}\cdot |\bm{\Delta}(x,y)|}\cdot  \pi(x)\cdot P(x, y)}.
    \end{align*}
    Here $\pi$ and $\wt{\pi}$ are the (unique) stationary distributions for $P$ and $\wt{P}$, respectively, and $\mathbf{1}_{(a, b) \in \calP}$ is the indicator variable which captures whether $(a,b)$ appears in the sequence $\calP$.
\end{theorem}

\Cref{lem:comparison schreier} is a direct consequence of the following \Cref{cor:comparison random walks}, with symmetry applied.

\begin{corollary}
\label{cor:comparison random walks}
    Let $\wt{P}$ and $P$ be transition matrices for random walks on undirected (multi)graphs $\wt{G}=(V,\wt{E})$ and $G=(V,E)$, respectively. The graph $G$ is $d$-regular, and $\wt{G}$ is $\wt{d}$-regular. Assume that for each $e\in E$ there exists a random path
    \begin{align*}
        \bm{\Gamma}(e)=(\wt{\bm{e}}_1,\dots,\wt{\bm{e}}_{T_e}),
    \end{align*}
    where $(\wt{\bm{e}}_1,\dots,\wt{\bm{e}}_{T_e})$\footnote{Here $T_e$ is a (deterministic) quantity determined by the edge $e$.} is drawn from a distribution on sequences of edges connecting the endpoints of $E$. Then we have for any $f:V\to\R$ that
    \begin{align*}
        \lambda_2(L)\geq \pbra{\max_{v\in V}\frac{\pi(v)}{\widetilde{\pi}(v)}}B(\bm{\Gamma})\lambda_2(\wt{L}).
    \end{align*}
    where the \emph{multigraph comparison constant} is defined as the maximum congestion over all edges $\tilde{e} \in \wt{E}$ connecting vertices $a,b\in V$,
    \begin{align*}
        B(\bm{\Gamma}) \coloneq \max_{\wt{e}\in \wt{E}} \cbra{\frac{\wt{d}}{d}\sum_{e\in E}\Ex_{\bm{\Gamma}}\sbra{\mathbf{1}_{\wt{e} \in \bm{\Gamma}(e)}\cdot |\bm{\Gamma}(e)|}}.\footnotemark
    \end{align*}
    Here $\pi$ and $\wt\pi$ are the (unique) stationary distributions for $P$ and $\wt{P}$, respectively.
    \footnotetext{If $\calP$ is a sequence then $|\calP|$ is its length.}
\end{corollary}
\begin{proof}
    We construct a (randomized) map $\bm\Delta$ as in \Cref{lem:comparison schreier} from the map $\bm\Gamma$ as follows. For each $(x,y)\in V^2$, if $x$ and $y$ are not connected by an edge in $E$ then set $\bm\Delta(x,y)=()$ (the sequence of length 0). Otherwise select a random edge $\bm{e}$ from $(E)_{x,y}$ and let $\pbra{\bm{\wt{e}}_1,\dots,\bm{\wt{e}}_{{T_e}}}$ be the random path $\bm\Gamma(\bm{e})$. For each $i\in [T_e]$ let $(\bm{u}_i,\bm{v}_i)$ be the vertices connected by $\wt{\bm{e}}_i$. Then set 
    \begin{align*}
        \bm\Delta(x,y)=\pbra{(\bm{u}_1,\bm{v}_1),\dots, (\bm{u}_{\ell},\bm{v}_{T_e})}=\pbra{(x,\bm{v}_1),\dots, (\bm{u}_{T_e},y)}.
    \end{align*}
    The comparison constant of $\bm{\Delta}$ is
    $$A(\bm\Delta)=\max_{(a,b)\in V^2} \cbra{\frac{1}{\wt{\pi}(x)\wt{P}(a,b)}\sum_{(x,y)\in V^2}\Ex_{\bm{\Delta}}\sbra{\mathbf{1}_{(a, b) \in \bm{\Delta}(x,y)}\cdot |\bm{\Delta}(x,y)|}\cdot  \pi_{\mathsf{ref}}(x)\cdot P(x, y)}.$$
    Since both Markov chains have the same stationary distribution over the same state space, the $\pi$ terms cancel, so the above is equal to
    \begin{align*}
        &\max_{(a,b)\in V^2} \cbra{\frac{1}{\wt{P}(a,b)}\sum_{(x,y)\in V^2}\Ex_{\bm{\Delta}}\sbra{\mathbf{1}_{(a, b) \in \bm{\Delta}(x,y)}\cdot |\bm{\Delta}(x,y)|}\cdot  P(x, y)}\\
        =&\max_{(a,b)\in V^2} \cbra{\frac{\wt{d}}{\abs{\wt{E}_{a,b}}}\sum_{(x,y)\in V^2}\Ex_{\bm{\Delta}}\sbra{\mathbf{1}_{(a, b) \in \bm{\Delta}(x,y)}\cdot |\bm{\Delta}(x,y)|}\cdot  \frac{\abs{E_{x,y}}}{d}}.
    \end{align*}
    Let us now start translating from pairs of vertices to edges of the multigraph. Each pair of vertices $(x, y)$ has $\left|E_{x,y}\right|$ edges connecting them, we can change the summation from pairs of vertices to edges $e \in E$. Recall that $u(e), v(e)$ are the endpoints of edge $e$. Continuing the calculation, the above is equal to 
    \begin{align*}
        &\max_{(a,b)\in V^2} \cbra{\frac{\wt{d}}{d\abs{\wt{E}_{a,b}}}\sum_{e\in E}\Ex_{\bm{\Delta}}\sbra{\mathbf{1}_{(a, b) \in \bm{\Delta}(u(e),v(e))}\cdot |\bm{\Delta}(u(e),v(e))|} }\\
        =&\max_{(a,b)\in V^2} \cbra{\frac{\wt{d}}{d}\sum_{e\in E} \frac{1}{\abs{\wt{E}_{a,b}}}\Ex_{\bm{\Delta}}\sbra{\mathbf{1}_{(a, b) \in \bm{\Delta}(u(e),v(e))}\cdot |\bm{\Delta}(u(e),v(e))|}}\\
        =&\max_{(a,b)\in V^2} \cbra{\frac{\wt{d}}{d}\sum_{e\in E} \frac{1}{\abs{\wt{E}_{a,b}}}|\bm{\Gamma}(e)|\Pr_{\bm{\Delta}}\sbra{\mathbf{1}_{(a, b) \in \bm{\Delta}(u(e),v(e))}}}.
    \end{align*}
    The last equality is because $|\bm\Delta(u(e),v(e))|=|\bm\Gamma(e)|$ with certainty, and $|\bm\Gamma(e)|$ is a deterministic quantity that only depends on $e$.

    The sum of probabilities that $\wt{e}\in (\wt{E})_{a,b}$ appears in the sequence $\bm\Gamma(e)$, over all such $\wt{e}$, is equal to the probability that $(a,b)$ appears in $\bm\Delta(u(e),v(e))$. By averaging, we have that the probability that $\wt{e}\in \wt{E}_{a,b}$ appears in $\bm\Gamma(e)$, where $\wt{e}$ maximizes this quantity, is at least $\frac1{|\wt{E}_{a,b}|}$ times the appearance probability of $(a,b)$. This results in
    \begin{align*}
        A(\bm\Delta)\leq&\max_{\wt{e}\in \wt{E}} \cbra{\frac{\wt{d}}{d}\sum_{e\in E}\Ex_{\bm{\Gamma}}\sbra{\mathbf{1}_{\wt{e} \in \bm{\Gamma}(e)}\cdot |\bm{\Gamma}(e)|} }=B(\bm\Gamma).
    \end{align*}
    Applying \Cref{thm:wilmer comparison} with this new map $\bm\Delta$ completes the proof.
\end{proof}

\section{Extension to $D$-dimensional Lattices}
\label{sec:generallattices}

\subsection{More Bit Arrays and Color Classes}

For $1\leq D' \leq D$, we regard an element $x \in \{\pm1\}^{n^{D'/D}}$ as a function $x : \sbra{n^{1/D}\,}^{\otimes D'} \to \{\pm1\}$. 
Similarly, we regard an element $X \in \{\pm1\}^{n^{D'/D}k}$ as a function $X : \sbra{n^{1/D}\,}^{\otimes D'} \times \sbra{k} \to \{\pm1\}$. 
For $X \in \{\pm1\}^{n^{D'/D}k}$, $i \in \sbra{n^{1/D}\,}$, $\tau \in \sbra{n^{1/D}\,}^{\otimes D'-1}$ and $\ell \in [k]$, we use the notation:
\begin{itemize}
    \item $X^\ell_{i, \tau} = X(i, \tau, \ell) \in \{\pm1\}$
    \item $X^\ell = X \mid_{\sbra{n^{1/D}\,}^{\otimes D'} \times \{\ell\}} \in \{\pm1\}^{n^{D'/D}}$
    \item $X_{i, \cdot} = X \mid_{\{i\} \times \sbra{n^{1/D}\,}^{\otimes D'-1} \times [k]} \in \{\pm1\}^{n^{(D'-1)/D}k}$
    \item $X_{\cdot, \tau} = X \mid_{\sbra{n^{1/D}\,} \times \{\tau\} \times [k]} \in \{\pm1\}^{n^{1/D}k}$
\end{itemize}

Our definition for coloring will remain the same, namely for $X \in \{\pm1\}^{n^{D'/D}k}$ we will say $X^\ell_{i, \cdot}$ and $X^m_{i, \cdot}$ are colored the same if they are equal, but it is worth noting that these objects are $(D'-1)$-dimensional sublattices and the underlying relations are then $n^{1/D}$ tuples of equivalence relations. Note that in the case $D = 2$ these do in fact correspond to rows. Since the number of such sublattices is $n^{1/D}$ in general, the number of color classes is at most $k^{kn^{1/D}}$.

Our partition into $B_{\mathrm{safe}}$, $B_{\mathrm{coll}}$, and $B_{=0}$ remains mostly the same but based on the generalized notion of color class defined above:
\begin{align*}
    &B_\text{safe} := \cbra{X \in \sD : \forall \ell \neq m \in [k], i \in [n^{1/D}], X^\ell_{i, \cdot} \neq X^m_{i, \cdot}},\\
    &B_\text{coll} := \sD_{n^{D'/D}} \setminus B_\text{safe},\\
    &B_{=0} := \{\pm1\}^{n^{D'/D}k} \setminus \sD_{n^{D'/D}}.
\end{align*}
Throughout this section the value of $D'$ will be clear from context.

\begin{fact}
\label{fact:gencolorclasssizes}
    $\frac{\abs{B_{\mathrm{coll}}}}{\abs{\sD_{n^{D'/D}}}} \leq \frac{2n^{1/D}k^2}{2^{n^{(D'-1)/D}}}$.
\end{fact}

\begin{proof}
We write:
\begin{equation*}
    \frac{\abs{B_{\mathrm{coll}}}}{\abs{\sD_{n^{D'/D}}}} = \frac{\abs{B_{\mathrm{coll}}}}{\abs{\{\pm1\}^{n^{D'/D}k}}} \cdot \frac{\abs{\{\pm1\}^{n^{D'/D}k}}}{\abs{\sD_{n^{D'/D}}}}.
\end{equation*}
The process of sampling from $\{\pm1\}^{n^{D'/D}k}$ can now be seen as sampling $n^{1/D}k$ sublattices from $\{\pm1\}^{n^{(D'-1)/D}}$. Under this view, a simple union bound tells us that there are at most $n^{1/D}k^2$ possible collisions, allowing us to bound the probability by $\frac{n^{1/D}k^2}{2^{n^{(D'-1)/D}}}$. Again, this bounds the size of $\abs{B_{=0}}$ as well, allowing us to crudely claim $\frac{\abs{\sD_{n^{D'/D}}}}{\abs{\{\pm1\}^{n^{D'/D}k}}} \geq \frac{1}{2}$ using our assumption on $k$.
\end{proof}

\subsection{Inductively Defined Random Permutations}\label{sec:defs higher dim}
Fix $n$, $k$, and $D\geq 2$. Let $\mathcal{P}_1$ be a random permutation of $\{\pm1\}^{n^{1/D}}$. We will inductively define for all $2 \leq D'\leq D$ a random permutation $\mcP_{D'}$ on $\{\pm1\}^{n^{D'/D}}$.
\begin{itemize}
    \item Let $\mcP_{D'-1}$ be a distribution on $\mfS{\{\pm1\}^{n^{(D'-1)/D}}}$. 
    \item Let $\mcP_C$ be a distribution on $\mfS{\{\pm1\}^{n^{D'/D}}}$ such that $\pi \sim \mcP_C$ is sampled as follows: 
          Sample $\sigma_\tau \sim \mcP_1$ independently for each $\tau \in \sbra{n^{1/D}\,}^{\otimes D'-1}$ and define $\pi$ such that $\pi(x)_{\cdot, \tau} = \sigma_\tau(x_{\cdot, \tau})$ for all $x \in \{\pm1\}^{n^{D'/D}}$ and all $\tau \in \sbra{n^{1/D}}^{\otimes D'-1}$. 
    \item Let $\mcP_{L,D'-1}$ be a distribution on $\mfS{\{\pm1\}^{n^{D'/D}}}$ such that $\pi \sim \mcP_{L, D'-1}$ is sampled as follows: 
          Sample $\sigma_i \sim \mcP_{D'-1}$ independently for each $i \in \sbra{n^{1/D}}$ and define $\pi$ such that $\pi(x)_{i, \cdot} = \sigma_i(x_{i, \cdot})$ for all $x \in \{\pm1\}^{n^{D'/D}}$ and all $i \in \sbra{n^{1/D'}}$. 
    \item Let $\mcP^0_{D'} = \mcP_{L,D'-1}$. For all $s \ge 1$, let $\mcP_{D'}^{s}$ be the distribution on $\mfS{\{\pm1\}^{n^{D'/D}}}$ such that $\pi \sim \mcP^{s}_{D'}$ is sampled as follows:
          Sample $\sigma_1 \sim \mcP_{D'}^{s-1}$, $\sigma_2 \sim \mcP_{C}$, and $\sigma_3 \sim \mcP_{L,D'-1}$ and define $\pi=\sigma_3 \circ \sigma_2 \circ \sigma_1$.
    \item Set $\mcP_{D'} = \mcP_{D'}^t$, where $t$ is the constant from \Cref{lem:genreduction} below if $D'\geq 3$. Otherwise if $D'=2$ then set $t=\Theta(k\log k)$, where the constant is chosen from the statement of \Cref{thm:2D to 1D reduction technical}.
\end{itemize}
For ease of analyzing the above random permutations, we define the idealized versions of the above distributions based on the following pieces.
\begin{itemize}
    \item Let $\mcG_C$ be a distribution on $\mfS{\{\pm1\}^{n^{D'/D}}}$ such that $\pi \sim \mcG_C$ is sampled as follows: 
          Sample $\sigma_\tau \sim \mcU\pbra{\mfS{\{\pm1\}^{n^{1/D}}}}$ independently for each $\tau \in \sbra{n^{1/D}\,}^{\otimes D'-1}$ and define $\pi$ such that $\pi(x)_{\cdot, \tau} = \sigma_\tau(x_{\cdot, \tau})$ for all $x \in \{\pm1\}^{n}$ and all $\tau \in \sbra{n^{1/D}\,}^{\otimes D'-1}$. 
    \item Let $\mcG_{L, D'-1}$ be a distribution on $\mfS{\{\pm1\}^{n^{D'/D}}}$ such that $\pi \sim \mcG_{L, D'-1}$ is sampled as follows: 
          Sample $\sigma_i \sim \mcU\pbra{\mfS{\{\pm1\}^{n^{(D'-1)/D}}}}$ independently for each $i \in \sbra{n^{1/D}\,}$ and define $\pi$ such that $\pi(x)_{i, \cdot} = \sigma_i(x_{i, \cdot})$ for all $x \in \{\pm1\}^{n}$ and all $i \in \sbra{n^{1/D}\,}$. 
\end{itemize}

\subsection{Generalization of Main Theorem}

Our proof will largely follow the blueprint of the $D = 2$ case, our main result. For $X\in\sD$,
\begin{equation*}
    d_{\textrm{TV}}\pbra{\mcP^t_{D, X}, \mcG_X} \leq d_{\textrm{TV}}\pbra{\mcP^t_{D,X}, \mcG^t_{D,X}} + d_{\textrm{TV}}\pbra{\mcG^t_{D, X}, \mcG_{D, X}}.
\end{equation*}

We prove analogues of \Cref{lem:reduction} and \Cref{lem:maintrick}. 
\begin{lemma}
    \label{lem:genreduction}
    Assume the hypotheses of \Cref{thm:genresult}. Fix any $D'\geq 3$. Suppose that $\mathcal{P}_{D'-1}$ is a $\frac{1}{(4(t+1)n)^{D-D'+1}} \cdot \frac1{2^{n^{1/D}}}$-approximate $k$-wise independent permutation of $\{\pm1\}^{n^{(D'-1)/D}}$ and $\mathcal{P}_1$ is a $\frac1{(4(t+1)n)^{D}}\cdot\frac1{2^{n^{1/D}}}$-approximate $k$-wise independent permutation of $\{\pm1\}^{n^{1/D}}$. Then with the above definitions, for any $X \in \{\pm1\}^{n^{D'/D}k}$,
    \begin{equation*}
        \sum_{Y \in \{\pm1\}^{n^{D'/D}k}} \abs{\ip{e_X}{(T_{\mcP_{D'}^t}-T_{\mcG^t_{D'}}) e_Y}} \leq  \frac12\cdot \frac1{(4(t+1)n)^{D-D'}}\cdot\frac1{2^{n^{1/D}}}.
    \end{equation*}
\end{lemma}

\begin{lemma}
    \label{lem:genmaintrick}
    Assume that $k\log k\leq n^{1/3}$, that $n$ is large enough, and fix $D$. Then for all $t \geq 2500$, any $3\leq D'\leq D$, and any $X \in \{\pm1\}^{n^{D'/D}k}$,
    \begin{equation*}
        \sum_{Y \in \{\pm1\}^{n^{D'/D}k}} \abs{\ip{e_X}{(T_{\mcG^t_{D'}}-T_{\mcG_{D'}}) e_Y}} \leq \frac{1}{(4(t+1)n)^{D-D'+1}}\cdot\frac1{2^{n^{1/D}}}.
    \end{equation*}
\end{lemma}

We apply these two lemmas along with \Cref{tvdistancetolinearform} to obtain the generalization of our main result to higher-dimensional lattices. 

\begin{proof}[Proof of \Cref{thm:genresult}]
    Fix $D$ and set $t\geq 2500$ as in \Cref{lem:genmaintrick}. Let $\mathcal{P}_1$ be a $\frac1{(4(t+1)n)^D}\cdot \frac1{2^{n^{1/D}}}$-approximate $k$-wise independent permutation of $\{\pm1\}^{n^{1/D}}$. Let $\mathcal{P}_{D'}$ be constructed from $\mathcal{P}_1$ as in \Cref{sec:defs higher dim} for all $2\leq D'\leq D$.

    We prove by induction on $D'$ that for all $D'\leq D$, the random permutation $\mathcal{P}_{D'}$ is a $\frac1{(4(t+1)n)^{D-D'}}\cdot\frac1{2^{n^{1/D}}}$-approximate $k$-wise independent permutation of $\{\pm1\}^{n^{D'/D}}$. In the base case $D'=1$, this follows by assumption on $\mcP_1$. In the other base case $D'=2$, this follows from \Cref{thm:2D to 1D reduction technical}. 

    Now fix $3\leq D'\leq D$. Because $k\log k\leq n^{1/3}$ so that the hypothesis of \Cref{lem:genmaintrick} is satisfied. Assume that $\mathcal{P}_{D'-1}$ is a $\frac1{(4(t+1)n)^{D-D'+1}}\cdot\frac1{2^{n^{1/D}}}$-approximate $k$-wise independent permutation of $\{\pm1\}^{n^{(D'-1)/D}}$. By \Cref{lem:genreduction} and \Cref{lem:genmaintrick}, we have that $\mathcal{P}_{D'}^t$ is a $\frac1{(4(t+1)n)^{D-D'}}\cdot\frac1{2^{n^{1/D}}}$-approximate $k$-wise independent permutation of $\{\pm1\}^{n^{D'/D}}$:
    \begin{align*}
        d_{\textrm{TV}}\pbra{\mcP^t_{D', X}, \mcG_X} &\leq d_{\textrm{TV}}\pbra{\mcP^t_{D',X}, \mcG^t_{D',X}} + d_{\textrm{TV}}\pbra{\mcG^t_{D', X}, \mcG_{D',X}}\\
        &\leq \frac12\cdot\frac1{(4(t+1)n)^{D-D'}}\cdot\frac1{2^{n^{1/D}}}+\frac{1}{(4(t+1)n)^{D-D'+1}}\cdot\frac1{2^{n^{1/D}}}\\
        &\leq \frac12\cdot\frac1{(4(t+1)n)^{D-D'}}\cdot\frac1{2^{n^{1/D}}}+\frac12\cdot\frac{1}{(4(t+1)n)^{D-D'}}\cdot\frac1{2^{n^{1/D}}}\\
        &\leq\frac1{(4(t+1)n)^{D-D'}}\cdot\frac1{2^{n^{1/D}}}.
    \end{align*}
    This completes the induction on $D'$. As a result of the induction, we find that $\mathcal{P}_D$ is a $\frac1{2^{n^{1/D}}}$-approximate $k$-wise independent permutation of $\{\pm1\}^{n^{D/D}}=\{\pm1\}^{n}$.

    To instantiate our construction, we take $\mathcal{P}_1$ to be the depth $\widetilde{O}(k)\cdot (n^{1/D}k + n^{1/D}D\log n)=\widetilde{O}(n^{1/D}Dk^2)$ random one-dimensional brickwork circuit from \Cref{thm:1D main}. By \Cref{thm:2D main}, the random permutation $\mathcal{P}_2$ is implemented by a random two-dimensional brickwork circuit of depth $\widetilde{O}(n^{1/D}Dk^3)$. By the construction, if $\mathcal{P}_{D'-1}$ can be implemented by a random $D'-1$-dimensional brickwork circuit of depth $\leq d$ and $\mcP_1$ can be implemented by a random one-dimensional brickwork of depth $\leq d$ then $\mathcal{P}_{D'}$ can be implemented by a random $D'$-dimensional brickwork circuit of depth $d\cdot(2t+1)$. This implies that $\mathcal{P}_D$ can be implemented by a $D$-dimensional brickwork circuit of depth $(2t+1)^{D-2}\cdot \widetilde{O}(n^{1/D}Dk^3)=\mathrm{exp}(D)\cdot \widetilde{O}(n^{1/D}k^3)$. 
\end{proof}

\subsubsection{Proof of \Cref{lem:genreduction}}
Following the proof of \Cref{lem:reduction} in the $D = 2$ case, we use \Cref{fact:differenceofproducts} to bound:
\begin{align*}
     &\sum_{Y \in \{\pm1\}^{n^{D'/D}k}} \abs{\ip{e_X}{(T_{\mcP_{D'}^t}-T_{\mcG^t_{D'}}) e_Y}}\\
     &\leq (t+1) \cdot \sum_{Y \in \{\pm1\}^{n^{D'/D}k}} \abs{\ip{e_X}{(T_{\mcP_{L, D'-1}}-T_{\mcG_{L, D'-1}}) e_Y}} + t \cdot \sum_{Y \in \{\pm1\}^{nk}} \abs{\ip{e_X}{(T_{\mcP_C}-T_{\mcG_C}) e_Y}}.
\end{align*}
To bound each of the two terms, we will establish the following two lemmas.
\begin{lemma}\label{lem:genlatticereduced}
    Assume the hypothesis of \Cref{lem:genreduction}. Then,
    \begin{align*}
        \sum_{Y \in \sD} \abs{\ip{e_X}{(T_{\mcP_{L, D'-1}}-T_{\mcG_{L, D'-1}}) e_Y}} \leq n^{1/D}\cdot \frac{1}{(4(t+1)n)^{D-D'+1}} \cdot \frac1{2^{n^{1/D}}}.
    \end{align*}
\end{lemma}

\begin{lemma}
    \label{lem:genrowreduced}
    Assume the hypothesis of \Cref{lem:genreduction}. Then,
    \begin{align*}
        \sum_{Y \in \sD} \abs{\ip{e_X}{(T_{\mcP_C}-T_{\mcG_C}) e_Y}} \leq n^{(D'-1)/D}\cdot \frac1{(4(t+1)n)^D}\cdot \frac1{2^{n^{1/D}}}.
    \end{align*}
\end{lemma}
Plugging directly into the equation above finishes the proof of \Cref{lem:genreduction}.
\begin{align*}
     &\sum_{Y \in \{\pm1\}^{n^{D'/D}k}} \abs{\ip{e_X}{(T_{\mcP_{D'}^t}-T_{\mcG^t_{D'}}) e_Y}}\\
     &\leq(t+1)\cdot n^{1/D}\cdot \frac{1}{(4(t+1)n)^{D-D'+1}}\cdot \frac1{2^{n^{1/D}}}+ t\cdot n^{(D'-1)/D}\cdot \frac1{(4(t+1)n)^D}\cdot \frac1{2^{n^{1/D}}}\\
     &\leq \frac14\cdot\frac1{(4(t+1)n)^{D-D'}}\cdot\frac1{2^{n^{1/D}}}+\frac14\cdot\frac1{(4(t+1)n)^{D-1}}\cdot\frac1{2^{n^{1/D}}}\\
     &\leq \frac12\cdot\frac1{(4(t+1)n)^{D-D'}}\cdot\frac1{2^{n^{1/D}}}.
\end{align*}
Note that we used the definitions of $\mcP_{D'}^t$ and $\mcG_{D'}^t$ from \Cref{sec:defs higher dim}. This concludes the proof of \Cref{lem:genreduction}.

\begin{proof}[Proof of \Cref{lem:genlatticereduced}]
Recall that $X, Y \in \{\pm1\}^{n^{D'/D}k}$ and we write $X_{i, \cdot}$ for $i \in \sbra{n^{1/D}}$ to denote one of $n^{1/D}$ $(D'-1)$-dimensional slices. The operator $T_{\mcP_{L, D'-1}}$ can be seen as a $n^{1/D}$-wise tensorization of $T_{\mcP_{D-1}}$ acting individually on each slice. As such, we compute:
\begin{align*}
    &\sum_{Y \in \{\pm1\}^{n^{D'/D}k}} \abs{\ip{e_X}{(T_{\mcP_{L, D'-1}}-T_{\mcG_{L, D'-1}}) e_Y}}\\
    &=\sum_Y \abs{ \prod_{i=1}^{n^{1/D}}\Pr[X_{i, \cdot} \to_{T_{\mcP_{D'-1}}} Y_{i ,\cdot}] - \prod_{i=1}^{n^{1/D}}\Pr[X_{i, \cdot} \to_{T_{\mcG_{n^{D'/D}}}} Y_{i ,\cdot}]}\\
    &=\sum_Y \abs{\sum_{j = 1}^{n^{1/D}} \prod_{i=1}^{j-1}\Pr[X_{i, \cdot} \to_{T_{\mcP_{D'-1}}} Y_{i ,\cdot}] \pbra{\Pr[X_{j, \cdot} \to_{T_{\mcP_{D'-1}}} Y_{j ,\cdot}] - \Pr[X_{j, \cdot} \to_{T_{\mcG_{n^{D'/D}}}} Y_{j ,\cdot}]} \prod_{i=j+1}^{n^{1/D}}\Pr[X_{i, \cdot} \to_{T_{\mcG_{n^{D'/D}}}} Y_{i ,\cdot}]}\\
    &\leq  \sum_{j = 1}^{n^{1/D}} \sum_Y \abs{\prod_{i=1}^{j-1}\Pr[X_{i, \cdot} \to_{T_{\mcP_{D'-1}}} Y_{i ,\cdot}] \pbra{\Pr[X_{j, \cdot} \to_{T_{\mcP_{D'-1}}} Y_{j ,\cdot}] - \Pr[X_{j, \cdot} \to_{T_{\mcG_{n^{D'/D}}}} Y_{j ,\cdot}]} \prod_{i=j+1}^{n^{1/D}}\Pr[X_{i, \cdot} \to_{T_{\mcG_{n^{D'/D}}}} Y_{i ,\cdot}]}\\
    &= \sum_{j = 1}^{n^{1/D}} \sum_y \abs{\Pr[X_{j, \cdot} \to_{T_{\mcP_{D'-1}}} y] - \Pr[X_{j, \cdot} \to_{T_{\mcG_{n^{D'/D}}}} y]} \sum_{\substack{Y\\ Y_{j, \cdot} = y}} \prod_{i=1}^{j-1}\Pr[X_{i, \cdot} \to_{T_{\mcP_{D'-1}}} Y_{i ,\cdot}] \prod_{i=j+1}^{n^{1/D}}\Pr[X_{i, \cdot} \to_{T_{\mcG_{n^{D'/D}}}} Y_{i ,\cdot}]\\
    &\leq \sum_{j = 1}^{n^{1/D}} \sum_{y \in \{\pm1\}^{n^{(D'-1)/D}k}} \abs{\Pr[X_{j, \cdot} \to_{T_{\mcP_{D'-1}}} y] - \Pr[X_{j, \cdot} \to_{T_{\mcG_{n^{D'/D}}}} y]}\\
    &\leq n^{1/D}\cdot \frac{1}{(4(t+1)n)^{D-D'+1}} \cdot \frac1{2^{n^{1/D}}}.
\end{align*}
The last line follows from \Cref{genkwiseimplies}.
\end{proof}

\begin{lemma}
    \label{genkwiseimplies}
    For every $x \in \{\pm1\}^{n^{(D'-1)/D}k}$ we have:
    \begin{equation*}
        \sum_{y \in \{\pm1\}^{n^{(D'-1)/D}k}} \abs{\Pr[x \to_{T_{\mcP_{D'-1}}} y] - \Pr[x \to_{T_{\mcG_{n^{(D'-1)/D}}}} y]} \leq \frac{1}{(4(t+1)n)^{D-D'+1}} \cdot \frac1{2^{n^{1/D}}}.
    \end{equation*}
\end{lemma}

\begin{proof}[Proof of \Cref{genkwiseimplies}]
We view $x$ as a $k$-tuple of $(D'-1)$-dimensional grids. We denote by $B$ the ``tuple-wise'' color class of $x$ (if two grids are equal they are colored the same). We create a projection function $\varphi_B$ defined analogously to that in \Cref{kwiseimplies}, taking $x$ to a corresponding $\tau$-tuple with distinct elements.
    \begin{align*}
    &\sum_{y \in \{\pm1\}^{n^{(D'-1)/D}k}} \abs{\Pr[x \to_{T_{\mcP_{D'-1}}} y] - \Pr[x \to_{T_{\mcG_{n^{(D'-1)/D}}}} y]}\\
    =&\sum_{y \in B(x)} \abs{\Pr[x \to_{T_{\mcP_{D'-1}}} y] - \Pr[x \to_{T_{\mcG_{n^{(D'-1)/D}}}} y]}\\
    =& \sum_{y \in B(x)} \abs{\Pr[\varphi_B(x) \to_{T_{\mcP_{D'-1}}} \varphi_B(y)] - \Pr[\varphi_B(x) \to_{T_{\mcG_{n^{(D'-1)/D}}}} \varphi_B(y)]}\\
    =& \sum_{\varphi_B(y) \in \sD_{n^{(D'-1)/D}}^{(\tau)}} \abs{\Pr[\varphi_B(x) \to_{T_{\mcP_{D'-1}}} \varphi_B(y)] - \Pr[\varphi_B(x) \to_{T_{\mcG_{n^{(D'-1)/D}}}} \varphi_B(y)]}\\
    =& \sum_{\varphi_B(y) \in \sD_{n^{(D'-1)/D}}^{(\tau)}} \abs{\sum_{z \in \sD_{n^{(D'-1)/D}}^{(k-\tau)}}\Pr[(\varphi_B(x), \cdot) \to_{T_{\mcP_{D'-1}}} (\varphi_B(y), z)] - \Pr[(\varphi_B(x), \cdot) \to_{T_{\mcG_{n^{(D'-1)/D}}}} (\varphi_B(y), z)]}\\
    \leq& \sum_{\substack{\varphi_B(y) \in \sD_{n^{(D'-1)/D}}^{(\tau)} \\ z \in \sD_{n^{(D'-1)/D}}^{(k-\tau)}}} \abs{\Pr[(\varphi_B(x), \cdot) \to_{T_{\mcP_{D'-1}}} (\varphi_B(y), z)] - \Pr[(\varphi_B(x), \cdot) \to_{T_{\mcG_{n^{(D'-1)/D}}}} (\varphi_B(y), z)]}\\
    =& \sum_{y \in \sD^{(k)}_{n^{(D'-1)/D}}} \abs{\Pr[(\varphi_B(x), \cdot) \to_{T_{\mcP_{D'-1}}} (\varphi_B(y), y_{[k] \setminus T})] - \Pr[(\varphi_B(x), \cdot) \to_{T_{\mcG_{n^{(D'-1)/D}}}} (\varphi_B(y), y_{[k] \setminus T})]}.
\end{align*}
The last step assumes $(\varphi_B(x), \cdot) \in \sD_{n^{(D'-1)/D}}$, that is, it is a distinct $k$-tuple. We then appeal to the fact that ${\mcP_{D'-1}}$ is assumed to be $\frac{1}{(4(t+1)n)^{D-D'+1}} \cdot \frac1{2^{n^{1/D}}}$-approximate $k$-wise independent to finish.
\end{proof}

The proof of \Cref{lem:genrowreduced} is nearly identical to that of \Cref{lem:genlatticereduced}, but partitioning $\{\pm1\}^{n^{D'/D}}$ over one-dimensional columns yields a tensor product of order $n^{(D'-1)/D}$, which becomes a factor in the result, and additionally we appeal to the error in $\mcP_1$ at the end.

\subsubsection{Proof of \Cref{lem:genmaintrick}}
\label{sec:genmaintrick}

This proof follows near identically to \Cref{subsec:inductiontrick}. Throughout this section we assume the hypothesis of \Cref{lem:genmaintrick}, namely that $k\log k\leq n^{1/3}$. It suffices to prove for any $X \in \sD_{n^{D'/D}}$ via \Cref{lem:genoffdiagonalmoment}:
\begin{equation*}
    \sum_{Y \in \sD_{n^{D'/D}}} \abs{\ip{e_X}{(T_{\mcG^t_{D'}}-T_{\mcG_{D'}}) e_Y}} \leq \frac{1}{(4(t+1)n)^{D-D'+1}}\cdot\frac1{2^{n^{1/D}}}.
\end{equation*}

For clarity, we will assume all operators and distributions from this point on are implicitly parameterized by $D'$ and drop the subscript.

\begin{lemma}
\label[lemma]{lem:genoffdiagonalmoment}
    Assume the hypotheses of \Cref{lem:genmaintrick}. Then $\abs{\ip{e_X}{(T_{\mcG^t}-T_{\mcG}) e_Y}} \leq \frac{t+1}{2^{n^{(D'-1)/D}(t-1)/128}} \cdot \frac{1}{\abs{B(Y)}}$.
\end{lemma}
The lemma is used in the following calculation:
\begin{equation*}
    \sum_{Y \in \sD_{n^{D'/D}}} \abs{\ip{e_X}{(T_{\mcG^t}-T_{\mcG}) e_Y}} \leq \frac{t+1}{2^{n^{(D'-1)/D}(t-1)/128}} \sum_{Y \in \sD_{n^{(D'-1)/D}}} \frac{1}{\abs{B(Y)}} \leq \frac{k^{kn^{1/D}} \cdot (t+1)}{2^{n^{(D'-1)/D}(t-1)/128}}.
\end{equation*}
We use that the number of color classes is less than $k^{kn^{1/D}}$. Since $k \log k \leq n^{1/3} \leq n^{(D'-2)/D}$ for $D, D' \geq 3$, we have that:
\begin{equation*}
    \sum_{Y \in \sD_{n^{D'/D}}} \abs{\ip{e_X}{(T_{\mcG^t}-T_{\mcG}) e_Y}} \leq \frac{2^{n^{(D'-1)/D}} \cdot (t+1)}{2^{n^{(D'-1)/D}(t-1)/128}} \leq \frac{1}{2^{(n^{(D'-1)/D}/128-1)(t-1)-1}}.
\end{equation*}

If we set $t = \frac{n^{1/D}D\log_2 (4(t+1)n)+1}{n^{(D'-1)/D}/128-1}+1$ we achieve the desired bound. Note that for large enough $n$ we have $t \leq \frac{2500 D \log_2 n}{n^{1/D}}$. Further if $D \leq \frac{1}{2} \cdot\frac{\log_2 n}{\log_2 \log_2 n}$ then we have $n^{1/D} \geq (\log_2 n)^2$, which is enough to conclude $t \leq 2500$. This concludes the proof of \Cref{lem:genmaintrick}.

\begin{proof}[Proof of \Cref{lem:genoffdiagonalmoment}.]
Recall $T_{\mcG^t} = T_{\mcG_L}(T_{\mcG_C}T_{\mcG_L})^t$ so $T_{\mcG^t} - T_\mcG = (T_{\mcG_L}T_{\mcG_C})^t(T_{\mcG_L} - T_\mcG)$. We induct on $t$. Consider first when $t = 0$.
\begin{equation*}
    \abs{\ip{e_X}{(T_{\mcG_L}-T_{\mcG}) e_Y}} = \abs{\Pr[X \to_{T_{\mcG_L}} Y] - \Pr[X \to_{T_{\mcG}} Y]} = \abs{\Pr[Y \to_{T_{\mcG_L}} X] - \Pr[Y \to_{T_{\mcG}} X]}.
\end{equation*}
Note that we guarantee inductively that $\mcG_L$ is self-adjoint. This quantity is bounded by $\frac{1}{\abs{B(Y)}}$ as before.

For the induction step, assume the lemma for $t \geq 0$ and compute
\begin{align*}
    \abs{\ip{e_X}{(T_{\mcG_L}T_{\mcG_C})^{t+1}(T_{\mcG_L}-T_{\mcG}) e_Y}} 
    =& \abs{\ip{e_X}{(T_{\mcG_L}T_{\mcG_C})(T_{\mcG_L}T_{\mcG_C})^t(T_{\mcG_L}-T_{\mcG}) e_Y}}\\
    =& \abs{T_{\mcG_L}\pbra{T_{\mcG_C}(T_{\mcG_L}T_{\mcG_C})^t(T_{\mcG_L}-T_{\mcG}) e_Y}(X)}\\
    =& \abs{\sum_{Z \in \sD} \Pr[X \to_{T_{\mcG_C}T_{\mcG_L}} Z] \ip{e_Z}{(T_{\mcG_L}T_{\mcG_C})^t(T_{\mcG_L}-T_{\mcG}) e_Y}}\\
    \leq &\abs{\sum_{Z \in B_\text{safe}} \Pr[X \to_{T_{\mcG_L}T_{\mcG_C}} Z] \ip{e_Z}{(T_{\mcG_L}T_{\mcG_C})^t(T_{\mcG_L}-T_{\mcG}) e_Y}}\\
    &\;\;\;\;+ \abs{\sum_{Z \in B_\text{coll}} \Pr[X \to_{T_{\mcG_C}T_{\mcG_C}} Z] \ip{e_Z}{(T_{\mcG_L}T_{\mcG_C})^t(T_{\mcG_L}-T_{\mcG}) e_Y}}\\
    \leq& \max_{Z \in B_{\mathrm{safe}}} \abs{\ip{e_Z}{(T_{\mcG_L}T_{\mcG_C})^t(T_{\mcG_L}-T_{\mcG}) e_Y}}\\
    &\;\;\;\;+ \frac{1}{2^{n^{(D'-1)/D}/128}} \cdot \max_{Z \in B_{\mathrm{coll}}} \abs{\ip{e_Z}{(T_{\mcG_L}T_{\mcG_C})^t(T_{\mcG_L}-T_{\mcG}) e_Y}}\\
    \leq& \max_{Z \in B_{\mathrm{safe}}} \abs{\ip{e_Z}{(T_{\mcG_L}T_{\mcG_C})^t(T_{\mcG_L}-T_{\mcG}) e_Y}} + \frac{t+1}{2^{n^{(D'-1)/D}t/128}} \cdot \frac{1}{\abs{B(Y)}}.
\end{align*}
We apply the induction in the last line and \Cref{lem:genlowcollprob} as stated below in the previous line, in order to bound the probability $X$ lands in the collision region.

\begin{lemma}\label{lem:genspectralnorm}
    Assume the hypotheses of \Cref{lem:genmaintrick}. Then $\norm{T_{\mcG_L}T_{\mcG_C}T_{\mcG_L}-T_{\mcG}}_{2} \leq \frac1{2^{n^{(D'-1)/D}/128}}$.
\end{lemma}

\begin{lemma}
\label{lem:genlowcollprob}
    Assume the hypotheses of \Cref{lem:genmaintrick}. Then for all $X \in \sD$, $\Pr[X \to_{T_{\mcG_C}T_{\mcG_L}} B_\text{coll}] \leq \frac1{2^{n^{(D'-1)/D}/128}}$.
\end{lemma}
We prove these two lemmas in \Cref{sec:genspectral}. The use of \Cref{lem:genlowcollprob} is in bounding the latter term above. To use \Cref{lem:genspectralnorm} we write for $Z \in B_{\mathrm{safe}}$:
\begin{align*}
    \abs{\ip{e_Z}{(T_{\mcG_L}T_{\mcG_C})^t(T_{\mcG_L}-T_{\mcG}) e_Y}} =& \abs{\ip{T_{\mcG_L}e_Z}{(T_{\mcG_L}T_{\mcG_C}T_{\mcG_L}-T_{\mcG})^t T_{\mcG_L} e_Y}} \\
    \leq& \norm{T_{\mcG_L}T_{\mcG_C}T_{\mcG_L}-T_{\mcG}}_{2}^t \norm{T_{\mcG_L}e_Z}_{2} \norm{T_{\mcG_L}e_Y}_{2}\\
    \leq& \frac{1}{2^{n^{(D'-1)/D}t/128}} \cdot \frac{1}{\abs{B(Y)}^{1/2}\abs{B_\text{safe}}^{1/2}}\\
    \leq& \frac{1}{2^{n^{(D'-1)/D}t/128}} \cdot \frac{1}{\abs{B(Y)}}.
\end{align*}
The first step uses the self-adjointness of $T_{\mcG_C}$, the fact that $T_{\mcG_C}^2 = T_{\mcG_C}$, and \Cref{fact:TGabsorbs}. The inequality is an application of Cauchy-Schwarz and submultiplicativity of the operator norm. The second to last step uses \Cref{lem:genspectralnorm} and \Cref{TGL eU 2 norm} below, and the last step uses \Cref{fact:gencolorclasssizes}, namely that $B_{\mathrm{safe}}$ is larger than every other color class for our choice of $k$ and large enough $n$.

\begin{claim}\label{TGL eU 2 norm}
    For arbitrary $U \in \{\pm1\}^{nk}$:
\begin{equation*}
    \norm{T_{\mcG_L}e_U}_{2} = \frac{1}{\abs{B(U)}^{1/2}}.
\end{equation*}
\end{claim}

\begin{proof}
    Observe:
    \begin{equation*}
        \norm{T_{\mcG_L}e_U}_{2} = \sqrt{\sum_{W \in B(U)} \pbra{T_{\mcG_L}e_U(W)}^2} = \sqrt{\sum_{W \in B(U)} \Pr[W \to_{\mcG_L} U]^2} = \sqrt{\sum_{W \in B(U)} \pbra{\frac{1}{\abs{B(U)}}}^2} = \frac{1}{\abs{B(U)}^{1/2}}.\qedhere
    \end{equation*}
\end{proof}
Putting the two together then gives us:
\begin{equation*}
    \abs{\ip{e_X}{(T_{\mcG_L}T_{\mcG_C})^{t+1}(T_{\mcG_L}-T_{\mcG}) e_Y}} \leq \frac{1}{2^{n^{(D'-1)/D}t/128}} + \frac{t+1}{2^{n^{(D'-1)/D}t/128}} \leq \frac{t+2}{2^{n^{(D'-1)/D}t/128}}.
\end{equation*}
This concludes the proof of \Cref{lem:genoffdiagonalmoment}.
\end{proof}

\subsection{Proof of Spectral Properties}\label{sec:genspectral}

In this section we prove \Cref{lem:genspectralnorm} and \Cref{lem:genlowcollprob}. We will proceed by decomposing $f = f_{B_{\mathrm{safe}}} + f_{B_{\mathrm{coll}}} + f_{B_{=0}}$ where $f_B$ is supported on $B \subseteq \{\pm1\}^{n^{D'/D}k}$.
\begin{align*}
    \abs{\ip{f}{(T_{\mcG_L}T_{\mcG_C}T_{\mcG_L} - T_\mcG) f}} \leq &\abs{\ip{f_{B_{\mathrm{safe}}}}{(T_{\mcG_L}T_{\mcG_C}T_{\mcG_L} - T_\mcG) f_{\sD_{n^{D'/D}}}}}
    + \abs{\ip{f_{B_{\mathrm{coll}}}}{(T_{\mcG_L}T_{\mcG_C}T_{\mcG_L} - T_\mcG) f_{\sD_{n^{D'/D}}}}}\\
    &\;\;+ \abs{\ip{f_{B_{=0}}}{(T_{\mcG_L}T_{\mcG_C}T_{\mcG_L} - T_\mcG) f_{B_{=0}}}}.
\end{align*}

\subsubsection{$f$ Supported on $B_{\mathrm{safe}}$}

\begin{lemma}
    $\abs{\ip{f_{B_{\mathrm{safe}}}}{(T_{\mcG_L}T_{\mcG_C}T_{\mcG_L} - T_\mcG) f_{\sD}}} \leq \frac{4n^{1/D}k^2}{2^{n^{(D'-1)/D}}} \cdot \ip{f}{f}$.
\end{lemma}

\begin{proof}
Let $X \in B_{\mathrm{safe}}$, $g : \{\pm1\}^{n^{D'/D}k} \to \R$.
\begin{align*}
    (T_{\mcG_L} - T_\mcG)g(X) &= \sum_{Y \in \{\pm1\}^{n^{D'/D}k}} \Pr[X \to_{T_{\mcG_L}} Y] \cdot g(Y) - \sum_{Y \in \{\pm1\}^{n^{D'/D}k}} \Pr[X \to_{T_{\mcG}} Y] \cdot g(Y)\\
    &= \frac{1}{\abs{B_{\mathrm{safe}}}}\sum_{Y \in B_{\mathrm{safe}}} g(Y) - \frac{1}{\abs{\sD_{n^{D'/D}}}}\sum_{Y \in \sD_{n^{D'/D}}} g(Y)\\
    &= \pbra{\frac{1}{\abs{B_{\mathrm{safe}}}} - \frac{1}{\abs{\sD_{n^{D'/D}}}}}\sum_{Y \in B_{\mathrm{safe}}} g(Y) - \frac{1}{\abs{\sD_{n^{D'/D}}}}\sum_{Y \in B_{\mathrm{coll}}} g(Y)\\
    &= \pbra{1 - \frac{\abs{B_{\mathrm{safe}}}}{\abs{\sD_{n^{D'/D}}}}} \cdot \frac{1}{\abs{B_{\mathrm{safe}}}}\sum_{Y \in B_{\mathrm{safe}}} g(Y) - \frac{\abs{B_{\mathrm{coll}}}}{\abs{\sD_{n^{D'/D}}}} \cdot \frac{1}{\abs{B_{\mathrm{coll}}}}\sum_{Y \in B_{\mathrm{coll}}} g(Y)\\
    &= \frac{\abs{B_{\mathrm{coll}}}}{\abs{\sD_{n^{D'/D}}}} \pbra{T_{\mcG_L} - \mcH}g(X),
\end{align*}
where $\mcH g(X) = \frac{1}{\abs{B_{\mathrm{coll}}}}\sum_{Y \in B_{\mathrm{coll}}} g(Y)$. We write:
\begin{align*}
    \abs{\ip{f_{B_{\mathrm{safe}}}}{(T_{\mcG_L}T_{\mcG_C}T_{\mcG_L} - T_\mcG) f_{\sD_{n^{D'/D}}}}} &= \sum_{X \in \{\pm1\}^{n^{D'/D}k}} f_{B_{\mathrm{safe}}}(X) \cdot (T_{\mcG_L}-T_\mcG)(T_{\mcG_C}T_{\mcG_L}f_{\sD_{n^{D'/D}}})(X)\\
    &= \sum_{X \in B_{\mathrm{safe}}} f_{B_{\mathrm{safe}}}(X) \cdot (T_{\mcG_L}-T_\mcG)(T_{\mcG_C}T_{\mcG_L}f_{\sD_{n^{D'/D}}})(X)\\
    &= \frac{\abs{B_{\mathrm{coll}}}}{\abs{\sD_{n^{D'/D}}}} \sum_{X \in B_{\mathrm{safe}}} f_{B_{\mathrm{safe}}}(X) \cdot (T_{\mcG_L}-\mcH)(T_{\mcG_C}T_{\mcG_L}f_{\sD_{n^{D'/D}}})(X)\\
    &= \frac{\abs{B_{\mathrm{coll}}}}{\abs{\sD_{n^{D'/D}}}} \abs{\ip{f_{B_{\mathrm{safe}}}}{(T_{\mcG_L}T_{\mcG_C}T_{\mcG_L} - \mcH T_{\mcG_C}T_{\mcG_L}) f_{\sD_{n^{D'/D}}}}}.
\end{align*}
The fact that $\mcH$ is a random walk operator once again establishes:
\begin{align*}
    &\abs{\ip{f_{B_{\mathrm{safe}}}}{(T_{\mcG_L}T_{\mcG_C}T_{\mcG_L} - T_\mcG) f_{\sD_{n^{D'/D}}} }} \\
    \leq &\frac{\abs{B_{\mathrm{coll}}}}{\abs{\sD_{n^{D'/D}}}} \pbra{\abs{\ip{f_{B_{\mathrm{safe}}}}{T_{\mcG_L}T_{\mcG_C}T_{\mcG_L} f_{\sD_{n^{D'/D}}}}} + \abs{\ip{f_{B_{\mathrm{safe}}}}{\mcH T_{\mcG_C}T_{\mcG_L} f_{\sD_{n^{D'/D}}}}}}\\
    \leq& \frac{2\abs{B_{\mathrm{coll}}}}{\abs{\sD_{n^{D'/D}}}} \norm{f_{B_{\mathrm{safe}}}}_{2} \norm{f_{\sD_{n^{D'/D}}}}_{2} \tag{\Cref{lem:escape probs}}\\
    \leq& \frac{2\abs{B_{\mathrm{coll}}}}{\abs{\sD_{n^{D'/D}}}} \langle f, f \rangle.
\end{align*}
\Cref{fact:gencolorclasssizes} then suffices to prove the claim. \end{proof}

\subsubsection{$f$ Supported on $B_{\mathrm{coll}}$}

\begin{lemma}
    $\abs{\ip{f_{B_{\mathrm{coll}}}}{(T_{\mcG_L}T_{\mcG_C}T_{\mcG_L} - T_\mcG) f_{\sD_{n^{(D'-1)/D}}}}} \leq \frac{8n^{1/D}k^2}{2^{n^{D'/D}/32}} \ip{f}{f}$.
\end{lemma}

\begin{proof}
First, we can decompose $f_{\sD_{n^{D'/D}}} = f_{B_{\mathrm{safe}}} + f_{B_{\mathrm{coll}}}$ and bound:
\begin{align*}
    \abs{\ip{f_{B_{\mathrm{coll}}}}{(T_{\mcG_L}T_{\mcG_C}T_{\mcG_L} - T_\mcG) f_{\sD_{n^{D'/D}}}}} 
    &\leq \abs{\ip{f_{B_{\mathrm{coll}}}}{(T_{\mcG_L}T_{\mcG_C}T_{\mcG_L} - T_\mcG) f_{B_{\mathrm{safe}}}}} \\
    &+ \abs{\ip{f_{B_{\mathrm{coll}}}}{(T_{\mcG_L}T_{\mcG_C}T_{\mcG_L} - T_\mcG) f_{B_{\mathrm{coll}}}}}.
\end{align*}
By the self-adjointness of the operator, the first term is bounded by the case above, so it suffices to bound the latter. For this term, we can appeal directly to \Cref{lem:escape probs} and the triangle inequality to get:
\begin{align*}
    &\abs{\ip{f_{B_{\mathrm{coll}}}}{(T_{\mcG_L}T_{\mcG_C}T_{\mcG_L} - T_\mcG) f_{B_{\mathrm{coll}}}}} \\
    &\leq\sqrt{\max_{X \in B_{\mathrm{coll}}} \Pr[X \to_{T_{\mcG_L}T_{\mcG_C}T_{\mcG_L}} B_{\mathrm{coll}}] + \max_{X \in B_{\mathrm{coll}}} \Pr[X \to_{T_\mcG} B_{\mathrm{coll}}]} \ip{f_{B_{\mathrm{coll}}}}{f_{B_{\mathrm{coll}}}}.
\end{align*}
Note that regardless of choice of $X$, the latter probability $\Pr[X \to_{T_\mcG} B_{\mathrm{coll}}] = \frac{\abs{B_{\mathrm{coll}}}}{\abs{\sD_{n^{D'/D}}}} \leq \frac{2n^{1/D}k^2}{2^{n^{(D'-1)/D}}}$ by \Cref{fact:gencolorclasssizes}. We finish by proving \Cref{lem:genlowcollprob} from the previous section below.
\end{proof}
\begin{lemma}[Restatement of \Cref{lem:genlowcollprob}]
    For all $X \in \sD_{n^{D'/D}}$, $\Pr[X \to_{T_{\mcG_C}T_{\mcG_L}} B_{\mathrm{coll}}] \leq \frac{2n^{1/D}k^2}{2^{n^{(D'-1)/D}/16}}$.
\end{lemma}
\begin{proof}
    Our goal is to union bound over the probability of any pair of sublattices colliding. There are at most $n^{1/D}k^2$ pairs of sublattices. 
    
    Let $X \in \sD_{n^{D'/D}}$. We will model our process as:
    \begin{equation*}
        X \to_{T_{\mcG_L}} Y \to_{T_{\mcG_C}} Z.
    \end{equation*}
    We then fix $Z^\ell_{i, \cdot}$ and $Z^m_{i, \cdot}$ (which recall are $(D'-1)$-dimensional slices, in $\{\pm1\}^{n^{(D'-1)/D}}$) for $i \in [n^{1/D}]$, $\ell \neq m \in [k]$. We use that for some $j \in [k]$, we have $(Y^\ell_{j, \cdot}, Y^m_{j, \cdot})$ are uniform from ${\{\pm1\}^{n^{(D'-1)/D}} \choose 2}$. To see this note that there must exist some $j$ s.t. $X^\ell_{j, \cdot} \neq X^m_{j,\cdot}$, otherwise $X \notin \sD_{n^{D'/D}}$. Since the permutation applied to these two grids is uniform from $\mfS{\{\pm1\}^{n^{(D'-1)/D}}}$, the resulting rows in $Y$ look like a uniform distinct pair.
    
    With this in mind, we will now condition on the event that $d(Y^\ell_{j, \cdot}, Y^m_{j, \cdot}) \geq n^{(D'-1)/D}/4$ and compute for $n$ large enough that
    \begin{align*}
        \Pr[Z^\ell_{i, \cdot} = Z^m_{i, \cdot}] =& \Pr[Z^\ell_{i, \cdot} = Z^m_{i, \cdot} \mid d(Y^\ell_{j, \cdot}, Y^m_{j, \cdot}) > n^{(D'-1)/D}/4] \\
        +& \Pr[Z^\ell_{i, \cdot} = Z^m_{i, \cdot} \mid d(Y^\ell_{j, \cdot}, Y^m_{j, \cdot}) \leq n^{(D'-1)/D}/4]\Pr[d(Y^\ell_{j, \cdot}, Y^m_{j, \cdot}) \leq n^{(D'-1)/D}/4]\\
        \leq&\Pr[Z^\ell_{i, \cdot} = Z^m_{i, \cdot} \mid d(Y^\ell_{j, \cdot}, Y^m_{j, \cdot}) > n^{(D'-1)/D}/4] +\Pr[d(Y^\ell_{j, \cdot}, Y^m_{j, \cdot}) \leq n^{(D'-1)/D}/4]\\
        \leq&\frac{1}{2^{n^{(D'-1)/D}/16}}+\frac1{e^{n^{(D'-1)/D}/16}}\tag{\Cref{lem:given Y far}, \Cref{lem:Y probably far}}\\
        \leq&\frac{1}{2^{n^{(D'-1)/D}/32}}.
    \end{align*}
    Applying a union bound over all $n^{1/D}k^2$ pairs of sublattices completes the proof.
\end{proof}

\begin{lemma}\label{lem:given Y far}
    $\Pr[Z^\ell_{i, \cdot} = Z^m_{i, \cdot} \mid d(Y^\ell_{j, \cdot}, Y^m_{j, \cdot}) > n^{(D'-1)/D}/4] \leq \frac{1}{2^{n^{(D'-1)/D}/4}}$
\end{lemma}
\begin{proof}
    The probability that $Z^\ell_{i, \cdot}$ and $Z^m_{i, \cdot}$ are equal can be viewed as the probability that all of their individual bits are equal, and they are all independent since they come from independently sampled rows. Since $Y^\ell_{j, \cdot}$ and $Y^m_{j, \cdot}$ differ in at least $n^{(D'-1)/D}/4$ places, $Y^\ell$ and $Y^m$ must differ in at least that many rows. In these rows, it can be seen that the corresponding bits in $Z^\ell_{i, \cdot}$ and $Z^m_{i, \cdot}$ are the same with probability $\leq \frac{1}{2}$. By independence the probability is less than $\frac{1}{2^{n^{(D'-1)/D}/4}}$. 
\end{proof}

\begin{lemma}\label{lem:Y probably far}
    $\Pr[d(Y^\ell_{j, \cdot}, Y^m_{j, \cdot}) \leq n^{(D'-1)/D}/4] \leq \frac{1}{e^{n^{(D'-1)/D}/16}}$
\end{lemma}
\begin{proof}
    This can be seen by a simple Chernoff bound. Note that $\Pr[d(Y^\ell_{j, \cdot}, Y^m_{j, \cdot}) \leq n^{(D'-1)/D}/4] \leq \Pr_{x, y \sim \{\pm1\}^{n^{D'/D}}}[d(x, y) \leq n^{(D'-1)/D}/4]$, as if they are equal the distance is minimized. For uniform $x,y$, $d(x,y)$ can be seen as the sum of $n^{(D'-1)/D}$ independent Bernoulli$(1/2)$ r.v.s. By Hoeffding's Inequality:
    \begin{equation*}
        \Pr_{x,y \sim \{\pm1\}^{n^{(D'-1)/D}}}[d(x,y) \leq n^{(D'-1)/D}/4] \leq e^{-n^{(D'-1)/D}/16}.\qedhere
    \end{equation*}
\end{proof}

\subsubsection{The Induction Case}

\begin{lemma}
Let $f : \{\pm1\}^{n^{D'/D}k} \to \R$ be supported on $B_{=0}$ and $k \geq 2$. Then, we have
\[ 
\abs{\ip{f}{\pbra{T_{\mcG_L}^{(k)}T_{\mcG_C}^{(k)}T_{\mcG_L}^{(k)} - T_\mcG^{(k)}}f}} \le \norm{T_{\mcG_L}^{(k-1)}T_{\mcG_C}^{(k-1)}T_{\mcG_L}^{(k-1)} - T_{\mcG}^{(k-1)}}_{2} \ip{f}{f}. 
\]
\end{lemma}

\begin{proof}
    The proof is nearly notationally identical to \Cref{lemma:f supported on B0} as the notion of color class developed in that section is on the tuple so is not dependent on the choice of sublattice, so we will refer back for brevity.
\end{proof}

\end{document}